\documentclass{article}
\usepackage[utf8]{inputenc}

\usepackage[backend=biber, style=numeric, maxbibnames=99]{biblatex}
\usepackage{amsmath}
\usepackage{amsfonts}
\usepackage{mathrsfs}
\usepackage{amssymb}
\usepackage{amsthm}
\usepackage{thmtools}
\usepackage{mdframed}
\usepackage{bm}
\usepackage{mathtools}
\usepackage{mathdots}
\usepackage{multirow}
\usepackage{xcolor}
\usepackage[most]{tcolorbox}
\usepackage{tikz}
\usepackage{tikz-cd}
\usetikzlibrary{patterns.meta, positioning, arrows, decorations.pathmorphing, tikzmark}

\usepackage{pgfplots}
\pgfplotsset{compat = newest,
                        layers/my layer set/.define layer set={
            background,
            main,
            foreground}{ }, set layers=my layer set,}

\usepackage{enumitem}
\usepackage[normalem]{ulem}
\usepackage{float}
\usepackage[breaklinks=true,colorlinks=true,
linkcolor=blue,urlcolor=blue,citecolor=blue]{hyperref}
\usepackage{caption}
\usepackage{subcaption}
\usepackage[a4paper, hmargin = 20mm, vmargin = 20mm]{geometry}
\usepackage{tocbibind}

\colorlet{mycolor0}{black}
\colorlet{mycolor1}{blue}
\colorlet{mycolor2}{red}
\colorlet{mycolor3}{blue!20!green!80!}
\colorlet{mycolor4}{red!20!yellow!80!}
\colorlet{mycolor5}{cyan!10!}
\colorlet{mycolor6}{yellow!10!}
\colorlet{mycolor7}{magenta!10!}
\colorlet{mycolor8}{green!10!}

\tcbset{enhanced jigsaw, breakable,
        width=1\textwidth, boxrule=0pt, arc=0pt,
        auto outer arc, left=0pt, right=0pt, boxsep=5pt,
        frame hidden}

\declaretheoremstyle[
  qed=\qedsymbol]{thmwithend}

\mdfdefinestyle{theoremstyle}{backgroundcolor=mycolor7,
    topline=false, leftline=false, bottomline=false, rightline=false}
\mdfdefinestyle{definitionstyle}{backgroundcolor=mycolor6,
    topline=false, leftline=false, bottomline=false, rightline=false}
\mdfdefinestyle{examplestyle}{backgroundcolor=mycolor5,
    topline=false, leftline=false, bottomline=false, rightline=false}
\mdfdefinestyle{conjecturestyle}{backgroundcolor=mycolor8,
    topline=false, leftline=false, bottomline=false, rightline=false}
\mdfdefinestyle{algorithmstyle}{backgroundcolor=mycolor8,
    topline=false, leftline=false, bottomline=false, rightline=false}

\declaretheoremstyle[
headfont = \normalfont\bfseries,
bodyfont = \normalfont]{theorem}

\declaretheoremstyle[
headfont = \normalfont\bfseries,
bodyfont = \normalfont\itshape]{italictheorem}

\declaretheoremstyle[
headfont = \normalfont\bfseries,
bodyfont = \normalfont]{plain}

\declaretheorem[
  style=italictheorem,
    mdframed={style=theoremstyle},
  title=Theorem,
  refname={Theorem, Theorems},
  Refname={Theorem, Theorems},
  within=section
]{theorem}

\declaretheorem[
  style=plain,
    mdframed={style=definitionstyle},
  title=Definition,
  refname={Definition, Definitions},
  Refname={Definition, Definitions},
  sibling=theorem
]{definition}

\declaretheorem[
  style=italictheorem,
    mdframed={style=theoremstyle},
  title=Lemma,
  refname={Lemma, Lemmas},
  Refname={Lemma, Lemmas},
  sibling=definition
]{lemma}

\declaretheorem[
  style=thmwithend,
  mdframed={style=definitionstyle},
  title=Remark,
  refname={Remark, Remarks},
  Refname={Remark, Remarks},
  sibling=lemma
]{remark}

\declaretheorem[
  style=italictheorem,
  mdframed={style=definitionstyle},
  title=Corollary,
  refname={Corollary, Corollaries},
  Refname={Corollary, Corollaries},
  sibling=lemma
]{corollary}

\declaretheorem[
  style=remark,
    mdframed={style=examplestyle},
  title=Example,
  refname={Example, Examples},
  Refname={Example, Examples},
  sibling=remark
]{example}

\declaretheorem[
  style=italictheorem,
  title=Question,
  refname={Question, Questions},
  Refname={Question, Questions},
  sibling=example
]{question}

\counterwithin{table}{section}

\usepackage[ruled, lined, linesnumbered, commentsnumbered, longend, algosection]{algorithm2e}
\SetArgSty{textrm}
\SetKwInOut{KwArguments}{Arguments}
\SetKwInOut{KwReturns}{Returns}
\SetKwComment{Comment}{\#\ }{}

\SetCommentSty{mycommfont}

\usepackage[capitalise]{cleveref}
\crefname{algocf}{Algorithm}{Algorithms}
\crefname{algocf}{Algorithm}{Algorithms}

\usepackage[parfill]{parskip}
\makeatletter
\let\c@algocf=\c@theorem \makeatother

\makeatletter
\newcommand{\oast}{\mathbin{\mathpalette\make@circled\ast}}
\newcommand{\ocirc}{\mathbin{\mathpalette\make@circled\circ}}
\newcommand{\make@circled}[2]{  \ooalign{$\m@th#1\smallbigcirc{#1}$\cr\hidewidth$\m@th#1#2$\hidewidth\cr}}
\newcommand{\smallbigcirc}[1]{  \vcenter{\hbox{\scalebox{0.77778}{$\m@th#1\bigcirc$}}}}
\makeatother

\DeclareMathOperator{\cont}{\mathtt{cont}}
\DeclareMathOperator{\pref}{\mathtt{pref}}
\DeclareMathOperator{\suff}{\mathtt{suff}}
\DeclareMathOperator{\ltof}{\mathtt{ltof}}
\DeclareMathOperator{\ftol}{\mathtt{ftol}}

\newcommand{\bool}{\{0,1\}}
\DeclareMathOperator{\MinWord}{\mathtt{MinWord}}
\DeclareMathOperator{\Multiply}{\mathtt{Multiply}}
\DeclareMathOperator{\EqualInFreeBand}{\mathtt{EqualInFreeBand}}
\DeclareMathOperator{\ComputeK}{\mathtt{ComputeK}}
\DeclareMathOperator{\TrimTransducerIsomorphism}{\mathtt{TrimTransducerIsomorphism}}
\DeclareMathOperator{\Minimize}{\mathtt{Minimize}}
\DeclareMathOperator{\IntervalTransducer}{\mathtt{IntervalTransducer}}
\DeclareMathOperator{\ClassifyCase}{\mathtt{ClassifyCase}}

\renewcommand{\spadesuit}{\ocirc}
\renewcommand{\clubsuit}{\oast}

\newcommand{\FB}{\textrm{FB}^1}

\DeclareMathOperator{\rght}{\mathtt{RIGHT}}
\DeclareMathOperator{\lft}{\mathtt{LEFT}}

\DeclareMathOperator{\dom}{dom}
\newcommand{\defn}[1]{\textit{\textbf{#1}}}

\pgfdeclaredecoration{sl}{initial}{
 \state{initial}[width=\pgfdecoratedpathlength-1sp]{
  \pgfmoveto{\pgfpointorigin}
 }
 \state{final}{
  \pgflineto{\pgfpointorigin}
 }
}

\definecolor{color0}{RGB}{238,20,135}
\definecolor{color1}{RGB}{0,221,164}
\definecolor{color2}{RGB}{86,151,209}
\definecolor{color3}{RGB}{249,185,131}
\definecolor{color4}{RGB}{150,114,196}

\tikzset{parallel_arrow_1/.style={-latex, decoration={sl,raise=-1mm},decorate}}
\tikzset{parallel_arrow_2/.style={-latex reversed, dashed, decoration={sl,raise=1mm},decorate}}

\title{Polynomial time multiplication and normal forms in free bands}
\author{R. Cirpons and J. D. Mitchell}
\date{}

\begin{document}

\maketitle

\begin{abstract}
We present efficient computational solutions to the problems of checking equality, performing multiplication, and computing minimal representatives of elements of free bands. A band is any semigroup satisfying the identity $x ^ 2 \approx x$ and the \defn{free band} $\operatorname{FB}(k)$ is the free object in the variety of $k$-generated bands. Radoszewski and Rytter developed a linear time algorithm for checking whether two words represent the same element of a free band.
In this paper we describe an alternate linear time algorithm for the same problem. The algorithm we present utilises a representation of words as synchronous deterministic transducers that lend themselves to efficient (quadratic in the size of the alphabet) multiplication in the free band. This representation also provides a means of finding the short-lex least word representing a given free band element with quadratic complexity. \end{abstract}

\section{Introduction}\label{sect:intro}

This paper is about efficient computational solutions to the problems of checking equality, performing multiplication, and computing minimal representatives of elements of free bands.
This paper arose from that of Radoszewski
and Rytter~\cite{radoszewski2010aa}, and is related to the papers of Neto and Sezinando~\cite{Neto2000aa,Neto1996aa}. We will discuss in more detail below the interactions of \cite{Neto2000aa, Neto1996aa, radoszewski2010aa} and the current paper.

A \defn{band} is any semigroup satisfying the identity $x ^ 2 \approx x$ and a \defn{free band} $\operatorname{FB}(k)$ of rank $k$ is the free object in the variety of $k$-generated bands.
Bands have been widely studied in the literature, some highlights include:
 the lattice of varieties of bands is completely classified~\cite{Biryukov1970, Fennemore1970, Gerhard1989, Gerhard1970} and every such variety is defined by a single identity. Siekmann and Szab\'o~\cite{Siekmann1982aa} provide an infinite complete rewriting system for every free band; rewriting in further varieties of bands has also been considered in~\cite{Klima2011aa}.
 More recently, combinatorial structures such as hyperplane arrangements, interval greedoids, and matroids have been endowed with the structure of left regular bands
 (i.e. those bands belonging to the variety defined by the identity $xy \approx xyx$), and
this connection has been used in several disparate areas of mathematics; see, for example~\cite{Margolis2021}  and the references therein. Other recent references related to bands include~\cite{Dolinka2013, QuinnGregson2018, Saliola2007}.

 In the context of the famous Burnside problem for semigroups, Green and Rees~\cite{green1952aa} showed that if every $k$-generated group, where $k\in \mathbb{N}$, satisfying the identity $x ^ {m - 1} \approx 1$ is finite, then every $k$-generated semigroup satisfying the identity $x ^ m \approx x$ is finite. Brown~\cite{Brown1964} gave a simple proof of Green and Rees' result.
In the case that $m = 2$, every $k$-generated group satisfying the identity $x ^ 1 \approx 1$ is trivial, by virtue of which, every $k$-generated semigroup satisfying the identity $x ^ 2\approx x$ is finite. In particular, every finitely generated band is finite, including every free band $\operatorname{FB}(k)$.
The size of $\operatorname{FB}(k)$ grows super-exponentially with $k$ the number of generators; a closed form for $|\operatorname{FB}(k)|$ is given in \cite[Theorem 4.5.3]{howie1995aa}, see also \cite{A030449}. For example, the free band $\operatorname{FB}(5)$ has size $2,751,884,514,765 \approx 2 ^ {41}$.

 The ability to compute efficiently with elements in a free band is a fundamental prerequisite for any further computational tools for the study of bands; such as determining the least variety of bands containing a given band, or computing with finite band presentations. However,
the vast number of elements renders it impossible to practically apply to free bands, with $5$ or more generators, any of the known general purpose algorithms for finite semigroups such as, for example, the Froidure-Pin Algorithm~\cite{Froidure1997aa} or the algorithms described in~\cite{East2019aa}. When it is possible to readily multiply and check equality of elements of a finitely generated semigroup $S$,
the Froidure-Pin Algorithm~\cite{Froidure1997aa} can be used to exhaustively enumerate $S$. Even assuming that the question of how to multiply and check the equality of elements in the free band is resolved, the Froidure-Pin Algorithm~\cite{Froidure1997aa} requires too much space to compute $\operatorname{FB}(k)$ with $k\geq 5$.  If it were possible (and it is not) to store every one of the approximately $2 ^ {41}$ elements of $\operatorname{FB}(5)$ in a single bit, this would require more than 274GB of RAM, a figure that is unlikely to be commonplace any time in the near future. The algorithms in~\cite{East2019aa} are most effective for semigroups containing large subgroups, but alas every subgroup of a band is trivial, so the algorithms from \cite{East2019aa} are of little practical use for bands.

Semigroup presentations provide a possible alternative approach to computing with free bands, for example, using the Todd-Coxeter~\cite{Todd1936aa}, Knuth-Bendix~\cite{Knuth1983} or other rewriting methods.
If $A$ is any non-empty alphabet, then the collection $A ^ +$ of all finite non-empty words over $A$ forms a semigroup where the operation is juxtaposition; $A ^ +$ is called the \defn{free semigroup} on $A$.
If $w\in A ^ +$, then the length of $w$ is denoted by $|w|$.
We denote the free band generated by $A$ by $\operatorname{FB}(A)$ and henceforth will not use the notation $\operatorname{FB}(k)$ where $k\in\mathbb{N}$.
One, not particularly illuminating, description of the free band $\operatorname{FB}(A)$ on $A$ is the quotient of $A ^ +$ by the least congruence $\sim$ containing $(w ^ 2, w)$ for all $w\in A ^ +$.  In other words,
$\operatorname{FB}(A)$ has the infinite presentation
\begin{equation}\label{eq-infinite-presentation}
    \langle A\, |\, x ^ 2 = x, \,x\in A ^ +\rangle.
\end{equation}
Of course, since $\operatorname{FB}(A)$  is finite, it is finitely presented, although it is not immediately clear how to find a finite presentation when $|A| \geq 5$.
Siekmann and Szab\'o~\cite{Siekmann1982aa} show that the following length reducing infinite conditional rewriting system for $\operatorname{FB}(A)$ is complete:
\begin{align}\label{align-1}
    x^2&\rightarrow x\\
    xyz &\rightarrow xz \textrm{ if }\cont(y) \subseteq \cont(x) = \cont(z)
\end{align}
where $x, y, z$ range over all appropriate words in $A ^ +$ and
where $\cont(w) = \{b_1, b_2, \ldots, b_n\}$ is the \defn{content} of the word $w = b_1b_2\cdots b_n\in A ^{+}$. A simpler proof that the rewriting system of
Siekmann and Szab\'o~\cite{Siekmann1982aa} is complete was given by Neto and Sezinando in \cite[Theorem 6.1]{Neto2000aa}.
Although the papers~\cite{Neto2000aa, Neto1996aa, Siekmann1982aa} make a valuable contribution to the study of bands, their focus is mainly on the mathematical aspects.
In particular, it is claimed in \cite{Neto2000aa} that there is a polynomial-time solution for the word problem in all relatively free bands other than $\operatorname{FB}(A)$, and that Siekmann and Szab\'o provide a quadratic time algorithm for solving the word problem in the free band in \cite{Siekmann1982aa}. However, the algorithm is not explicitly given in ~\cite{Neto2000aa}, or elsewhere in the literature, nor is its time complexity formally analysed.
Since Siekmann and Szab\'o's rewriting system is complete and length reducing, it can be used to solve the word problem by rewriting a given word $w$ in $\mathcal{O}(|w|)$ steps.
However, it is unclear what the time complexity of each rewriting step is, since the rewriting system is infinite, and so it is non-trivial to detect which rewrite rules apply to any given word $w$. This question is not addressed in \cite{Siekmann1982aa}.
It is possible to detect all subwords $x^2$ where $x$ is primitive (not of the form $x=y^e$ for some $y\in A^+$, $e\geq 2$) of a given $w$ in $\mathcal{O}(|w| \log|w|)$ time; see~\cite{Crochemore1981aa}. It is unclear how to efficiently find occurrences of $xyz$ with $\cont(y) \subseteq \cont(x) = \cont(z)$ in $w$. One approach might be to iterate through the subwords of $w$, to find a single subword of the form $xyz$ where $\cont(y) \subseteq \cont(x) = \cont(z)$. The time complexity of this approach is, at best, $\mathcal{O}(|w|^2)$ yielding an algorithm with time complexity at best $\mathcal{O}(|w|^3)$ for finding a normal form for $w$. Hence for words $u, v\in A ^ *$ solving the word problem using this approach has complexity at best $\mathcal{O}(|u|^3 + |v| ^ 3)$.

Rather than using the infinite rewriting system given in \eqref{align-1}, it is possible to compute the unique finite complete rewriting system for $\operatorname{FB}(A)$ when $|A| < 5$ with respect to the short-lex reduction ordering on $A ^ +$ obtained from any linear order of $A$ (using, for example, the Froidure-Pin Algorithm~\cite{Froidure1997aa}, as implemented in \cite{Mitchell2022aa}); see \cref{table-1}.
The Froidure-Pin Algorithm~\cite{Froidure1997aa} stores every element of the semigroup it is enumerating, and so, as discussed above, it was not possible to compute a finite complete rewriting system for $\operatorname{FB}(A)$ when $|A| \geq 5$. For reference, the time complexity of the Froidure-Pin Algorithm is $\mathcal{O}(|A| \cdot |\operatorname{FB}(A)|)$.

\begin{table}
    \centering
    \begin{tabular}{|c|r|r|r|}
    \hline
         $|A|$ & $\left|\operatorname{FB}(A)\right|$ & Number of relations & Total length of relations \\

         \hline
1 & 1             & 1     & 1      \\
2 & 6             & 4     & 20     \\
3 & 159           & 45    & 468    \\
4 & 332,380        & 11,080 & 217,072 \\
5 & 2,751,884,514,765 & ?     & ?      \\
\hline
    \end{tabular}
    \caption{Number of relations and total length of relations in the unique finite complete rewriting system of $\operatorname{FB}(A)$ for some small values of $|A|$.}
    \label{table-1}
\end{table}

With the preceding comments in mind, it seems unlikely that presentations or rewriting systems would allow for practical computation of the word problem, or normal forms,  when $|A| \geq 5$.
On the other hand, Radoszewski and Rytter's algorithm from~\cite{radoszewski2010aa} (as implemented in \cite{Mitchell2022ab}, and available in \cite{Mitchell2022aa} and \cite{Mitchell2022ad}) can test equality of words in $\operatorname{FB}(A)$ when $|A| \geq 100$ with relative ease.

In this paper, we will show how to use a representation of elements of a free band as transducers in linear time algorithms for equality checking and the computation of normal forms, and an algorithm for multiplication that is quadratic in the size of the alphabet and linear in the size of the representations of the elements; see below for a formal description of the complexity. This representation is similar to the ``admissible maps'' of~\cite{Neto2000aa,Neto1996aa} and the ``factor automata'' of~\cite{radoszewski2010aa}.

The starting point for the development of the representation of free band elements by transducers (and the other representations in \cite{Neto2000aa,Neto1996aa,radoszewski2010aa}), is the solution to the word problem  given by Green and Rees in~\cite{green1952aa}; we will restate this description in \cref{sec:prelims}.
 Recall that we denoted by $\sim$ the least congruence on $A ^ +$ containing $\{(x ^ 2, x): x\in A ^ +\}$. We refer to the problem of determining whether or not $u\sim v$  for any $u, v\in A ^ +$ as \defn{checking equality in the free band}.
 If $w\in A ^ +$, then we define $\pref(w)$ to be the longest prefix of $w$
 containing one fewer unique letters than $w$ itself (i.e.
 $\cont(\pref(w)) = \cont(w) - 1$) and $\suff(w)$ is defined analogously for suffixes.
 We denote the letter immediately after $\pref(w)$ in $w$ by $\ltof(w)$ (from ``last letter to occur first''), which is the unique letter in the content of $w$ not in the content of $\pref(w)$.
 The first letter $\ftol(w)\in A$ to occur last in $w$ is defined analogously with respect to $\suff(w)$.
 Green and Rees~\cite{green1952aa} showed that $u\sim v$ if and only if $\pref(u)\sim\pref(v)$, $\ltof(u) = \ltof(v)$, $\ftol(u)=\ftol(v)$, and $\suff(u) \sim \suff(v)$. This is illustrated in \cref{fig:WEDIDIT-2022-08-18-16-28}.

\begin{figure}
\[
  w = ababdbddcccb = \rlap{$\overbrace{\phantom{ababdbdd}}^{\pref(w)}$}ab\tikzmark{a}{\color{blue}a}
  \underbrace{bdbdd\,\tikzmark{c}{\color{red}c}ccb}_{\suff(w)}
\]

\begin{tikzpicture}[remember picture,overlay]
\draw[<-, color=blue]
  ([shift={(3pt,-2pt)}]pic cs:a) |- ([shift={(-10pt,-10pt)}]pic cs:a)
  node[anchor=east] {\scriptsize$\ftol(w)$};
\draw[<-, color=red]
  ([shift={(3pt,6pt)}]pic cs:c) |- ([shift={(16pt,15pt)}]pic cs:c)
  node[anchor=west] {\scriptsize$\ltof(w)$};
\end{tikzpicture}
\caption{An illustration of the functions $\pref$, $\suff$, $\ftol$, and $\ltof$ applied to the word $w = ababdbddcccb$.}
\label{fig:WEDIDIT-2022-08-18-16-28}
\end{figure}

It is trivial to check whether or not two letters in $A$ are equal. If $\ltof(u) = \ltof(v)$ and $\ftol(u)=\ftol(v)$, then
the Green-Rees Theorem reduces the problem of checking $u\sim v$ for $u, v\in A ^ +$ to the problem of checking $\pref(u)\sim\pref(v)$ and $\suff(u)\sim\suff(v)$ where $\pref(u), \pref(v)\in (A \setminus \{\ltof(u) = \ltof(v)\}) ^ +$ and $\suff(u), \suff(v)\in (A \setminus \{\ftol(u) = \ftol(v)\}) ^ +$. In particular, the size of the underlying alphabet is reduced by one, and therefore, by repeatedly applying the Green-Rees Theorem, we eventually reduce the size of the alphabet to one.
Performing this na\"ively requires $\mathcal{O}((|u|+|v|)\cdot 2 ^ {|A|})$ time and $|u| + |v|$ space.
A similar recursive approach was found by Gerhard and Petrich~\cite{Gerhard1989} to solve the word problem in relatively free bands (the free objects in sub-varieties of bands).
Suppose that $A$ is any fixed alphabet and that $u, v\in A^+$ are such that $|u|+|v| = k$ for some $k\in \mathbb{N}$.
Neto and Sezinando~\cite{Neto2000aa} estimate the time complexity of Gerhard and Petrich's approach for the free band to be
$\mathcal{O}(2^k)$, which is worse than the na\"ive approach using the Green-Rees Theorem.
However, for other relatively free bands, Neto and Sezinando~\cite{Neto2000aa} estimate the time complexity of Gerhard and Petrich's approach to be
$\mathcal{O}((k+1)^{n-2})$ where $n\in\mathbb{N}$ is the height of the relatively free band in the lattice of all varieties of bands.

In contrast to the exponential time algorithms described in the previous paragraph, Radoszewski and Rytter~\cite{radoszewski2010aa} gave a $\mathcal{O}(|A|\cdot(|u|+|v|))$ time and $\mathcal{O}(|u|+|v|)$ space algorithm  for checking equality in the free band on $A$. Their method constructs an acyclic automaton, called the \defn{interval automaton}, for each element to be compared. In~\cite{radoszewski2010aa} it is shown that $u\sim v$ if and only if the corresponding interval automata recognize the same language. It is possible to check whether two acyclic automata recognize the same language in $\mathcal{O}(n)$ time where $n$ is the total number of states in both automata; see, for example,~\cite{revuz1992aa}. At least from the perspective of the asymptotic time complexity, it seems extremely unlikely that a better solution than that of Radoszewski and Rytter exists for checking whether or not $u\sim v$ where $u, v\in A^+$. At best we could hope to replace the $|A|$ factor by a lesser function of $|A|$. But $|u| + |v|$ can be significantly larger than $|A|$ and so the $|A|$ factor does not contribute significantly to the overall time complexity in such cases.

Radoszewski and Rytter's approach does not immediately give rise to an efficient representation of elements in the free band. In particular, the only representation of elements in the free band considered in~\cite{radoszewski2010aa} is by words in a free semigroup.
Representing elements by words has the advantage of a very simple and (time) efficient multiplication algorithm: simply concatenate the words $u, v\in A^+$ in constant time and $\mathcal{O}(|u|+|v|)$ space.
The main drawback of this representation is that the length of these words is unbounded.
Of course, given any word $w \in A ^+$ and any well-ordering $\prec$ on $A ^ +$, the equivalence class $w/{\sim}$ contains a $\prec$-minimum word, and since $\operatorname{FB}(A)$ is finite, this means that length of such minimum words is also bounded.
Therefore, any method for finding such minimum words would provide a bound for the space required to store an element represented as a word.
The rewriting system of~\cite{Siekmann1982aa} gives a way of finding the  $\prec$-minimum word with respect to the short-lex ordering. As discussed above, in practice, this only works for alphabets $A$ with $|A| \leq 4$.
Assuming this was not the case, Neto and Sezinando's claim in \cite{Neto1996aa} implies that minimizing a word $w\in A^+$ has time complexity at best $\mathcal{O}(|w|^2)$.

In this paper, we derive an alternative space efficient representation for an element in the free band by using synchronous deterministic acyclic transducers based on the interval automata of Radoszewski and Rytter~\cite{radoszewski2010aa} and related to the admissible maps of~\cite{Neto2000aa, Neto1996aa}. This representation lends itself to efficient solutions to the problems of checking equality and multiplication of elements of the free band, and to an algorithm for determining $\min(w)$ --- the short-lex minimum representative of the $\sim$-class of a given word $w$. In \cref{sec:prelims}, we establish some notation, and state the Green-Rees Theorem in terms of this notation. In \cref{sec:transducer}, we describe the novel transducer based representation of elements of the free band that is the central notion of this paper. We also describe how Radoszewski and Rytter's construction from~\cite{radoszewski2010aa} relates to these transducers. In \cref{sec:equality}, we describe
how to check equality of the elements in the free band represented by transducers in $\mathcal{O}(|A|\cdot(|u|+|v|))$ space and time. The procedure for checking equality is called
$\EqualInFreeBand$ and given in \cref{algorithm:equal_in_free_band}. Although the $\mathcal{O}(|A|\cdot(|u|+|v|))$ space representation by transducers we present is worse than the $\mathcal{O}(|u|+|v|)$ space complexity of Radoszewski and Rytter in~\cite{radoszewski2010aa}, we will show that the same representation can be used efficiently for multiplication and finding short-lex least words. In \cref{sec:multiplication}, we describe how to obtain a transducer representing the product of two free band elements $x$ and $y$ given transducers representing each of the elements. This algorithm is called $\Multiply$ and is given in \cref{algorithm:multiply}. $\Multiply$ has time and space complexity $\mathcal{O}(m + n + |A| ^ 2)$ where $m$ and $n$ are the number of states in any transducers representing $x$ and $y$. We will show that the number of states of the minimal transducer representing a free band element $w/{\sim}$ is $|\min(w)||A|$ where $\min(w)$ is the length of the short-lex least word  belonging to $w/{\sim}$.
Hence the best case time and space complexity of the multiplication algorithm $\Multiply$ in \cref{sec:multiplication} is
$\mathcal{O}(|\min(u)||A| + |\min(v)||A| + |A| ^ 2)$ and the worst case is $\mathcal{O}(|u||A| + |v||A| + |A| ^ 2)$. Finally, in \cref{sec:minimal}, we give the $\MinWord$ algorithm, \cref{algorithm:minimalword}, for computing the short-lex least word $\min(w)$ in $w/{\sim}$ given any transducer $\mathcal{T}$ representing $w/{\sim}$. This algorithm has time and space complexity $\mathcal{O}(n\cdot|A|)$ where $n$ is the number of states of $\mathcal{T}$.

We conclude the introduction by stating some open problems.
Clearly, any algorithm deciding whether or not words $u, v\in A ^ *$ represent the same element of a free band has time complexity bounded below by $\Omega(|u|+|v|)$. On the other hand, $\EqualInFreeBand$ given in \cref{algorithm:equal_in_free_band} has time complexity  $\mathcal{O}(|A|\cdot (|u|+|v|))$ time, so it is natural to ask if there is a way of reducing or removing the $|A|$ factor?
\begin{question}
What is the lower bound for an algorithm checking the equality of two elements in the free band?
 \end{question}

In a similar vein, any algorithm computing $\min(w)$ for $w\in A^\ast$ must take at least $\Omega(|w|)$ time to read the input, if nothing else. The algorithm $\MinWord$ \cref{algorithm:minimalword} (together with the construction of the interval transducer that is the input to $\MinWord$) solves this problem in $\mathcal{O}(|A|\cdot |w| +|A|^2\cdot |\min(w)|)$ time.
\begin{question}
What is the lower bound for an algorithm for finding $\min(w)$ given $w\in A ^ *$?
\end{question}

A \textit{variety} of bands is a collection of bands that is closed under taking subsemigroups, homomorphic images and arbitrary Cartesian products.
Given that the lattice of varieties of bands is completely determined (as mentioned above), we ask the following question.

\begin{question}
Are there analogues for any of the algorithms $\EqualInFreeBand$, $\Multiply$, and $\MinWord$ for the free object in other varieties of bands?
\end{question}

If such algorithms do exist, then they could serve as an entry point for the further development of computational tools for finitely presented bands.

\section{Preliminaries}\label{sec:prelims}

In this section, we introduce some notation and preliminary results from the literature that will be used throughout the remainder of this paper.

As mentioned in the previous section, the Green-Rees Theorem~\cite{green1952aa} characterizes elements of the free band in terms of the functions $\pref$, $\suff$, $\ltof$, and $\ftol$. This characterisation is one of the key components in this paper, but we require slightly different (and more general) notation for these functions, which we now introduce.

If $A$ is any non-empty alphabet, then the collection $A ^ *$ of all finite words over $A$ forms a monoid with juxtaposition being the operation. This monoid is called the \defn{free monoid} over $A$; the identity element of $A ^ *$ is the empty word $\varepsilon$; and $A ^ * = A ^ + \cup \{\varepsilon\}$.

If $X$ and $Y$ are sets, then a \defn{partial function} $f: X \to Y$ is just a function from some subset of $X$ to $Y$. If $f : X \to Y$ is any partial function, and $f$ is not defined at $x\in X$, then we will write $f(x) = \bot$ and we also assume that $f(\bot) = \bot$.

\begin{definition}\label{definition-circ}
    We define the partial function $\circ: A ^ * \times \bool ^ *\to A ^ *$ such that for all $w \in A ^ *$ the following hold:
    \begin{enumerate}[label=(\roman*)]
    \item
    $\varepsilon \circ \alpha = \bot$ for all $\alpha\in\bool^+$ and $w \circ \varepsilon = w$;
    \item
    $w \circ 0$ is the longest prefix of $w$ such that $|\cont(w \circ 0)| = |\cont(w)|-1$;
    \item
    $w\circ 1$ is the longest suffix of $w$ such that $|\cont(w \circ 1)| = |\cont(w)|-1$;
    \item
    if $\alpha = \beta\gamma\in \bool ^ +$, for some $\beta\in \bool$ and $\gamma\in \bool ^ *$,
    then $w \circ \alpha = (w \circ \beta) \circ \gamma$.
    \end{enumerate}
\end{definition}
Note that $w\circ 0 = \pref(w)$ and $w\circ 1 = \suff(w)$ in the notation of the
 Green-Rees Theorem.

\begin{definition}\label{definition-ast}
   We define the partial function $\ast: A ^ * \times \bool ^ * \to A$ so that for all $w \in A ^ *$ the following hold:
    \begin{enumerate}[label=(\roman*)]
    \item
     $\varepsilon \ast \alpha = \bot$ and $w\ast \varepsilon =\bot$ for all $\alpha\in\bool^+$;
    \item
    $w \ast 0 \in A$ is such that $w = (w \circ 0)(w\ast 0)w'$ for some $w'\in A ^ *$;
    \item
    $w\ast 1 \in A$ is such that $w = w'' (w\ast 1)(w\circ 1)$ for some $w'' \in A ^ *$;
    \item
    if $\alpha = \beta\gamma\in \bool ^ +$, for some $\beta\in \bool$ and $\gamma\in \bool ^ *$, then
    $w \ast \alpha = (w \circ \beta) \ast \gamma$.
    \end{enumerate}
\end{definition}
As above, $w\ast 0 = \ltof(w)$ and $w\ast 1 = \ftol(w)$ in the notation of the Green-Rees Theorem.

We now state the Green-Rees Theorem using $\circ$ and $\ast$ instead of $\pref$, $\ltof$, $\ftol$, and $\suff$; for a proof see Lemma 4.5.1 in~\cite{howie1995aa}.

\begin{theorem}[Green-Rees]
    \label{thm:equal_free_band}
    Let $v, w\in A^+$. Then
    $v\sim w$ if and only if the following hold:
    \begin{enumerate}[label={\normalfont{(\roman*)}}]
        \item $\cont(v) = \cont(w)$;
        \item $v \circ \alpha \sim w \circ \alpha$ for all $\alpha \in \bool$; and
        \item $v\ast \alpha = w\ast \alpha$ for all $\alpha \in \bool$.\qed
    \end{enumerate}
\end{theorem}

For technical reasons, it is preferable to consider the free band with identity adjoined; which we denote by $\FB(A)$.  Clearly $\FB(A)$ is the quotient of the free monoid $A^\ast$ by the least congruence containing $\left(w, w^2\right)$ for all $w\in A^\ast$. Henceforth we will only consider $\FB(A)$ and as such we write $\sim$ without ambiguity to denote the  least congruence  on $A ^ *$ containing $\left(w, w^2\right)$ for all $w\in A^\ast$; and we will refer to $\FB(A)$ as the free band.

If $v, w\in A ^ *$ are such that $v\sim w$, then $\cont(v) = \cont(w)$ by \cref{thm:equal_free_band}. Hence we may define the \defn{content} of $w/{\sim}$ in $\FB(A)$ as being $\cont(w)$; we denote this by $\cont(w/{\sim})$ also.

The partial function $\circ: A ^ * \times \bool ^ * \to A ^ *$ induces a partial function
from $\FB(A) \times \bool ^ *$ to $\FB(A)$.
Abusing our notation, we will also denote the induced partial function by $\circ$.
Specifically, $\circ: \FB(A) \times \bool ^ *\to \FB(A)$ is defined by
\[(w/{\sim}) \circ \alpha = (w \circ \alpha) / {\sim}\]
for all $w\in A ^ *$ such that $w\circ \alpha$ is defined; see \cref{fig-induced-circ}.

The domain of the function $\circ:\FB(A) \times \bool ^ *\to\FB(A)$ is $\{(x,\alpha) \in \FB(A) \times \bool ^ + : |\alpha| \leq |\cont(x)|\}$.
It follows from \cref{thm:equal_free_band} that
$\circ: \FB(A) \times \bool ^ * \to \FB(A)$ satisfies the analogue of \cref{definition-circ} where $w\in A ^ *$ is replaced by $w/ {\sim}$ in $\FB(A)$.
Similarly, we may define $\ast: \FB(A)\times \bool ^ *\to A$ by
\[(w/{\sim}) \ast \alpha = w \ast \alpha\]
for every $w\in A^*$ such that $w\ast \alpha$ is defined; see \cref{fig-induced-circ}. \cref{thm:equal_free_band} ensures that $\ast$ is well-defined, and satisfies the appropriate analogue of \cref{definition-ast}.

\begin{figure}
  \centering
  \begin{tikzcd}
    A ^ *\times \bool ^ *  \arrow[r, "\circ"] \arrow[d, "/{\sim}\times\textrm{id}"']
    & A ^ * \arrow[d, "/{\sim}"] \\
     \FB(A) \times \bool ^ * \arrow[r, dashed, "\circ"] & \FB(A)
  \end{tikzcd} \quad
  \begin{tikzcd}
    A ^ *\times \bool ^ *  \arrow[r, "\ast"] \arrow[d, "/{\sim}\times\textrm{id}"']
    & A\\
     \FB(A) \times \bool ^ * \arrow[ur, dashed, "\ast"]
  \end{tikzcd}
  \caption{Commutative diagrams giving the definition of the partial functions $\circ$ and $\ast$ for the free band $\FB(A)$.}
  \label{fig-induced-circ}
\end{figure}

\section{Representation by synchronous transducers}
\label{sec:transducer}

If $w\in A ^ *$ is arbitrary, then na\"ively computing all possible values of $w\circ \alpha$ and $w\ast\alpha$ for $\alpha\in \bool ^ *$ has exponential time complexity because there are an exponential number $2 ^ {|\cont(w)|}$ of $\alpha\in \bool ^ *$ such that $w\circ \alpha$ and $w\ast\alpha$ are defined.

In this section, we introduce a means of representing an element $x = w/{\sim}$ of $\FB(A)$ which will permit us to efficiently compute all the possible values of $w\circ \alpha$ and $w\ast\alpha$ for $\alpha\in \bool ^ *$ in \cref{sec:equality}.
We will associate to every word $w\in A ^ *$ a partial function $f_w: \{0, 1\} ^ *\to A ^ *$.
This will yield a bijection from $\FB(A)$ to
the partial functions with domain $\{0, 1\} ^ k$ and codomain $A ^ k$ for $1\leq k\leq |A|$.

\begin{definition}\label{definition-function}
    For a given $w\in A^*$, define the partial function $f_w: \bool^\ast\rightarrow A^\ast$ recursively by
    letting $f_\varepsilon(\varepsilon)=\varepsilon$ and $f_\varepsilon(\alpha) = \bot$ for all $\alpha\in \bool^+$.
    If $w\in A^+$, then we define $f_w:\bool^\ast\rightarrow A^\ast$ by $f_w(\varepsilon) = \bot$ and $f_w(\alpha) = (w\ast\beta)
    \ f_{w\circ \beta}(\gamma)$ for all $\alpha\in\bool^+$ where $\alpha = \beta\gamma$ for some $\beta\in \bool$ and $\gamma\in\bool^\ast$.
\end{definition}

For $\alpha = \alpha_1\cdots \alpha_n\in \bool^*$, we write $\alpha_{(i, j)} = \alpha_i\alpha_{i+1}\cdots \alpha_{j}$ to be the subword starting at index $i\in \{1, \ldots, n\}$ and ending at index $j\in \{i, \ldots, n\}$. If $j<i$, then we define $\alpha_{(i, j)} = \varepsilon$.
By \cref{definition-function},
\[
f_w(\alpha \beta) =
\left(\prod_{i = 1}^n (w\ast \alpha_{(1, i)})\right) f_{w\circ \alpha}(\beta)\]
for all $\alpha, \beta\in \bool^\ast$ where  $\alpha = \alpha_1\alpha_2\cdots \alpha_n$
and $\alpha_i\in\bool$ for $1\leq i\leq n$.
As an example, we compute $f_{abac}(010)$ as follows:
\begin{align*}
f_{abac}(010) & = (abac\ast 0)\ f_{abac\circ 0}(10) = cf_{aba}(10) = c (aba \ast 1) f_{aba\circ 1}(0)
\\ & = cbf_a(0)  = cb(a\ast 0) f_{a\circ 0}(\varepsilon)
 = cbaf_{\varepsilon}(\varepsilon) = cba.
\end{align*}
In a similar way, it is possible to determine $f_{abac}(\alpha)$ for any $\alpha \in \{0, 1\} ^ 3$:
\begin{align*}
f_{abac}(000) = cba, & & f_{abac}(001) = cba, & & f_{abac}(010) = cba, & & f_{abac}(011) = cba,\\
f_{abac}(100) = bca, & & f_{abac}(101) = bca, & & f_{abac}(110) = bac, & & f_{abac}(111) = bac.
\end{align*}

We will show that the partial function $f_w$ can be realized by a transducer for every $w\in A^\ast$.
For the sake of completeness, we next give a definition of a synchronous deterministic transducer.
\begin{definition}
    A \defn{deterministic synchronous transducer} is a septuple $\mathcal{T} = (Q, \Sigma, \Gamma, q_0, T, \spadesuit, \clubsuit)$ where
    \begin{enumerate}[label=(\roman*)]
        \item $Q$ is a non-empty set called the \defn{set of states};
        \item $\Sigma$ is a set called the \defn{input alphabet};
        \item $\Gamma$ is a set called the \defn{output alphabet};
        \item $q_0\in Q$ is called the \defn{initial state};
        \item $\varnothing \neq T\subseteq Q$ is called the \defn{set of terminal states};
        \item a partial function $\spadesuit: Q\times \Sigma \rightarrow Q$ called the \defn{state transition function};
        \item a partial function $\clubsuit: Q\times \Sigma \rightarrow \Gamma$ called the \defn{letter transition function};
    \end{enumerate}
    where for all $q \in Q, \alpha\in \Sigma$, $q\spadesuit \alpha \neq \bot$ is if and only if $q\clubsuit \alpha \neq \bot$.
\end{definition}
Since the only type of transducers we consider are deterministic synchronous transducers, we will write ``transducer'' to mean ``deterministic synchronous transducer''.

The definition of $\spadesuit$ can be extended recursively for inputs in $Q\times \Sigma^\ast$ by letting $q\spadesuit \varepsilon = q$ and $q\spadesuit \alpha\beta = (q\spadesuit \alpha) \spadesuit \beta$ for all $\alpha\in \Sigma$, and $\beta \in\Sigma^\ast$. Similarly, we may extend $\clubsuit$ to inputs in $Q\times \Sigma^\ast$ by letting $q\clubsuit \varepsilon = \bot$ and $q\clubsuit \alpha \beta = (q\spadesuit \alpha)\beta$ for all $\alpha\in \Sigma, \beta\in\Sigma^\ast$.
Let $\mathcal{T} = (Q, \Sigma, \Gamma, q_0, T, \spadesuit, \clubsuit)$ be a transducer. Then  the \defn{language accepted} by $\mathcal{T}$ is $L(\mathcal{T}) = \{\alpha\in \Sigma^\ast : q_0\spadesuit \alpha \in T\}$. We say that a partial function $f: \Sigma^\ast \rightarrow \Gamma^\ast$ \defn{is realized by} $\mathcal{T}$ if and only if $f(\alpha) = q_0\clubsuit \alpha$ for all $\alpha\in L(\mathcal{T})$ and $f(\alpha) = \bot$ for all $\alpha \not \in L(\mathcal{T})$.
If $A$ is any set, and $w\in A ^ *$, then there are many transducers realising the function $f_w$ given in \cref{definition-function};
one such transducer is defined as follows.

\begin{example}
    \label{ex:treelike_transducer}
    Let $w\in A ^ *$ be arbitrary and let $k = |\cont(w)|$. Then we define the transducer
    \[\left({\bool}^{\leq k}, \bool, A, \varepsilon, {\{0, 1\}}^{k}, \clubsuit, \spadesuit\right)\]
    with states ${\bool}^{\leq k} = \bigcup_{i=0}^k \bool^i$,
    input alphabet $\bool$, output alphabet $A$, initial state $\varepsilon$,
    terminal states ${\{0, 1\}}^{k}$, and state and letter transition functions given respectively by:
    \begin{align*}
        q\spadesuit x = qx & & q\clubsuit x = w \ast (qx).
    \end{align*}
    It is routine to verify that this transducer realises $f_w$.
A picture of the corresponding transducer realizing $f_{abac}$ can be found in~\cref{fig:transducer-abac-treelike}.
\end{example}

\begin{figure}
    \centering
        \begin{tikzpicture}
    \node (q_e) at (0,0)  {$\varepsilon$};
    \node (q_0) at (-4, -1) {$0$};
    \node (q_1) at (4, -1)  {$1$};
    \node (q_00) at (-6, -2) {$00$};
    \node (q_01) at (-2, -2)  {$01$};
    \node (q_10) at (2, -2) {$10$};
    \node (q_11) at (6, -2)  {$11$};
    \node (q_000) at (-7, -3) {$000$};
    \node (q_001) at (-5, -3)  {$001$};
    \node (q_010) at (-3, -3) {$010$};
    \node (q_011) at (-1, -3)  {$011$};
    \node (q_100) at (1, -3) {$100$};
    \node (q_101) at (3, -3)  {$101$};
    \node (q_110) at (5, -3) {$110$};
    \node (q_111) at (7, -3)  {$111$};
    \draw [-latex] (q_e)++(0,0.5) -- (q_e);
    \draw [-latex,draw=color2] (q_e) -- node[midway, above left]{$0|c$} (q_0);
    \draw [-latex reversed, dashed, dashed,draw=color1] (q_e) -- node[midway, above right]{$1|b$} (q_1);
    \draw [-latex,draw=color1] (q_0) -- node[midway, above left]{$0|b$} (q_00);
    \draw [-latex reversed, dashed, dashed,draw=color1] (q_0) -- node[midway, above right]{$1|b$} (q_01);
    \draw [-latex,draw=color2] (q_1) -- node[midway, above left]{$0|c$} (q_10);
    \draw [-latex reversed, dashed, dashed,draw=color0] (q_1) -- node[midway, above right]{$1|a$} (q_11);
    \draw [-latex,draw=color0] (q_00) -- node[midway, above left]{$0|a$} (q_000);
    \draw [-latex reversed, dashed, dashed,draw=color0] (q_00) -- node[midway, above right]{$1|a$} (q_001);
    \draw [-latex,draw=color0] (q_01) -- node[midway, above left]{$0|a$} (q_010);
    \draw [-latex reversed, dashed, dashed,draw=color0] (q_01) -- node[midway, above right]{$1|a$} (q_011);
    \draw [-latex,draw=color0] (q_10) -- node[midway, above left]{$0|a$} (q_100);
    \draw [-latex reversed, dashed, dashed,draw=color0] (q_10) -- node[midway, above right]{$1|a$} (q_101);
    \draw [-latex,draw=color2] (q_11) -- node[midway, above left]{$0|c$} (q_110);
    \draw [-latex reversed, dashed, dashed,draw=color2] (q_11) -- node[midway, above right]{$1|c$} (q_111);
    \draw [-latex] (q_000) -- ++(0,-0.5);
    \draw [-latex] (q_001) -- ++(0,-0.5);
    \draw [-latex] (q_010) -- ++(0,-0.5);
    \draw [-latex] (q_011) -- ++(0,-0.5);
    \draw [-latex] (q_100) -- ++(0,-0.5);
    \draw [-latex] (q_101) -- ++(0,-0.5);
    \draw [-latex] (q_110) -- ++(0,-0.5);
    \draw [-latex] (q_111) -- ++(0,-0.5);
\end{tikzpicture}

    \caption{A transducer realizing $f_{abac}$. Note that transitions on input $0$ are drawn with a solid arrow and transitions on input $1$ are drawn with a dashed and reversed arrow for visual differentiation. In subsequent diagrams we will omit the node labels and part of the edge labels for the sake of brevity.}
    \label{fig:transducer-abac-treelike}
\end{figure}
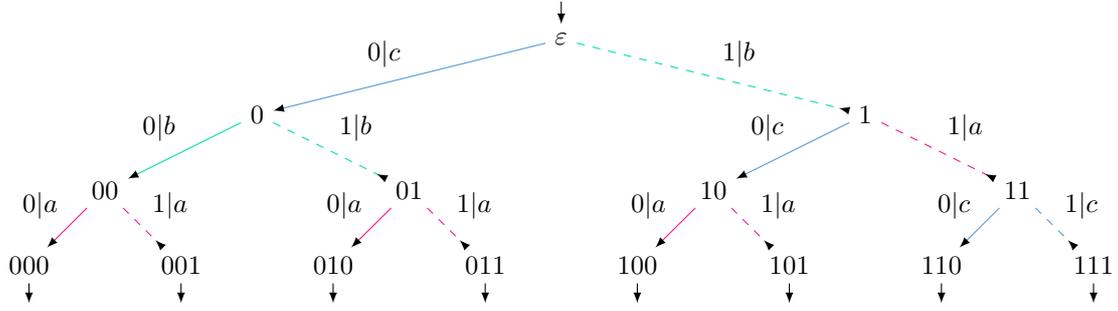

The relevance of the functions $f_w$ to free bands is established in the next result.

\begin{lemma}\label{lem:word_function_lemma}
    Let $v, w\in A^*$. Then $v\sim w$ if and only if $f_v = f_w$.
\end{lemma}
\begin{proof}
    If $|\cont(v)|\neq|\cont(w)|$, then $\cont(v)\neq \cont(w)$, so $v\not\sim w$ by \cref{thm:equal_free_band}. Also $f_v\neq f_w$, since the domains of $f_v$ and $f_w$ disagree. Therefore we only need to prove the lemma in the case $|\cont(v)| = |\cont(w)|$.

    We proceed by induction on the size of the content. Clearly if $|\cont(w)| = |\cont(v)| = 0$, then $v=w=\varepsilon$ and $f_v=f_w=f_\varepsilon$, so the statement holds.
    Assume that the statement of the lemma holds for all $v, w\in A ^ *$ with $|\cont(v)| = |\cont(w)| < k$ for some $k\geq 1$.

    Let $v,w\in A^\ast$ with $|\cont(v)| = |\cont(w)| = k$. If $v\sim w$, then, by \cref{thm:equal_free_band}, $v\circ\beta\sim w\circ\beta$ and $v\ast\beta=w\ast\beta$ for all $\beta\in\bool$. By the inductive hypothesis, since $|\cont(v\circ\beta)|, |\cont(w\circ\beta)| < k$, it follows that $f_{v\circ\beta} = f_{w\circ\beta}$ for all $\beta\in\bool$. But then $f_v(\alpha) = (v\ast\beta) f_{v\circ \beta}(\gamma) = (w\ast\beta) f_{w\circ\beta}(\gamma) = f_w(\alpha)$ for all $\alpha\in \bool^k$, so $f_v=f_w$.

    On the other hand, if $f_v = f_w$, then $(v\ast\beta) f_{v\circ \beta}(\gamma) = (w\ast\beta) f_{w\circ\beta}(\gamma)$ for all $\beta\in\bool, \gamma\in\bool^{k-1}$. But then $v\ast\beta=w\ast\beta$ for all $\beta\in \bool$, and $f_{v\circ\beta}(\gamma) = f_{w\circ\beta}(\gamma)$ for all $\gamma\in \bool^{k-1}$, so $f_{v\circ\beta} = f_{w\circ\beta}$ for all $\beta\in\bool$. Hence, by the inductive hypothesis, $v\circ\beta \sim w\circ\beta$ for all $\beta\in \bool$. Finally, by \cref{thm:equal_free_band}, $v\circ0\sim  w\circ0$ implies $\cont(v\circ0)=\cont(w\circ0)$ so that $\cont(v) = \{v\ast0\}\cup\cont(v\circ 0) = \{w\ast0\}\cup\cont(w\circ0) = \cont(w)$. It follows, again by \cref{thm:equal_free_band}, that $v\sim w$.
\end{proof}

It follows from \cref{lem:word_function_lemma} that
if $x\in \FB(A)$, then we may denote by $f_x$ any of the functions $f_w$ where $w\in A ^ *$ and $w /{\sim} = x$ without ambiguity. In the same vein, we say that a transducer $\mathcal{T}$ \defn{represents} $x\in \FB(A)$ if $\mathcal{T}$ realizes $f_x$.

For a transducer $\mathcal{T} = (Q, \Sigma, \Gamma, q_0, T, \spadesuit, \clubsuit)$ and a state $p_0\in Q$, we define the \defn{induced subtransducer rooted at $p_0$} to be the transducer $\mathcal{T}^\prime = (Q, \Sigma, \Gamma, p_0, T, \spadesuit, \clubsuit)$, i.e. the transducer with the same states, transitions and terminal states, but whose initial state is $p_0$ instead of $q_0$.
Let $g$ be the function realized by $\mathcal{T}$, and let $g'$ be the function realized by the subtransducer $\mathcal{T}^\prime$. It can be shown that if  $q_0\spadesuit \alpha = p_0$ for some $\alpha\in\Sigma^\ast$, then $g^\prime(\beta)$ is the suffix of $g(\alpha\beta)$ of length $|\beta|$ for all $\beta \in \Sigma^\ast$ such that $g(\alpha\beta)\neq \bot$; and where $g^\prime(\beta) = \bot$ otherwise.

If a transducer $\mathcal{T}$ represents $x\in \FB(A)$, it would be convenient if the subtransducer rooted at $q_0\spadesuit \alpha$ represented $x\circ\alpha$ for any $\alpha\in \bool^\ast$, wherever $x\circ \alpha$ is defined. This is indeed the case,  which we establish in the next lemma.
This lemma will allow us to determine properties of elements of the free band by considering the structure of a representative transducer directly, rather than considering the function realized by that transducer.
\begin{lemma}\label{thm:subtransducers}
    Let $\mathcal{T} = (Q, \bool, A, q_0, T, \spadesuit, \clubsuit)$ be a transducer representing $x\in \FB(A)$. If $\alpha\in \bool^\ast$ is such that $q_0\spadesuit \alpha= q$, then
    the induced subtransducer rooted at $q$ realizes $f_{x\circ \alpha}$.
\end{lemma}
\begin{proof}
    Let $x\in \FB(A)$, let $\mathcal{T} = (Q, \bool, A, q_0, T, \spadesuit, \clubsuit)$ be a transducer realizing the function $f_x$, let $q\in Q$ be reachable, and let $\alpha\in \bool^\ast$ be such that $q_0\spadesuit \alpha = q$.

    Suppose that $g$ is the function realized by the induced subtransducer rooted at $q$. To compute $g$ for $\beta\in \bool ^ \ast$ we take the suffix of $f_x(\alpha\beta)$ obtained by removing the first $|\alpha|$ output letters. But, by the definition of $f_x$,
    \[f_x(\alpha\beta) = \left(\prod_{i=1}^{|\alpha|} x\ast\alpha_{(1,i)}\right) f_{x\circ \alpha}(\beta).\]
    Since $\prod_{i=1}^{|\alpha|} x\ast \alpha_{(1, i)}$ is the length $|\alpha|$ prefix of $f_x(\alpha\beta)$, $g(\beta) = f_{x\circ \alpha}(\beta)$ for all $\beta\in\bool^\ast$. In particular,  $g$ is the function $f_{x\circ \alpha}$ realized by the induced subtransducer rooted at $q$. Therefore the subtransducer rooted at $q$ represents $x\circ \alpha$.
\end{proof}

Let  $\mathcal{T}_x = (Q, \bool, A, q_0, T, \spadesuit, \clubsuit)$ be a transducer representing $x\in \FB(A)$. If we define the partial function $\Psi : Q\rightarrow \FB(A)$ by
\[
\Psi(q) =
\begin{cases}
x\circ \alpha & \text{if } q = q_0\spadesuit \alpha \text{ for some } \alpha\in\bool^\ast  \\
\bot          & \text{otherwise},
\end{cases}
\]
then it follows from \cref{thm:subtransducers} that the diagrams in \cref{fig:induced_state_function} commute. Consequently, we may refer to a state $q$ in a transducer as \defn{representing} the element $\Psi(q)$ of $\FB(A)$.
\begin{figure}
  \centering
  \begin{tikzcd}
    A ^ *\times \bool ^ *  \arrow[r, "\circ"] \arrow[d, "/{\sim}\times\textrm{id}"']
    & A ^ * \arrow[d, "/{\sim}"] \\
     \FB(A) \times \bool ^ * \arrow[r, "\circ"] & \FB(A)\\
     Q\times \bool^\ast \arrow[u, "\Psi\times\textrm{id}", dashed]\arrow[r, "\spadesuit"] & Q\arrow[u,"\Psi"', dashed]
  \end{tikzcd} \quad
  \begin{tikzcd}
    A ^ *\times \bool ^ *   \arrow[dr,"\ast"]\arrow[d, "/{\sim}\times\textrm{id}"']
    \\
     \FB(A) \times \bool ^ * \arrow[r, "\ast"] & A\\
    Q\times \bool ^ *   \arrow[ur,"\clubsuit"']\arrow[u, dashed, "\Psi\times\textrm{id}"]
  \end{tikzcd}
  \caption{Commutative diagrams giving correspondence between states in a transducer and elements of the free band.}
  \label{fig:induced_state_function}
\end{figure}
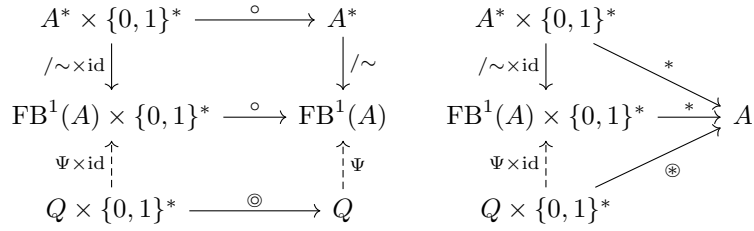

\section{Equality Checking}\label{sec:equality}

In this section we give a method using transducers for checking equality in the
free band. The algorithm we describe in this section is heavily based on that
of Radoszewski and Rytter from~\cite{radoszewski2010aa}. Roughly speaking,
in~\cite{radoszewski2010aa}, to compare two words $u, v\in A^\ast$ for equality
in the free band, the authors construct an automaton, called the \defn{interval
automaton} that has $\mathcal{O}((|u|+|v|)\cdot |A|)$ states. Equality of $u$
and $v$ in the free band is then established by checking the equivalence of two
states corresponding to $u$ and $v$ within the interval automaton. This is done
without storing the entire automaton, but only the relevant parts at every
step.  In this section, we adapt this idea,  constructing one transducer for
each word and then comparing the transducers. The principal difference is that
we retain the whole transducer, which leads to worse space complexity than
Radoszewski and Rytter's approach, but has the benefit of permitting us to
define multiplication in \cref{sec:multiplication} and obtain the short-lex
least word in a $\sim$-class in \cref{sec:minimal}.

We start this section by recasting the interval automaton for $w\in A^\ast$ as a transducer using the notation established in the previous section.
Let $w\in A^+$. A \defn{content-$k$ subword} of $w$ is a subword $v$ of $w$ such that $|\cont(v)| = k$. A content-$k$ subword is called \defn{prefix maximal} if it is not a proper prefix of any other content-$k$ subword. Similarly, a subword is \defn{suffix maximal} if it is not the proper suffix of any content-$k$ subword.  There are at most $|w|$ prefix  maximal content-$k$ subwords of $w$ for any given $k$, since each letter of $w$ can be the start of at most one such word. An analogous statement holds for suffix maximal content-$k$ subwords also.

 Let $w = a_1\cdots a_n\in A^+$ where $a_1, \ldots, a_n\in A$ and $k\in\mathbb{N}_0$. Then we define the partial function
 \[
 \rght_k:\{1,\ldots, n\}\rightarrow \{ 1,\ldots, n\}
 \]
 such that
 $w_{(i, \rght_k(i))}$ is the unique prefix maximal content-$k$ subword of $w$ starting at index $i$ if such a subword exists, and $w_{(i, \rght_k(i))} = \bot$ otherwise.
 The partial function $\lft_k:\{1,\ldots, n\}\rightarrow \{1,\ldots, n\}$ is defined dually.

The \defn{interval transducer} of $w\in A ^ *$ is defined as follows:
\begin{equation*}\label{eq:rr-transducer}
\mathcal{T} = \left(Q, \{0, 1\}, A, q_0, T, \spadesuit, \clubsuit\right)
\end{equation*}
where
\begin{enumerate}[label=(\roman*)]
    \item the set of states $Q$ consists of a symbol $0$ and
    all of the pairs $(i, j)\in\{1,\ldots, n\}\times \{1, \ldots, n\}$ such that $w_{(i,j)}$ is either a prefix maximal or suffix maximal content-$k$ subword of $w$ for some $k\in\mathbb{N}$;
    \item the initial state is $q_0 = (1, n)$ (corresponding to the whole word $w$);
    \item the set of terminal states is $T = \{0\}$;
    \item if $(i,j)\in Q$ where $|\cont(w_{(i, j)})| = k > 1$, then the state and letter transition functions are given by
    \begin{align*}
        (i, j) \spadesuit 0 &= (i, \rght_{k-1}(i)) & (i, j) \spadesuit 1 =& (\lft_{k-1}(j), j)\\
        (i, j) \clubsuit 0 &= a_{\rght_{k-1}(i) + 1} & (i, j) \clubsuit 1 =& a_{\lft_{k-1}(j)-1};
    \end{align*}
    \item
      $(i,j)\in Q$ where $|\cont(w_{(i, j)})| = 1$, then
    \begin{align*}
        (i, j) \spadesuit \alpha = 0 & &
        (i, j) \clubsuit \alpha = a_{i}; &&
    \end{align*}
    \item $0\spadesuit\alpha = 0\clubsuit \alpha = \bot$ for all $\alpha \in \bool$.
\end{enumerate}
The  transducer has a linear number of states in $|A|\cdot |w|$ since  $|Q|\leq 2\cdot |A|\cdot |w| + 1$, since there are at most $|w|$ prefix or suffix maximal content-$k$ subwords for any given $k$, and no prefix or suffix maximal content-$k$ subwords for $k\geq |A|$.
The proof that $\mathcal{T}$ realizes $f_w$ is the same as that given in \cite[Lemmas 1 and 2]{radoszewski2010aa}, and for the sake of brevity, we will not duplicate the proof here.
An algorithm is given in~\cite[Theorem 2]{radoszewski2010aa} that constructs the interval transducer $\mathcal{T}$ for $w\in A ^ *$ in $\mathcal{O}(|A|\cdot |w|)$ time; we denote this algorithm by $\IntervalTransducer$.

\begin{figure}
    \centering
        \begin{tikzpicture}[scale=1.2]
\node (q_0) at (4,2) {$0$};
\node (q_1) at (0,4) {$(1, 1)$};
\node (q_2) at (2,4) {$(2, 2)$};
\node (q_3) at (4,4) {$(3, 3)$};
\node (q_4) at (6,4) {$(4, 4)$};
\node (q_5) at (2,6) {$(1, 3)$};
\node (q_6) at (4,6) {$(2, 3)$};
\node (q_7) at (6,6) {$(3, 4)$};
\node (q_8) at (0,6) {$(1, 2)$};
\node (q_9) at (4,8) {$(1, 4)$};
\node (q_10) at (6,8) {$(2, 4)$};

\draw [-latex] (q_0) -- ++(0,-0.5);
\draw [-latex] (q_9)++(0,0.5) -- (q_9);

\draw [ -latex ,draw=color0 ] (q_1) to[ out=-36, in=164 ] node[near start, below] {$0|a$} (q_0);
\draw [ -latex reversed, dashed ,draw=color0 ] (q_1) -- node[near start, above] {$1|a$} (q_0);

\draw [ -latex ,draw=color1 ] (q_2) to[ out=-55, in=145] node[near start, left] {$0|b$} (q_0);
\draw [ -latex reversed, dashed ,draw=color1 ] (q_2) to[ out=-35, in=125 ] node[near start, above] {$1|b$} (q_0);

\draw [ -latex ,draw=color0 ] (q_3) to[ out=-105, in= 105] node[near start, left] {$0|a$} (q_0);
\draw [ -latex reversed, dashed ,draw=color0 ] (q_3) to[out=-75, in=75] node[near start, right] {$1|a$} (q_0);

\draw [-latex, draw=color2] (q_4) to [out=-135, in=55] node[near start, above] {$0|c$} (q_0);
\draw [ -latex reversed, dashed ,draw=color2 ] (q_4) to [out=-125, in=35] node[near start, right] {$1|c$} (q_0);

\draw [ -latex ,draw=color1 ] (q_5) -- node[near start, inner sep=6pt, above] {$0|b$} (q_1);
\draw [ -latex reversed, dashed ,draw=color1 ] (q_5) -- node[near start, inner sep=6pt, above] {$1|b$} (q_3);

\draw [ -latex ,draw=color0 ] (q_6) -- node[near start, inner sep=6pt, above] {$0|a$} (q_2);
\draw [ -latex reversed, dashed ,draw=color1 ] (q_6) -- node[near start, right] {$1|b$} (q_3);

\draw [ -latex ,draw=color2 ] (q_7) -- node[near start, inner sep=6pt, above] {$0|c$} (q_3);
\draw [ -latex reversed, dashed ,draw=color0 ] (q_7) -- node[near start, inner sep=6pt, right] {$1|a$} (q_4);

\draw [ -latex ,draw=color1 ] (q_8) -- node[near start, left] {$0|b$} (q_1);
\draw [ -latex reversed, dashed ,draw=color0 ] (q_8) -- node[near start, inner sep=6pt, above] {$1|a$} (q_2);

\draw [-latex, draw=color2] (q_9)  -- node[near start, inner sep=6pt, above] {$0|c$} (q_5);
\draw [-latex reversed, dashed, draw=color1] (q_9) -- node[near start, inner sep=6pt, above] {$1|b$} (q_7);

\draw [-latex,draw=color2] (q_10) -- node[near start, inner sep=6pt, above] {$0|c$} (q_6);
\draw [-latex reversed, dashed, draw=color1] (q_10) -- node[near start, right] {$1|b$} (q_7);

\end{tikzpicture}

    \caption{The interval transducer realizing $f_{abac}$.}
    \label{fig:transducer-abac-interval}
\end{figure}
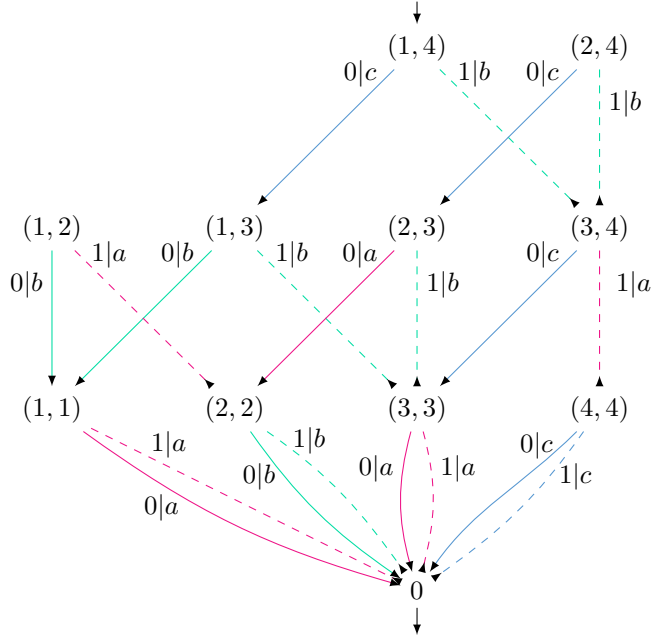

A transducer is \defn{trim} if every state belongs to some path starting at the initial state and ending at a terminal state.  Any trim transducer $\mathcal{T}$ realizing $f_w$ for $w\in A^+$ is acyclic, since the domain of $f_w$ is finite.
A linear time algorithm for minimizing acyclic deterministic synchronous transducers, such as those in this paper, can be derived from Revuz~\cite{revuz1992aa};  we refer to this algorithm as $\Minimize$.

We say that transducers $\mathcal{T} = (Q, \Sigma, \Gamma, q_0, T, \spadesuit, \clubsuit)$ and
$\mathcal{T}^\prime = (Q^\prime, \Sigma^\prime, \Gamma^\prime, q_0^\prime, T^\prime, \spadesuit^\prime, \clubsuit^\prime)$ are \defn{isomorphic} if the input and output alphabets coincide, $\Sigma = \Sigma^\prime, \Gamma = \Gamma^\prime$,
and there exists a bijection  $f: Q\rightarrow Q^\prime$ such that $f(q_0) = q_0^\prime, f(T) = T^\prime$, $f(q\spadesuit \alpha) = f(q)\spadesuit^\prime \alpha$ and $q\clubsuit \alpha = f(q)\clubsuit^\prime \alpha$ for all $q\in Q, \alpha\in\Sigma$. In other words, two transducers are isomorphic if they differ only by the labelling of states.

Just as with automata, the minimal deterministic synchronous transducer realizing a given function is unique up to isomorphism.
It is also possible to check if two trim deterministic synchronous transducers are isomorphic in linear time in the number of transitions of the transducers. Essentially we treat the synchronous transducer as an automaton by making each transition be labelled by the pair containing both the input and output symbol, then apply the isomorphism algorithm for connected trim automata.
We will call this algorithm $\TrimTransducerIsomorphism$.

The algorithms $\TrimTransducerIsomorphism$, $\Minimize$, and $\IntervalTransducer$ can now be combined into an algorithm for checking equality in the free band, \cref{algorithm:equal_in_free_band}.

\begin{algorithm}[ht]
\SetAlgoLined\DontPrintSemicolon\KwArguments{words $v, w\in A^\ast$;}
\KwReturns{if $v \sim w$ or not.}
Let $\mathcal{T} = \Minimize(\IntervalTransducer(v))$\;
Let $\mathcal{T}^\prime = \Minimize(\IntervalTransducer(w))$\;
\KwRet{$\TrimTransducerIsomorphism(\mathcal{T}_v, \mathcal{T}_w)$}\;
\caption{$\EqualInFreeBand$}
\label{algorithm:equal_in_free_band}
\end{algorithm}

The total time and space complexity of this algorithm is $O((|v|+|w|)\cdot|A|)$.
A version of $\EqualInFreeBand$ is present in~\cite[Theorem 2]{radoszewski2010aa}.
However, the algorithm described in~\cite{radoszewski2010aa} achieves a better space complexity ($\mathcal{O}(|v|+|w|)$) than \cref{algorithm:equal_in_free_band} by processing the transducers $\mathcal{T}$ and $\mathcal{T}^\prime$ layer by layer. Whereas in the method described here, we store the full transducers, and so our algorithm has complexity $\mathcal{O}((|v|+|w|)\cdot |A|)$.

 The preceding discussion also allows us to specify the size of the minimal
 transducer. Recall that if $A$ is any set, then the \defn{short-lex ordering} on $A^\ast$ is defined by $u\leq v$ if $|u|<|v|$ or if $|u|=|v|$ and $u$ is lexicographically less than $v$ (with respect to some underlying ordering on $A$) for all $u, v\in A^*$. For $w\in A^+$ let $\min(w)\in A^+$ denote the short-lex least word such that $w \sim \min(w)$. Applying $\IntervalTransducer$ to $\min(w)$ gives a transducer $\mathcal{T}$ with at most $2\cdot |\min(w)|\cdot |A|+1$ states realizing $f_{\min(w)} = f_w$. The minimal transducer realizing $f_w$ will have no more states than $\mathcal{T}$; we record this in the next theorem.
\begin{theorem}\label{thm:size_min_transducer}
If $A$ is any non-empty set, and $w\in A ^*$ is arbitrary, then
the minimal transducer realizing $f_w$ has at most $2\cdot |\min(w)|\cdot |A|+1$ states.\qed
\end{theorem}
Note that it is not necessary to compute (or even know) $\min(w)$ when constructing the minimal transducer realizing $f_w$.

\section{Multiplication}\label{sec:multiplication}
In this section, we describe how to multiply elements of the free band using transducers representing the elements. Suppose that $\cdot$ denotes the multiplication within $\FB(A)$ for some alphabet $A$. One approach to finding a transducer representing $x\cdot y$ given $x, y\in \FB(A)$ arises from the algorithms described in \cref{sec:equality}.
If $v, w\in A^\ast$ represent $x, y\in\FB(A)$, respectively, then a transducer for $x\cdot y$ is given by $\IntervalTransducer(vw)$. This construction has complexity $\mathcal{O}((|v|+|w|)\cdot |A|)$.
Concatenating words, in all but the most trivial cases, produces a word that is not a minimal representative of the product. Therefore repeated multiplications using words can quickly lead to very long representatives.
 In this section, we give an algorithm that constructs a transducer $\mathcal{T}_{x\cdot y}$, representing the product $x\cdot y$, from any transducers $\mathcal{T}_x$  and $\mathcal{T}_y$ representing $x, y\in \FB(A)$, respectively. The complexity of this approach is $\mathcal{O}(|\mathcal{T}_x| + |\mathcal{T}_y| + |A|^2)$ where $|\mathcal{T}_x|$ and $|\mathcal{T}_y|$ are the number of states of transducers $\mathcal{T}_x$ and $\mathcal{T}_y$ respectively. By \cref{thm:size_min_transducer}, if $\mathcal{T}_x$ and $\mathcal{T}_y$ are minimal transducers, then $\Minimize(\mathcal{T}_{x\cdot y})$ can be found in $\mathcal{O}((|\min(x)| + |\min(y)| + |A|)\cdot |A|)$ time and $|\Minimize(\mathcal{T}_{x\cdot y})|\leq 2\cdot|A|\cdot|\min(x\cdot y)| + 1$. The number of states in the minimal transducer for $\mathcal{T}_{x\cdot y}$ is hence within a constant multiple of the minimum possible length of any word in $A ^ *$ representing $x\cdot y$. This resolves the issue of the representative length blow-up in repeated multiplications.

In the following series of results, we establish some necessary properties of the product $x\cdot y$ and show how to compute these properties.

\begin{lemma}
    \label{lem:multiplication}
    Let $x, y \in \FB(A)$ with $x, y\neq \varepsilon$. Then
    \begin{align*}
        (x\cdot y) \circ 0 &= \begin{cases}
            x\cdot (y\circ 0) & \textrm{if}\ y\ast 0\not\in\cont(x)\\
            (x\cdot (y\circ 0))\circ 0 & \textrm{otherwise},
        \end{cases}\quad
        (x\cdot y) \ast 0 &= \begin{cases}
            y\ast 0 & \textrm{if}\ y\ast 0\not\in\cont(x)\\
            (x\cdot (y\circ 0))\ast 0 & \textrm{otherwise},
        \end{cases}\\
        (x\cdot y) \circ 1 &= \begin{cases}
            (x\circ 1)\cdot y & \textrm{if}\ x\ast 1\not\in\cont(y)\\
            ((x\circ 1)\cdot y)\circ 1 & \textrm{otherwise},
        \end{cases}\quad
        (x\cdot y) \ast 1 &= \begin{cases}
            x\ast 1 & \textrm{if}\ x\ast 1\not\in\cont(y)\\
            ((x\circ 1)\cdot y)\ast 1 & \textrm{otherwise}.
        \end{cases}
    \end{align*}
\end{lemma}
\begin{proof}
We will prove the first two statements, the second two statements follow by a similar argument.

Let $x, y\in \FB(A)$ and
let $v, w\in A ^ *$ such that $v/{\sim} = x$ and $w/{\sim} = y$. Then $(vw)/{\sim} =x\cdot y$ is a representative of  $x\cdot y$.  Hence it suffices by the commutative diagram in \cref{fig-induced-circ} to prove the result for $v$ and $w$ instead of $x$ and $y$.
By assumption $x\neq \varepsilon$ and $y\neq \varepsilon$, and so $v\ast \alpha$ and $v\circ \alpha$ are defined for all $\alpha\in\bool$ and similarly for $w$.

If $w\ast 0\not\in \cont(v)$, then $|\cont(v(w\circ 0))| = |\cont(vw)|-1$, since $w\ast 0 \not \in \cont(v)$ by assumption, and $w\ast 0 \not\in \cont(w\circ 0)$ by definition. Furthermore, $v(w\circ0)$ is the largest prefix of $vw$ with content strictly contained in $\cont(vw)$, since the next letter in $vw$ is $w\ast 0$. So $(vw)\circ 0 = v(w\circ 0)$ and $(vw)\ast 0 = w\ast 0$.

If $w\ast 0 \in\cont(v)$, then $\cont(v(w\circ0)) = \cont(vw)$. Since  $v(w\circ 0)$ is a prefix with the same content as $vw$, the first letter to occur last in $vw$ must occur somewhere within $v(w\circ 0)$. Therefore $(vw)\circ 0 = (v(w\circ0))\circ 0$ and $(vw)\ast 0 = (v(w\circ0))\ast 0$.
\end{proof}

Recall that \cref{thm:equal_free_band} essentially allows us to reconstruct $z\in \FB(A)$ given  $\cont(z), z\circ \alpha$, and $z\ast \alpha$ for all $\alpha\in \bool$. Consider the case when $z=x\cdot y$ for some $x, y\in \FB(A)$. We know that $\cont(x\cdot y) = \cont(x)\cup \cont(y)$. \cref{lem:multiplication} allows us to compute $(x\cdot y)\circ \alpha$ and $(x\cdot y)\ast \alpha$ for $\alpha\in \bool$. However, using \cref{lem:multiplication} na\"ively to compute $x\cdot y$ requires exponential time, since in order to compute $(x\cdot y)\circ 0$ and $(x\cdot y)\circ 1$ we must compute $x\cdot (y\circ 0)$ and $(x\circ 1) \cdot y$.
 We will show that it is possible to reduce the complexity by iterating \cref{lem:multiplication} and considering the possible values of $(x\cdot y)\circ \alpha$ for all $\alpha\in\bool^\ast$.
In \cref{cor:multiplication_iterative}, we will show that:
\[(x\cdot y)\circ \alpha
\in\{(x\circ 1^i) \cdot (y\circ 0^j) : 0\leq i\leq |\cont(x)|, 0\leq j \leq |\cont(y)|\}
\cup\{x \circ \beta : \beta \in \bool^\ast\}
\cup\{y \circ \gamma : \gamma \in \bool^\ast\}
\]
for all $\alpha \in \bool^\ast$. So, to compute $(x\cdot y)\circ \alpha$ for $\alpha\in\bool^\ast$, requires $\mathcal{O}(|\cont(x)|\cdot |\cont(y)|) = \mathcal{O}(|A|^2)$ products of the form $(x\circ 1^i) \cdot (y\circ 0^j)$, and the finitely many values in
\[\{x \circ \beta : \beta \in \bool^\ast\}
\cup\{y \circ \gamma : \gamma \in \bool^\ast\}\]
which are represented in the given transducers for $x$ and $y$.
We will show in \cref{lem:multiplication:iterative} that the values of $((x\circ 1^i) \cdot (y\circ 0^j))\circ \alpha$
and $((x\circ 1^i) \cdot (y\circ 0^j))\ast \alpha$ for $\alpha \in \bool$ can be derived by considering how $\cont(x\circ 1^i)$ and $\cont(y\circ 0^j)$ interact.
Taken together, this will allow us to construct a transducer representing $x\cdot y$, by taking the  transducers representing $x$ and $y$, adding $\mathcal{O}(|A|^2)$ states, and the necessary transitions  corresponding to the products $\{(x\circ 1^i) \cdot (y\circ 0^j) : 0\leq i\leq |\cont(x)|, 0\leq j \leq |\cont(y)|\}$.

We will make use of the following partial functions.

\begin{definition}
    Let $K_{0}, K_{1}:\FB(A)\times \FB(A)\rightarrow \mathbb{N}$ be partial functions defined by:
    \begin{align*}
        K_0(x, y) & =
        \begin{cases}
        \min \{k \in \mathbb{N} : y\ast 0 ^ k\neq \bot\text{ and }y\ast 0 ^k \not \in \cont(x)\} & \text{if such $k$ exists}\\
        \bot & \text{otherwise,}
        \end{cases}\\
        K_1(x, y) & =
        \begin{cases}
        \min \{l \in \mathbb{N} : x\ast 1 ^ l\neq \bot\text{ and }x\ast 1 ^l \not \in \cont(y)\} & \text{if such $l$ exists}\\
        \bot & \text{otherwise.}
        \end{cases}
    \end{align*}
    for every $x, y\in \FB(A)$.
\end{definition}

The following lemma gives a recursive way of calculating the values of $K_{0}(x, y)$ and $K_{1}(x,y)$ for any $x, y\in \FB(A)$.
\begin{lemma}\label{lem:kxy_computation}
    Let $x, y\in \FB(A)$.
    Then:
    \begin{align*}
    K_{0}(x, y) &=
        \begin{cases}
                1 & \text{if } y\ast 0\neq \bot \text{ and } y\ast 0\not \in \cont(x)\\
        1+K_{0}(x,y\circ 0) & \text{if } y\ast 0\in\cont(x)\text{ and }K_{0}(x, y\circ0)\neq \bot\\
        \bot & \text{otherwise,}
        \end{cases}\\
    K_{1}(x, y) &=
        \begin{cases}
                1 & \text{if } x\ast 1\neq \bot\text{ and } x\ast 1\not \in \cont(y)\\
        1+K_{1}(x\circ 1,y) & \text{if } x\ast 1\in\cont(y)\text{ and }K_{1}(x\circ 1, y)\neq \bot\\
        \bot & \text{otherwise.}
        \end{cases}
    \end{align*}
\end{lemma}
\begin{proof}
We will only prove the lemma for $K_{0}$; the proof for $K_{1}$ is analogous.

If $y\ast 0\neq \bot$ and $y\ast 0\not \in \cont(x)$, then, by definition, $K_{0}(x, y) = 1$ since there is no positive integer less than $1$.

If $y\ast 0\neq \bot$ and $y\ast 0\in\cont(x)$, then $K_{0}(x, y) = \bot$ or $K_{0}(x, y)\geq 2$.
If $K_{0}(x, y\circ 0)\neq \bot$, then $K_{0}(x, y\circ 0)$ is the least positive integer such that $(y\circ 0)\ast 0^{K_{0}(x, y\circ 0)} \neq \bot$ and $(y\circ 0) \ast 0^{K_{0}(x, y\circ 0)}\not \in \cont(x)$.
But, by the definition of $\circ$ and $\ast$, $(y\circ 0)\ast 0^{K_{0}(x, y\circ 0)}= y\ast 0^{1+K_0(x,y\circ0)}$ and so
$k = 1+K_{0}(x, y\circ 0)$ is the least positive integer such that $y\ast 0^{k}\neq \bot$ and $y\ast 0^{k}\not\in\cont(x)$, as required.
If $K_{0}(x, y\circ 0) =\bot$, then there is no $k\geq 1$ such that $y\ast 0^{1+k}\neq \bot$ and $y\ast 0^{1+k}\not\in\cont(x)$. This implies that there is no $k\geq 2$ such that $y\ast 0^{k}\neq \bot$ and $y\ast 0^{k}\not\in\cont(x)$. But $K_{0}(x, y) = \bot$ or $K_{0}(x, y) \geq 2$, and since the latter does not hold, we conclude that $K_{0}(x, y) = \bot$.

Finally, if $y\ast 0= \bot$, then $y\ast 0^{k}=\bot$ for all $k\geq 1$, so by definition $K_{0}(x, y)=\bot$.
\end{proof}

With the preceding lemma in mind, we now consider $((x\circ 1^i) \cdot (y\circ 0^j))\circ \alpha$
and $((x\circ 1^i) \cdot (y\circ 0^j))\ast \alpha$ for $\alpha \in \bool$.

\begin{lemma}\label{lem:multiplication:iterative}
    Let $x, y\in \FB(A)$ and let $0\leq i\leq |\cont(x)|$ and $0\leq j \leq |\cont(y)|$.
    If $i\neq |\cont(x)|$ or $j\neq |\cont(y)|$, then
        \[
    \begin{array}{ll}
        ((x\circ 1^i)\cdot (y\circ 0^j)) \circ 0 = \begin{cases}
            (x\circ 1^i)\cdot (y\circ 0^{j+k})\\
            x\circ 1^i0
        \end{cases}
        &
        \begin{array}{l}
             \textrm{if } k = K_{0}(x\circ 1^i, y\circ 0^j)\neq \bot\\
             \textrm{otherwise,}
        \end{array}
            \\
        ((x\circ 1^i)\cdot (y\circ 0^j)) \ast 0 = \begin{cases}
            y\ast 0^{j+k} \\
            x\ast 1^i0
        \end{cases}
        &
        \begin{array}{l}
        \textrm{if } k=K_{0}(x \circ 1 ^ i, y \circ 0 ^ j)\neq\bot\\
        \textrm{otherwise,}
        \end{array}
        \\
        ((x\circ 1^i)\cdot (y\circ 0^j)) \circ 1 = \begin{cases}
            (x\circ 1^{i+k})\cdot (y\circ 0^j) \\
            y\circ 0^j1
        \end{cases}
        &
        \begin{array}{l}
        \textrm{if } k=K_{1}(x \circ 1 ^ i, y \circ 0 ^ j)\neq \bot\\
        \textrm{otherwise,}
        \end{array}
        \\
        ((x\circ 1^i)\cdot (y\circ 0^j)) \ast 1 = \begin{cases}
            x\ast 1^{i+k}  \\
            y\ast 0^j1
        \end{cases}
        &
        \begin{array}{l}
        \textrm{if } k=K_{1}(x \circ 1 ^ i, y \circ 0 ^ j)\neq \bot\\
        \textrm{otherwise.}
        \end{array}
    \end{array}
    \]
        On the other hand, if $i = |\cont(x)|$ and $j = |\cont(y)|$, then
    $((x\circ 1^i)\cdot (y\circ 0^j)) \circ \alpha = \bot$ and $((x\circ 1^i)\cdot (y\circ 0^j)) \ast \alpha = \bot$ for $\alpha\in\bool$.
\end{lemma}
\begin{proof}
   We will show that if the result holds for $(x\circ 1^{i+1})\cdot (y\circ 0^j)$ and $(x\circ 1^i)\cdot (y\circ 0^{j+1})$, then it holds for $(x\circ 1^i)\cdot (y\circ 0^j)$.
    For the base case, let $i = |\cont(x)|$ or $j=|\cont(y)|$.

    If $i = |\cont(x)|$ and $j = |\cont(y)|$, then, since every application of $\circ$ reduces the content of $x$ by $1$,
    $|\cont(x\circ 1^{i})| = |\cont(x)| - i = 0$ and so $x\circ 1^i = \varepsilon$. Similarly $y\circ 0^j = \varepsilon$. It follows that
    $(x\circ 1^i)\cdot(y\circ0^j) = \varepsilon$ and so
    $((x\circ 1^i)\cdot(y\circ0^j))\circ \alpha$ and
    $((x\circ 1^i)\cdot(y\circ0^j))\ast \alpha$ are undefined for $\alpha \in \bool$.

    If $i = |\cont(x)|$ and $1\leq j < |\cont(y)|$.
    This implies that $x\ast 1^{i+k} = (x\circ 1^i)\ast 0^k = \varepsilon\ast 0^k$ is undefined for all $k\geq 1$.
    So $K_{1}(x\circ 1^i, y\circ 0^j)$ is undefined and
    \begin{align*}
        ((x\circ 1^i)\cdot(y\circ0^j))\circ1 & = \varepsilon\cdot (y\circ 0^j) \circ 1 = (y\circ 0^j)\circ 1=  y\circ 0^j1 \\
        ((x\circ 1^i)\cdot(y\circ0^j))\ast1 & = \varepsilon\cdot (y\circ 0^j) \ast 1 = (y\circ 0^j)\ast 1 = y\ast 0^j1.
        \end{align*}
    Since $j<|\cont(y)|$, $y\circ 0^j \neq \varepsilon$ so that $y\circ 0^{j+1}$ is defined and $y\circ 0^{j+1}\not\in \cont(x\circ 1^i) = \varnothing$.
    Hence by \cref{lem:kxy_computation}, $K_{0}(x\circ 1^i, y\circ 0^j)=1$. In this case, since $x\circ 1^i = \varepsilon$,
    \begin{align*}
        ((x\circ 1^i)\cdot(y\circ0^j))\circ0 = (x\circ 1^i)\cdot(y\circ 0^{j+1})\text{ and }
        ((x\circ 1^i)\cdot(y\circ0^j))\ast0 = y\ast 0^{j+1},
    \end{align*}
    as required.
    We have established the base case when $i=|\cont(x)|$, the proof when $j=|\cont(y)|$ is  analogously.

    Suppose that $0\leq i <|\cont(x)|$, $0\leq j<|\cont(y)|$ and assume that the conclusion of the lemma holds for
    $(x\circ 1^{i+1})\cdot (y\circ 0^j)$
    and $(x\circ 1)\cdot (y\circ 0^{j+1})$. Since $i<|\cont(x)|$ and $j<|\cont(y)|$, $x\circ 1^i, y\circ 0^j\neq \varepsilon$ and we can apply \cref{lem:multiplication} to $(x\circ 1^i)\cdot (y\circ 0^j)$.
    The value of $y\ast0^{j+1}$ is defined since $y\circ 0^j\neq \varepsilon$. If $y\ast0^{j+1}\not \in \cont(x\circ 1^i)$, then $K_{0}(x\circ 1^i, y\circ 0^j) = 1$ by \cref{lem:kxy_computation}. But then \cref{lem:multiplication} implies that
    \begin{align*}
    ((x\circ 1^i)\cdot (y\circ 0^j))\circ 0 &= (x\circ 1^i)\cdot ((y\circ 0^j)\circ 0) = (x\circ 1^i)\cdot (y\circ 0^{j+1})\\
    ((x\circ 1^i)\cdot (y\circ 0^j))\ast 0 &= (y\circ 0^j)\ast 0 = y\ast 0^{j+1},
    \end{align*}
    as required.
    If $y\ast 0^{j+1}\in \cont(x\circ 1^i)$, then, by \cref{lem:multiplication},
    \begin{align*}
    ((x\circ 1^i)\cdot (y\circ 0^j))\circ 0 &= ((x\circ 1^i)\cdot ((y\circ 0^j)\circ 0))\circ 0 = ((x\circ 1^i)\cdot (y\circ 0^{j+1}))\circ 0\\
    ((x\circ 1^i)\cdot (y\circ 0^j))\circ 0 &= ((x\circ 1^i)\cdot ((y\circ 0^j)\circ 0))\ast 0 = ((x\circ 1^i)\cdot (y\circ 0^{j+1}))\ast 0
    \end{align*}
    and so by the inductive hypothesis on $(x\circ 1^i)\cdot (y\circ 0^{j+1})$ the following holds
    \[
    \begin{array}{ll}
        ((x\circ 1^i)\cdot (y\circ 0^j)) \circ 0 = \begin{cases}
            (x\circ 1^i)\cdot (y\circ 0^{j+1+k})
            \\
            x\circ 1^i0
        \end{cases}
        &
        \begin{array}{l}
             \textrm{if } k=K_{0}(x\circ 1^i, y\circ 0^{j+1})\neq\bot \\
             \textrm{otherwise,}
        \end{array}
        \\
        ((x\circ 1^i)\cdot (y\circ 0^j)) \ast 0 = \begin{cases}
            y\ast 0^{j+1+k}\\
            x\ast 1^i0
        \end{cases}
        &
        \begin{array}{l}
             \textrm{if } k=K_{0}(x\circ 1^i, y\circ 0^{j+1})\neq \bot\\
             \textrm{otherwise.}
        \end{array}
    \end{array}
    \]
    \cref{lem:kxy_computation} implies that $K_{0}(x\circ 1^i, y\circ 0^j) = 1 + K_{0}(x\circ 1^i, y\circ0^{j+1})$ if $K_{0}(x\circ 1^i,y\circ 0^{j+1})$ is defined and $K_{0}(x\circ 1^i, y\circ0^j)$ is undefined otherwise. Therefore
        \[
    \begin{array}{ll}
    ((x\circ 1^i)\cdot (y\circ 0^j)) \circ 0 = \begin{cases}
            (x\circ 1^i)\cdot (y\circ 0^{j+k})
            \\
            x\circ 1^i0
        \end{cases}
        &
        \begin{array}{l}
             \textrm{if } k=K_{0}(x\circ 1^i, y\circ 0^j)\neq \bot \\
             \textrm{otherwise,}
        \end{array}
        \\
        ((x\circ 1^i)\cdot (y\circ 0^j)) \ast 0 = \begin{cases}
            y\ast 0^{j+k}
            \\
            x\ast 1^i0
        \end{cases}
        &
        \begin{array}{l}
        \textrm{if } k=K_{0}(x\circ 1^i, y\circ 0^j)\neq \bot\\
        \textrm{otherwise}
        \end{array}
        \end{array}
        \]
        which is exactly what we set out to show.

    An analogous proof establishes the conclusion of the lemma for
        $((x\circ 1^i)\cdot (y\circ 0^j)) \circ 1$ and
        $((x\circ 1^i)\cdot (y\circ 0^j)) \ast 1$, completing the induction.
\end{proof}

\cref{lem:multiplication:iterative} implies that for every $\alpha \in\bool$ one of the following holds: $((x\circ 1^i)\cdot(y\circ 0^j))\circ \alpha = (x\circ 1^{i^\prime})\cdot(y\circ 0^{j^\prime})$ for some $i^\prime, j^\prime$; $((x\circ 1^i)\cdot(y\circ 0^j))\circ \alpha = x\circ \beta$ for some $\beta \in \bool^\ast$; or $((x\circ 1^i)\cdot(y\circ 0^j))\circ \alpha = y\circ \gamma$ for some $\gamma\in\bool^\ast$.
By iteratively applying \cref{lem:multiplication:iterative}, we extend this observation in the following corollary.
\begin{corollary}\label{cor:multiplication_iterative}
 If $(x\cdot y)\circ \alpha\neq\bot$ for some $\alpha\in \bool ^ \ast$, then
\begin{multline*}
(x\cdot y)\circ \alpha
\in\{(x\circ 1^i) \cdot (y\circ 0^j) : 0\leq i\leq |\cont(x)|, 0\leq j \leq |\cont(y)|\}\\
\cup\{x \circ \beta : \beta \in \bool^\ast\}
\cup\{y \circ \gamma : \gamma \in \bool^\ast\}.\qed
\end{multline*}
\end{corollary}

\cref{lem:multiplication:iterative} and \cref{cor:multiplication_iterative} allow us to describe a transducer representing the product $x\cdot y$ given transducers for $x$ and $y$ as follows. Let $x, y\in\FB(A)$ and let
$\mathcal{T}_x = \left(Q_x, \{0, 1\}, A, q_x, T_x, \spadesuit_x, \clubsuit_x\right)$ and
$\mathcal{T}_y = \left(Q_y, \{0, 1\}, A, q_y, T_y, \spadesuit_y, \clubsuit_y\right)$ be any transducers representing $x$ and $y$, respectively. We define:
\begin{equation}\label{eq:product-transducer}
\mathcal{T}_{x\cdot y} = \left(Q, \{0, 1\}, A, q_0, T, \spadesuit, \clubsuit\right)
\end{equation}
where:

\begin{enumerate}[label=(\roman*)]
    \item the set of states $Q=Q_x\cup Q_y \cup Q^\prime$ consists of all the states of $\mathcal{T}_x$ and $\mathcal{T}_y$ together with the states $Q^\prime = \{(i, j): 0\leq i\leq |\cont(x)|\ \textrm{and}\ 0\leq j \leq |\cont(y)|\}$,
    where $\{(|\cont(x)|, j) : 0\leq j\leq |\cont(y)|\}$ and $\{(i, |\cont(y)|) : 1\leq i\leq |\cont(x)|\}$ are identified with the original states $\{q_y\spadesuit_y 0^j: 0\leq j\leq |\cont(y)|\}$ and $\{q_x\spadesuit_x 1 ^ i : 0\leq i \leq |\cont(x)|\}$. Each state in $(i, j)\in Q^\prime$ corresponds to the element $(x\circ 1^i) \cdot (y\circ 0^j)$ in \cref{lem:multiplication:iterative};
    \item the initial state is $q_0 = (0, 0)$;
    \item the terminal states are exactly those of $\mathcal{T}_x$ and $\mathcal{T}_y$, i.e. $T = T_x\cup T_y$;
    \item the state and letter transition functions $\spadesuit$ and $\clubsuit$ for the states in $Q_x$ and $Q_y$ are defined exactly as their counterparts in $\mathcal{T}_x$ and $\mathcal{T}_y$;
    \item if $(i,j)\in Q^\prime$, then the state transition and letter transition functions are given by
        \[
    \begin{array}{ll}
        (i, j) \spadesuit 0 = \begin{cases}
            (i, j+k)\\
            q_x\spadesuit_x 1^i0 &
        \end{cases}
        &
        \begin{array}{l}
            \textrm{if } k=K_{0}(x\circ 1^i, y\circ 0^j)\neq\bot\\
            \textrm{otherwise,}
        \end{array}
        \\
        (i, j) \clubsuit 0 = \begin{cases}
            q_y\clubsuit_y 0^{j+k}\\
            q_x\clubsuit_x 1^i0
        \end{cases}
         &
        \begin{array}{l}
        \textrm{if } k=K_{0}(x\circ 1^i, y\circ 0^j)\neq\bot\\
        \textrm{otherwise,}
        \end{array}
        \\
        (i, j) \spadesuit 1 = \begin{cases}
            (i+k, j)\\
        q_y\spadesuit_y 0^j1
        \end{cases}
        &
        \begin{array}{l}
        \textrm{if } k=K_{1}(x\circ 1^i, y\circ 0^j)\neq \bot\\
        \textrm{otherwise,}
        \end{array}
        \\
        (i, j) \clubsuit 1 = \begin{cases}
            q_x\clubsuit_x 1^{i+k}\\
            q_y\clubsuit_y 0^j1
        \end{cases}
        &
        \begin{array}{l}
        \textrm{if } k=K_{1}(x\circ 1^i, y\circ 0^j)\neq \bot\\
        \textrm{otherwise}.
        \end{array}
    \end{array}
    \]
    \end{enumerate}

The next lemma establishes that $\mathcal{T}_{x\cdot y}$ is indeed a transducer.

\begin{lemma}
If $x, y\in \FB(A)$ are arbitrary, and $\mathcal{T}_x$ and $\mathcal{T}_y$ are transducers representing $x$ and $y$ respectively, then $\mathcal{T}_{x\cdot y}$ defined in \eqref{eq:product-transducer} is a transducer.
\end{lemma}
\begin{proof}
   It suffices to show that  the state transition function $\spadesuit$ only takes values in $Q$ and that $\spadesuit$ and $\clubsuit$ are well-defined.

   Clearly, if $(i, j)\spadesuit \alpha\in Q_x \cup Q_y$, then $(i, j)\spadesuit\alpha \in Q$. On the other hand, if $(i, j)\spadesuit 0 = (i, j + k)$ where $k = K_0(x\circ 1 ^ i, y \circ 0 ^ j) \neq \bot$,
  then
$y\ast0^{j+k}\neq \bot$ so $j + k\leq |\cont(y)|$ and $(i, j + k) \in Q'$. Similarly, if $(i, j)\spadesuit 1 = (i + k, j)$, then it can be shown that $i + k \leq |\cont(x)|$ and so $(i, j)\spadesuit 1 \in Q'$ also.

To show that $\spadesuit$  and $\clubsuit$ are well-defined, it suffices to show that
 for all $\alpha \in \bool$, $(|\cont(x)|, j) \spadesuit \alpha = (q_y\spadesuit_y 0^j)\spadesuit \alpha$, $(|\cont(x)|, j) \clubsuit \alpha = (q_y\spadesuit_y 0^j)\clubsuit \alpha$ and the analogous statements hold for $(i, |\cont(y)|)$ and $q_x\spadesuit_x 1 ^ i$.

    If $i = |\cont(x)|$ and $1\leq j\leq|\cont(y)|$,
    then $x\ast 1^i = \varepsilon$ and so
    $x\ast 1^{i+k} = (x\circ 1^i)\ast 0^k= \bot$ for all $k\geq 1$.
    But then
    \begin{align*}
        (|\cont(x)|, j)\spadesuit 1  & = q_y\spadesuit_y 0^j1 = (q_y\spadesuit_y 0^j)\spadesuit 1\\
        (|\cont(x)|, j)\clubsuit 1 & = q_y\clubsuit_y 0^j1 = (q_y\spadesuit_y 0^j)\clubsuit 1.
    \end{align*}
    If $y\circ 0^{j+1}\not=\bot$, then $y\circ 0^{j+1}\not\in \cont(x\circ 1^i) = \varnothing$ and so
    $K_{0}(x\circ 1^i, y\circ 0^j)=1$. In this case,
    \begin{align*}
        (|\cont(x)|, j)\spadesuit0 &= (|\cont(x)|, j+1) = q_y\spadesuit_y0^{j+1} = (q_y\spadesuit_y 0^j)\spadesuit 0 \\
        (|\cont(x)|, j)\clubsuit 0 &= q_y\clubsuit 0^{j+1} = (q_y\spadesuit_y 0^j)\spadesuit 0.
    \end{align*}
    If, on the other hand, $y\circ 0^{j+1}= \bot$, then $y\circ 0^{j+k}=\bot$ for all $k\geq 1$. Therefore
    \begin{align*}
        (|\cont(x)|, j)\spadesuit0 &= q_x\spadesuit_x 1^{|\cont(x)|}0 = (q_x\spadesuit_x1^{|\cont(x)|}) \spadesuit 0 \\
        (|\cont(x)|, j)\clubsuit0 &= q_x\clubsuit_x 1^{|\cont(x)|}0 = (q_x\spadesuit_x1^{|\cont(x)|}) \clubsuit 0.
    \end{align*}
    Since $q_x\spadesuit_x 1^{|\cont(x)|}$ represents $x\circ 1^{|\cont(x)|} =\varepsilon$, and $\mathcal{T}_x$ represents $x$, it follows that $q_x\spadesuit_x 1^{|\cont(x)|}$ is terminal and so $(|\cont(x)|, j)\spadesuit 0 = \bot$ and $(|\cont(x)|, j)\clubsuit 0 = \bot$. But this is exactly what we wanted, since $y\circ0^{j+1}$ is undefined, $y\circ 0^j =\varepsilon$, so $q_y\spadesuit_y 0^j$ is terminal and therefore $(q_y\spadesuit 0^j)\spadesuit 0=\bot$
    $(q_y\spadesuit 0^j)\spadesuit 1=\bot$ as well.
\end{proof}

To show that $\mathcal{T}_{x\cdot y}$ does indeed represent $x\cdot y$ we require the following lemma.
\begin{lemma}\label{lem:transducer-psi}
 Let $\mathcal{T} = (Q, \bool, A, q_0, T, \spadesuit, \clubsuit)$ be a transducer and let $\Psi: Q\rightarrow \FB(A)$ be any partial function such that $\Psi(q_0) \neq \bot$; $\Psi(q) = \varepsilon$ if and only if $q\in T$; and $\Psi$ satisfies the commutative diagram in \cref{fig:induced_state_function}. Then $\mathcal{T}$ represents $\Psi(q_0)$.
\end{lemma}
\begin{proof}
    We show that for every $q\in Q$ with $\Psi(q)\neq \bot$, the induced subtransducer rooted at $q$ realizes the function $f_{\Psi(q)}$. Since $\Psi(q_0)\neq \bot$ and the induced subtransducer rooted at $q_0$ realizes the same function as $\mathcal{T}$, this will imply that $\mathcal{T}$ realizes $f_{\Psi(q_0)}$ and so represents $\Psi(q_0)$, as required.
    Note that for every state $q\in Q$, if $\Psi(q)\neq \bot$, then either: $\Psi(q) = \varepsilon$, and $q\spadesuit 0 = q\spadesuit 1 = \bot$; or $\Psi(q) \neq \varepsilon$, $q\spadesuit 0\neq \bot$, and $q\spadesuit 1\neq \bot$. This is because $\Psi$ respects the commutative diagram in \cref{fig:induced_state_function}, and the only element of $\FB(A)$ with $x\circ 0 = x\circ 1 = \bot$ is $\varepsilon$.

    We proceed by induction on $|\cont(\Psi(q))|$.
    For the base case, let $q\in Q$ be such that $\Psi(q)\neq \bot$ and $|\cont(\Psi(q))| = 0$. Then $\Psi(q) = \varepsilon$ and
    since $\Psi$ respects the diagram in \cref{fig:induced_state_function}, it follows that $q\spadesuit \alpha = \varepsilon\spadesuit \alpha = \bot$ for all $\alpha \in \bool$. By assumption, $\Psi(q)=\varepsilon$ if and only if $q$ is terminal. It follows that the induced subtransducer rooted at $q$ realizes the function which maps $\varepsilon$ to itself and every other input to $\bot$, which is exactly $f_\varepsilon$.

    Now fix $1\leq k\leq |\cont(\Psi(q_0))|$ and assume that the statement holds for all $q\in Q$ with $\Psi(q)\neq \bot$ and $|\cont(\Psi(q))| = k-1$.
    Let $q\in Q$ with $\Psi(q)\neq \bot$ be such that $|\cont(\Psi(q))| = k$. Since $k\geq 1$, $\Psi(q)\neq \varepsilon$ and so $q\spadesuit \alpha\neq \bot$ for all $\alpha\in\bool$. But then $|\cont(\Psi(q\spadesuit \alpha))| = |\cont(\Psi(q)\circ \alpha)| = |\cont(\Psi(q))| - 1$ for all $\alpha\in\bool$. It follows that the induced subtransducer rooted at $q\spadesuit \alpha$ realizes $f_{\Psi(q\spadesuit \alpha)} = f_{\Psi(q)\circ \alpha}$ for all $\alpha \in \bool$.
    The subtransducer rooted at $q$ realizes the partial function $f: \bool^\ast\rightarrow A^\ast$ given by $f(\varepsilon) = \bot$ and $f(\alpha\beta) = (q\clubsuit \alpha) f_{\Psi(q\spadesuit \alpha)}(\beta)$ for all $\alpha\in \bool$ and $\beta\in \bool^\ast$. But this is exactly the partial function $f_{\Psi(q)}$.
\end{proof}

   We now prove that the transducer $\mathcal{T}_{x\cdot y}$ represents $x\cdot y$ as claimed.

\begin{theorem}
Let $x, y\in\FB(A)$ and let $\mathcal{T}_x$ and $\mathcal{T}_y$ be transducers representing $x$ and $y$, respectively. Then the transducer $\mathcal{T}_{x\cdot y}$ defined in \eqref{eq:product-transducer} represents $x \cdot y$.
\end{theorem}
\begin{proof}
Suppose that
$\mathcal{T}_x = \left(Q_x, \{0, 1\}, A, q_x, T_x, \spadesuit_x, \clubsuit_x\right)$, that
$\mathcal{T}_y = \left(Q_y, \{0, 1\}, A, q_y, T_y, \spadesuit_y, \clubsuit_y\right)$, and that
$\mathcal{T}_{x\cdot y} = \left(Q, \{0, 1\}, A, q_0, T, \spadesuit, \clubsuit\right)$.

To show that $\mathcal{T}_{x\cdot y}$ represents $x\cdot y$ we will define a function $\Psi : Q\rightarrow \FB(A)$ such that $\Psi(q_0) = x\cdot y$, $\Psi$ satisfies the commutative diagrams given in \cref{fig:induced_state_function}, and $\Psi(q) = \varepsilon$ if and only if  $q\in T$. It will follow from this that $\mathcal{T}_{x\cdot y}$ realizes $f_{x\cdot y}$ and therefore $\mathcal{T}_{x\cdot y}$ represents $x\cdot y$ by  \cref{lem:transducer-psi}.

Let $\Psi_x : Q_x\rightarrow \FB(A)$ and $\Psi_y: Q_y\rightarrow \FB(A)$ be the functions realizing the commutative diagram in \cref{fig:induced_state_function} for $\mathcal{T}_x$ and $\mathcal{T}_y$ respectively. We define $\Psi: Q\rightarrow \FB(A)$ by
\[\Psi(q) = \begin{cases}
    (x\circ 1^i) \cdot (y \circ 0^j) & \textrm{ if }q=(i,j) \in Q^\prime\\
    \Psi_x(q) & \textrm{ if }q\in Q_x\\
    \Psi_y(q) & \textrm{ if }q\in Q_y.
\end{cases}\]

Since we identified the states $\{(|\cont(x)|, j) : 0\leq j\leq |\cont(y)|\}$ and $\{(i, |\cont(y)|) : 1\leq i\leq |\cont(x)\}$ with the states $\{q_y\spadesuit_y 0^j: 0\leq j\leq |\cont(y)\}$ and $\{q_x\spadesuit_x 1 ^ i : 0\leq i \leq |\cont(x)|\}$ within $Q$, we must check that $\Psi$ is well-defined. In particular, we must show that $\Psi((i, |\cont(y)|)) = \Psi(q_x\spadesuit 1^i)$ and $\Psi((|\cont(x)|, j)) = \Psi(q_y\spadesuit 0^j)$ for all relevant values of $i$ and $j$.
Every application of $\circ$ reduces the content of $y$ by $1$ letter, and so $|\cont(y\circ 0^{|\cont(y)|})| = |\cont(y)|-|\cont(y)| = 0$. In other words, $y\circ 0^{|\cont(y)|} = \varepsilon$. But then
$\Psi((i, |\cont(y)|)) = (x\circ 1^i) \cdot (y \circ 0^{|\cont(y)|}) = (x\circ 1^i) \cdot \varepsilon = x\circ 1^i$ for all $i$.
On the other hand, since $\mathcal{T}_x$ represents $x$,
$\Psi(q_x\spadesuit 1^i) = \Psi_x(q_x\spadesuit 1^i) = \Psi_x(q_x)\circ 1^i = x\circ 1^i$ for all $i$.
Therefore $\Psi((i, |\cont(y)|)) = \Psi(q_x\spadesuit 1^i)$ for all $i$. Similarly, $\Psi((|\cont(x)|, j)) = \Psi(q_y\spadesuit 0^j)$ for all $j$, and so $\Psi$ is well-defined.

Next, we establish that the 3 conditions of \cref{lem:transducer-psi} hold for $\Psi$.
By the definition of $\Psi$, $\Psi(q_0) = \Psi((0, 0)) = (x\circ 1^0) \cdot (y\circ 0^0) = x\cdot y \neq \bot$.

We now show that
$\Psi(q) = \varepsilon$ if and only if $q\in T = T_x\cup T_y$.
Since, for any $q\in Q_x$, $\Psi_x(q) = \varepsilon$ if and only if $q\in T_x$, and the analogously for $\Psi_y$, it follows that, for all $q\in Q_x \cup Q_y$,
$\Psi(q) = \varepsilon$ if and only if $q\in T$.
Let $q = (i, j)\in Q^\prime$. Then $\cont(\Psi(q)) = \cont((x\circ 1^i)\cdot (y\circ 0^j)) = \cont(x\circ 1^i)\cup \cont(y\circ 0^j)$. If $\Psi(q) = \varepsilon$, then $|\cont(x\circ 1^i)| = |\cont(y\circ 0^j)| = 0$. But $|\cont(x\circ 1^i)| = |\cont(x)| - i$, and so $i = |\cont(x)|$. Similarly $j = |\cont(y)|$. Therefore $q$ is the state that was identified with both $q_x\spadesuit 1^{|\cont(x)|}$  and $q_y\spadesuit 0 ^ {|\cont(y)|}$. The former state is terminal in $\mathcal{T}_x$ since $\Psi_x(q_x\spadesuit 1^{|\cont(x)|}) = \Psi_x(q_x)\circ 1^{|\cont(x)|} = x\circ 1^{|\cont(x)|} = \varepsilon$.

It remains to show that
$\Psi$ satisfies the commutative diagram in \cref{fig:induced_state_function}.
It suffices to show that
$\Psi(q\spadesuit \alpha) = \Psi(q)\circ \alpha$ and $q\clubsuit\alpha = \Psi(q)\ast \alpha$ for all $q\in Q$ and $\alpha \in \bool$, since the recursive definitions of $\spadesuit, \clubsuit, \circ$ and $\ast$ imply the result for all $\alpha\in \bool^\ast$.
If $q\in Q_x$, then $q\spadesuit \alpha =q\spadesuit_x \alpha\in Q_x$ and $q\clubsuit \alpha = q\clubsuit_x \alpha\in A$ for all $\alpha\in\bool$ since $\spadesuit$ and $\clubsuit$ coincide with $\spadesuit_x$ and $\clubsuit_x$ on $Q_x$. This implies that \[\Psi(q\spadesuit \alpha) = \Psi(q\spadesuit_x \alpha) =\Psi_x(q\spadesuit_x \alpha) = \Psi_x(q)\circ \alpha = \Psi(q)\circ \alpha,\]
and
\[
q\clubsuit \alpha = q\clubsuit_x \alpha  = \Psi_x(q)\ast \alpha = \Psi(q)\ast \alpha
\]
for all $\alpha\in\bool$. Hence the equalities represented in the commutative diagram are satisfied for $q\in Q_x$. A similar argument shows that the commutative diagram holds for $q\in Q_y$ also.

Suppose that $q = (i, j)\in Q^\prime$. If $k = K_{0}(x\circ 1^i, y\circ 0^j)\neq \bot$, then by applying \cref{lem:multiplication:iterative}
\begin{align*}
    \Psi((i, j)\spadesuit 0) &= \Psi((i, j+k)) = (x\circ 1^i)\cdot (y\circ 0^{j+k})
    = ((x\circ1^i)\cdot (y\circ 0^j))\circ 0 = \Psi(q)\circ 0\\
    (i, j)\clubsuit 0 &= q_y\clubsuit 0^{j+k} = \Psi_y(q_y)\ast 0^{j+k} = y\ast 0^{j + k} =
     ((x\circ1^i)\cdot (y\circ 0^j))\ast 0= \Psi(q)\ast 0.
\end{align*}
On the other hand, if $K_{0}(x\circ 1^i, y\circ 0^j) = \bot$, then
\begin{align*}
    \Psi((i, j)\spadesuit 0) &= \Psi(q_x\spadesuit 1^i0) =\Psi_x(q_x\spadesuit 1^i0)= \Psi_x(q_x)\circ 1^i0
    = x\circ 1^i0 = ((x\circ1^i)\cdot (y\circ 0^j))\circ 0= \Psi(q)\circ 0\\
    (i, j)\clubsuit 0 &= q_x\clubsuit 1^i0 = \Psi_x(q_x)\ast 1^i0 = x\ast 1^i0
    = ((x\circ1^i)\cdot (y\circ 0^j))\ast 0= \Psi(q)\ast 0.
\end{align*}
Hence $\Psi(q\spadesuit 0) = \Psi(q)\circ 0$ and $q\clubsuit 0 = \Psi(q)\ast 0$ for all $q\in Q^\prime$. Similarly, it is routine to verify that $\Psi(q\spadesuit 1) = \Psi(q)\circ 1$ and $q\clubsuit 1 = \Psi(q)\ast 1$ for all $q\in Q^\prime$.
\end{proof}

We will in \cref{algorithm:multiply} turn the definition in \eqref{eq:product-transducer} into an algorithm for computing $\mathcal{T}_{x\cdot y}$ from $\mathcal{T}_x$ and $\mathcal{T}_y$.
\cref{algorithm:multiply}
involves computing $(i, j)\spadesuit \alpha$ and $(i, j)\clubsuit \alpha$ for $\alpha\in\bool$, which in turn require us to calculate $K_{\alpha}(x\circ 1^i, y\circ 0^j)$.

We cannot compute $K_\alpha$ directly, since it is defined on pairs of elements in $\FB(A)$ and these have no representation here other than transducers.
In particular, computing transducers (or another representation of) the input values $x\circ 1 ^ i$ and  $y \circ 0 ^ j$ defeats the purpose of the algorithm presented in this section.
However, if $x, y\in \FB(A)$ are fixed, then \cref{lem:kxy_computation} tells us that
\begin{align*}
    K_{0}(x\circ 1^i, y\circ 0^j) &=
        \begin{cases}
                1 & \text{if } y\ast 0^{j+1}\neq \bot \text{ and } y\ast 0^{j+1}\not \in \cont(x\circ 1^i)\\
        1+K_{0}(x\circ 1^i,y\circ 0^{j+1}) & \text{if } y\ast 0^{j+1}\in\cont(x\circ 1^i)\text{ and }K_{0}(x\circ 1^i, y\circ0^{j+1})\neq \bot\\
        \bot & \text{otherwise}
        \end{cases}\\
    K_{1}(x\circ 1^i, y\circ 0^j) &=
        \begin{cases}
                1 & \text{if } x\ast 1^{i+1}\neq \bot\text{ and } x\ast 1^{i+1}\not \in \cont(y\circ 0^j)\\
        1+K_{1}(x\circ 1^{i+1},y\circ 0^j) & \text{if } x\ast 1^{i+1}\in\cont(y\circ 0^j)\text{ and }K_{1}(x\circ 1^{i+1}, y\circ 0^j)\neq \bot\\
        \bot & \text{otherwise.}
        \end{cases}
    \end{align*}
    We can effectively compute, as follows, the values $K_{\alpha}(x \circ 1 ^ i, y \circ 0 ^ j)$ for any $i, j\in \mathbb{N}_0$ from any transducers  $\mathcal{T}_x$ and $\mathcal{T}_y$
    representing $x$ and $y$, respectively.
If $q_x$ and $q_y$ are the initial states and $\clubsuit_x$ and $\clubsuit_y$ are the letter transition functions of $\mathcal{T}_x$ and $\mathcal{T}_y$, respectively, then $q_x\clubsuit_x 1^i = x\ast 1^i$ and $q_y\clubsuit_y 0^i = y\ast 0^i$ for all $i, j$, and so we can compute $K_{\alpha}(x \circ 1 ^ i, y \circ 0 ^ j)$ given $i, j\in \mathbb{N}_0$. For the remainder of this section we suppose that
 $x, y\in \FB(A)$ are fixed.  We define $\overline{K_\alpha}: \mathbb{N}_0\times\mathbb{N}_0\to \mathbb{N}$ by
\[
\overline{K_{\alpha}}(i,j) = K_{\alpha}(x \circ 1 ^ i, y \circ 0 ^ j).
\]
Computing $\overline{K_{\alpha}}(i, j)$ naively by considering all possible values $k\geq 1$ and checking if $y\ast 0^{j+k}$ is defined, for example, requires at least $\mathcal{O}(|A|)$ steps per $(i, j)\in Q'$. This would require a total of $\mathcal{O}(|A|^3)$ to compute $\overline{K_{\alpha}}(i, j)$ for all $(i,j)\in Q^\prime$.
However, by utilizing \cref{lem:kxy_computation} and traversing $(i, j)$ in reverse lexicographic order, we can precompute the values of $\overline{K_{\alpha}}(i, j)$ for all $(i, j)\in Q^\prime$ in $\mathcal{O}(|A|^2)$ time. Storing the values will take $\mathcal{O}(|A|^2)$ space and we can then retrieve each value in $\mathcal{O}(1)$ time.
This is exactly what we do to compute $\overline{K_0}$ in $\ComputeK_0$ given in \cref{algorithm:compute_k}.
A function $\ComputeK_1$ to compute $\overline{K_1}$ can be defined analogously.

    To illustrate, if $x = eaec$ and $y = bcacbcd$ are elements in $\FB(a, b, c, d, e)$, then
    the minimal transducers representing each $x$  and $y$ are shown in \cref{fig:xy-minimal-transducers} and the values of $\overline{K_0}$ and $\overline{K_1}$ are shown in \cref{table-k-0-1}.

    \begin{table}
    \begin{center}
    \begin{tabular}{cc}
    \begin{tabular}[t]{c|c|cccc}
             \multicolumn{2}{r|}{$i$}  & 0 & 1 & 2 & 3 \\ \hline
              \multicolumn{2}{r|}{$\cont(x\circ 1^i)$} & $\{a, e, c\}$ & $\{e, c\}$ & $\{c\}$ & $\varnothing$  \\
          \hline
          \hline
         0 & $d$    & 1 & 1 & 1 & 1\\
         1 & $a$    & 3 & 1 & 1 & 1\\
         2 & $c$    & 2 & 2 & 2 & 1\\
         3 & $b$    & 1 & 1 & 1 & 1\\
         4 & $\bot$ & $\bot$ & $\bot$ & $\bot$ & $\bot$\\
         \hline
         $j$ & $q_y\clubsuit0^{j+1}$ & & & &
    \end{tabular}&
    \begin{tabular}[t]{c|c|ccccc}
              \multicolumn{2}{c|}{$j$}  & 0 & 1 & 2 & 3 & 4 \\
              \hline
              \multicolumn{2}{c|}{$\cont(y\circ 0^j)$} &  $\{b, c, a, d\}$ & $\{b, c, a\}$ & $\{b, c\}$ & $\{b\}$ & $\varnothing$ \\
          \hline
          \hline
         0 & $a$    & 1 & 1 & 1 & 2 & 2\\
         1 & $e$    & 1 & 1 & 1 & 1 & 1\\
         2 & $c$    & 1 & 1 & $\bot$ & $\bot$ & $\bot$\\
         3 & $\bot$ & $\bot$ & $\bot$ & $\bot$ & $\bot$ & $\bot$\\
         \hline
         $i$ & $q_x\clubsuit1^{i+1}$ & & & & &
    \end{tabular}
    \end{tabular}
    \end{center}
    \caption{The functions $\overline{K_0}$ (left) and $\overline{K_1}$ (right) of $x = eaec, y = bcacbcd\in\FB(a, b, c, d, e)$. $\ComputeK_{0}$ and $\ComputeK_{1}$ fill in their respective tables column-by-column, starting from the right, filling each column from bottom to top.}
    \label{table-k-0-1}
    \end{table}

\begin{figure}
 \centering
  \begin{tikzpicture}
 \node (x_0) at (2,6) {$x_0$};
 \node (x_1) at (0,4) {$x_1$};
 \node (x_2) at (4,4) {$x_2$};
 \node (x_3) at (2,2) {$x_3$};
 \node (x_4) at (6,2) {$x_4$};
 \node (x_5) at (4,0) {$x_5$};

 \node (y_0) at (12,6) {$y_0$};
 \node (y_1) at (10,4) {$y_1$};
 \node (y_2) at (14,4) {$y_2$};
 \node (y_3) at (8,2) {$y_3$};  \node (y_4) at (12,2) {$y_4$};
 \node (y_5) at (16,2) {$y_5$};
 \node (y_6) at (8,0) {$y_6$};  \node (y_7) at (12,0) {$y_7$};  \node (y_8) at (16,0) {$y_8$};  \node (y_9) at (12,-2) {$y_9$};
 \draw [-latex] (x_0)++(0,0.5) -- (x_0);
 \draw [-latex] (x_5) -- ++(0,-0.5);
 \draw [-latex] (y_0)++(0,0.5) -- (y_0);
 \draw [-latex] (y_9) -- ++(0,-0.5);

 \draw [-latex, draw=color2] (x_0) -- node[near start, above, inner sep=6pt] {$0|c$} (x_1);
 \draw [-latex reversed, dashed, draw=color0] (x_0) -- node[near start, above, inner sep=6pt] {$1|a$} (x_2);

 \draw [parallel_arrow_1, draw=color0] (x_1) -- node[near start, left, inner sep=10pt] {$0|a$} (x_3);
 \draw [parallel_arrow_2, draw=color0] (x_1) -- node[near start, above, inner sep=10pt] {$1|a$} (x_3);

 \draw [-latex, draw=color2] (x_2) -- node[near start, above, inner sep=6pt] {$0|c$} (x_3);
 \draw [-latex reversed, dashed, draw=color4] (x_2) -- node[near start, above, inner sep=6pt] {$1|e$} (x_4);

 \draw [parallel_arrow_1, draw=color4] (x_3) -- node[near start, left, inner sep=10pt] {$0|e$} (x_5);
 \draw [parallel_arrow_2, draw=color4] (x_3) -- node[near start, above, inner sep=10pt] {$1|e$} (x_5);

 \draw [parallel_arrow_1, draw=color2] (x_4) -- node[near start, above, inner sep=10pt] {$0|c$} (x_5);
 \draw [parallel_arrow_2, draw=color2] (x_4) -- node[near start, right, inner sep=10pt] {$1|c$} (x_5);

 \draw [-latex, draw=color3] (y_0) -- node[near start, above, inner sep=6pt] {$0|d$} (y_1);
 \draw [-latex reversed, dashed, draw=color0] (y_0) -- node[near start, above, inner sep=6pt] {$1|a$} (y_2);

 \draw [-latex, draw=color0] (y_1) -- node[near start, above, inner sep=6pt] {$0|a$} (y_3);
 \draw [-latex reversed, dashed, draw=color0] (y_1) -- node[near start, above, inner sep=6pt] {$1|a$} (y_4);

 \draw [-latex, draw=color3] (y_2) -- node[near start, above, inner sep=6pt]
 {$0|d$} (y_4);
 \draw [-latex reversed, dashed, draw=color1] (y_2) -- node[near start, above, inner sep=6pt]
 {$1|b$} (y_5);

 \draw [-latex, draw=color2] (y_3) -- node[near start, left, inner sep=6pt]
 {$0|c$} (y_6);
 \draw [-latex reversed, dashed, draw=color1] (y_3) -- node[near start, above, inner sep=6pt]
 {$1|b$} (y_7);

 \draw [parallel_arrow_2, draw=color1] (y_4) -- node[near start, left, inner sep=6pt]
 {$0|b$} (y_7);
 \draw [parallel_arrow_1, draw=color1] (y_4) -- node[near start, right, inner sep=6pt]
 {$1|b$} (y_7);

 \draw [-latex, draw=color3] (y_5) -- node[near start, above, inner sep=6pt]
 {$0|d$} (y_7);
 \draw [-latex reversed, dashed, draw=color2] (y_5) -- node[near start, right, inner sep=6pt]
 {$1|c$} (y_8);

 \draw [parallel_arrow_2, draw=color1] (y_6) -- node[near start, above, inner sep=6pt]
 {$1|b$} (y_9);
 \draw [parallel_arrow_1, draw=color1] (y_6) -- node[near start, left, inner
  sep=16pt]
 {$0|b$} (y_9);

 \draw [parallel_arrow_2, draw=color2] (y_7) -- node[near start, right, inner sep=6pt]
 {$1|c$} (y_9);
 \draw [parallel_arrow_1, draw=color2] (y_7) -- node[near start, left, inner
  sep=6pt]
 {$0|c$} (y_9);

 \draw [parallel_arrow_2, draw=color3] (y_8) -- node[near start, right, inner
  sep=16pt] {$1|d$} (y_9);
 \draw [parallel_arrow_1, draw=color3] (y_8) -- node[near start, above, inner
  sep=6pt]
 {$0|d$} (y_9);

\end{tikzpicture}

 \caption{The minimal transducers for $x=eaec$ and $y = bcacbcd$.}
 \label{fig:xy-minimal-transducers}
\end{figure}
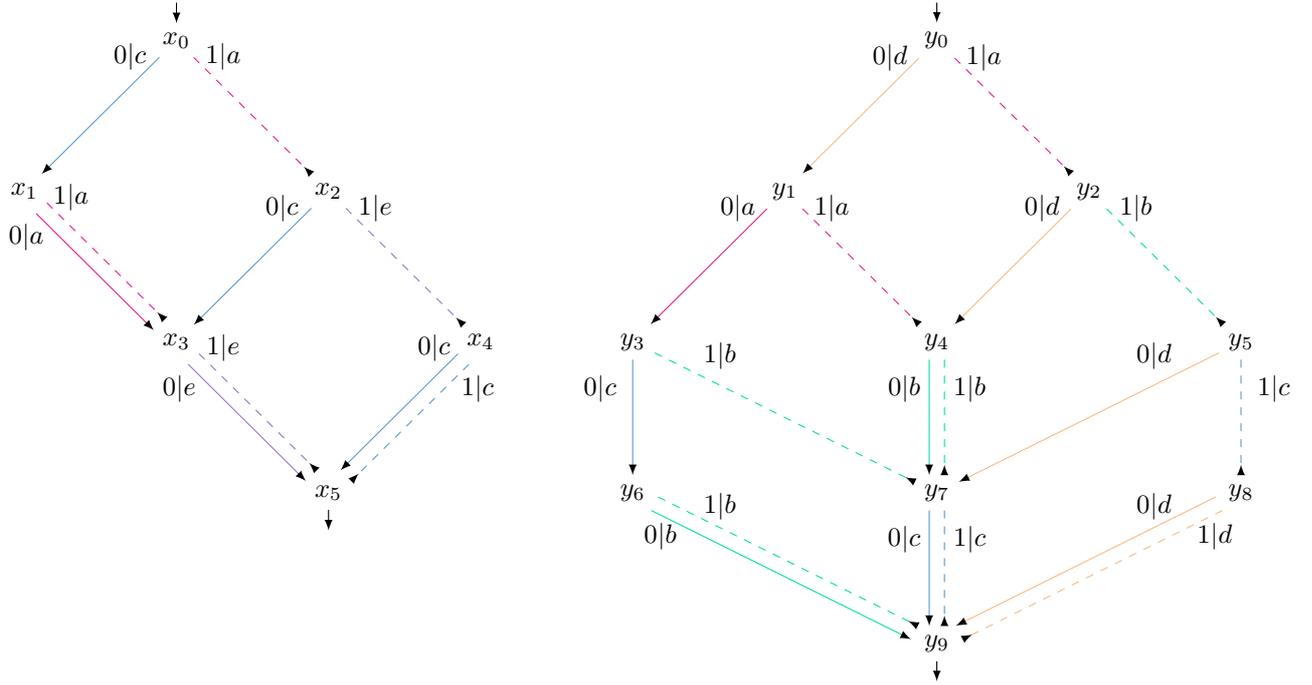
\begin{algorithm}[]
\SetAlgoLined\DontPrintSemicolon\KwArguments{a pair of transducers $\mathcal{T}_x = (Q_x, \bool, A, q_x, T_x, \spadesuit_x, \clubsuit_x)$ and $\mathcal{T}_y = (Q_y, \bool, A, q_y, T_y, \spadesuit_y, \clubsuit_y)$ representing some $x, y\in\FB(A)$ respectively.}

\KwReturns{the partial function $\overline{K_{0}}$.}
Let $\overline{K_0}:\mathbb{N}_0\times\mathbb{N}_0\rightarrow \mathbb{N}$ be such that $\overline{K_0}(i, j) = \bot$ for all $i, j$.\;
Let $c\gets \varnothing$ \Comment{The content of $x$ seen so far}
    \For{$i\in \{|\cont(x)|, \ldots, 0\}$}{    \For{$j\in \{|\cont(y)|, \ldots, 0\}$}{    \uIf{$q_y\clubsuit_y 0^{j+1}\neq \bot$ and $q_y\clubsuit_y 0^{j+1}\not\in c$}{        $\overline{K_0}(i, j) := 1$\;
    }
    \ElseIf{$q_y\clubsuit_y 0^{j+1}\neq \bot$ and $\overline{K_0}(i, j+1)\neq\bot$}{        $\overline{K_0}(i, j) := 1+\overline{K_0}(i, j+1)$\;
    }
    }
                \If{$i\neq 0$}{        $c\gets c\cup\{q_x\clubsuit_x 1^{i}\}$\;
    }
    }
\KwRet{$\overline{K_0}$}\;
\caption{$\ComputeK_0$}
\label{algorithm:compute_k}
\end{algorithm}

\begin{lemma}\label{lem:kxy_correct}
    Let $x, y\in \FB(A)$, $\alpha\in \bool$ and let $\mathcal{T}_x$ and $\mathcal{T}_y$ be transducers representing $x$ and $y$, respectively. Then $\ComputeK_\alpha(\mathcal{T}_x, \mathcal{T}_y) = \overline{K_{\alpha}}$ and
    $\ComputeK_\alpha$ has time complexity $\mathcal{O}(|A|^2)$.
\end{lemma}
\begin{proof}
    We prove the lemma for $\ComputeK_0$, the proof for $\ComputeK_1$ is analogous.

    We begin by showing that the value of $c$ equals $\cont(x\circ 1^i)$ at the start of each iteration of the loop in line 3. When $i=|\cont(x)|$ at the start of the first loop, $|\cont(x\circ 1^i)| = |\cont(x)|-i = 0$ and so $\cont(x\circ 1^i) = \varnothing = c$. Assume that  $c =\cont(x\circ 1^i)$ holds for some $0<i\leq|\cont(x)|$ at the start of the loop in line 3. Since $i\neq 0$, in line 12 we set $c$ to be $c\cup \{q_x\clubsuit_x 1^i\}$. Since $\mathcal{T}_x$ represents $x$, $q_x\clubsuit_x 1^i = x\ast 1^i$ and so $c\cup \{q_x\clubsuit_x 1^i\} = \cont(x\circ 1^i)\cup\{x\ast 1^i\} = \cont(x\circ 1^{i-1})$, as required.

    Since $c = \cont(x\circ 1^i)$ for every value of $i$, and $\mathcal{T}_y$ represents $y$, it follows that checking $q_y\clubsuit_y0^{j+1}\neq \bot$ and $q_y\clubsuit_y 0^{j+1}\not\in c$ in line 5 is the same as checking $y\ast 0^{j+1}\neq \bot$ and $y\ast 0^{j+1}\not\in \cont(x\circ 1^i)$.

    We will show that $\overline{K}_0(i, j) = K_0(x\circ 1^i, y\circ 0^j)$ at the end of the loop that starts in line 4 for all values of $i$ and $j$.
    Suppose that $0\leq i \leq |\cont(x)|$ is arbitrary and $j = |\cont(y)|$. Then $y\ast 0^{j+1} = \bot$ and so $K_0(x\circ 1^i, y\circ 0^j) = \bot$ and the checks in line 5 and line 7 both fail. Since the value of $\overline{K}_0(i, j)$ is not changed during the first iteration of the loop, $\overline{K}_0(i, j) = \bot = K_0(x\circ 1^i, y\circ 0^j)$ when $j= |\cont(y)|$.
    Assume that at the start of the loop in line 4, $\overline{K}_0(i,j+1) = K_0(x\circ 1^i, y\circ 0^{j+1})$ for some $i$ and $j$. Then we can translate the checks and assignments in lines 5-9 into
    \[
    \overline{K}_{0}(i, j) =
        \begin{cases}
                1 & \text{if } y\ast 0^{j+1}\neq \bot \text{ and } y\ast 0^{j+1}\not \in \cont(x\circ 1^i)\\
        1+K_{0}(x\circ 1^i,y\circ 0^{j+1}) & \text{if } y\ast 0^{j+1}\in\cont(x\circ 1^i)\text{ and }K_{0}(x\circ 1^i, y\circ0^{j+1})\neq \bot\\
        \bot & \text{otherwise}
        \end{cases}
    \]
    the right hand of which is exactly the definition of $K_0(x\circ 1^i, y\circ 0^j)$. So $\overline{K}_0(i,j)$ is correct at the end of the loop starting in line 4.
   This completes the proof that the values of $\overline{K}_0$ are correct by the end of the algorithm.

    To show that the time complexity is $\mathcal{O}(|A|^2)$, note that the values of $\cont(x), \cont(y)$, $q_y\clubsuit_x 0^j$ and $q_x\clubsuit_x 1^i$ for $1\leq i\leq |\cont(x)|$ and $1\leq j\leq |\cont(y)|$ can be precomputed in $\mathcal{O}(A)$ time. The assignments in lines 1 and 2, the checks in lines 5 and 7, and the assignment in line 12 are all constant time.
    The lines within the loops are executed at most $(|\cont(x)|+1)(|\cont(y)|+1)\in \mathcal{O}(|A|^2)$ times, as required.
\end{proof}

The multiplication algorithm $\Multiply$ for transducers is given in \cref{algorithm:multiply}.
Informally, we start with an initially empty transducer, then copy all the states and transitions from $\mathcal{T}_x$ and $\mathcal{T}_y$ into $\mathcal{T}_{x\cdot y}$ so that $\mathcal{T}_{x\cdot y}$ coincides with each of $\mathcal{T}_x$ and $\mathcal{T}_y$ on the states $Q_x$ and $Q_y$ respectively. We then add all the states $(i, j)\in Q^\prime$ to $\mathcal{T}_{x\cdot y}$, and assign their transitions according to \cref{lem:multiplication:iterative}.
Although, in some sense, this informal description of the algorithm is not far away from the pseudo-code given in \cref{algorithm:multiply}, we require the precise formulation given in \cref{algorithm:multiply} to describe its time complexity in \cref{thm:mult_is_correct}.
    It is clear from the definition of $\mathcal{T}_{x\cdot y}$ in \eqref{eq:product-transducer}, and the preceding discussion, that the value returned by $\Multiply(\mathcal{T}_x, \mathcal{T}_y)$ is $\mathcal{T}_{x\cdot y}$.

\begin{algorithm}[]
\SetAlgoLined\DontPrintSemicolon\KwArguments{a pair of transducers $\mathcal{T}_x = (Q_x, \bool, A, q_x, T_x, \spadesuit_x, \clubsuit_x)$ and $\mathcal{T}_y = (Q_y, \bool, A, q_y, T_y, \spadesuit_y, \clubsuit_y)$ representing $x, y\in\FB(A)$ respectively,}
\KwReturns{the product transducer $\mathcal{T}_{x\cdot y} = (Q, \bool, A, q_0, T, \spadesuit, \clubsuit)$.}
Let $\mathcal{T} = (Q=\varnothing, \bool, A, q_0=\varepsilon, T=\varnothing, \spadesuit, \clubsuit)$ be the empty transducer.\;
Copy all the states and transitions from $\mathcal{T}_x$ and $\mathcal{T}_y$ into $\mathcal{T}$\;
$\overline{K_{0}}\gets \ComputeK_0(\mathcal{T}_x, \mathcal{T}_y)$ and $\overline{K_{1}}\gets \ComputeK_1(\mathcal{T}_x, \mathcal{T}_y)$\;
\For{$i\in \{|\cont(x)|, \ldots, 0\}$}{\For{$j\in \{|\cont(y)|, \ldots, 0\}$}{Add state $(i, j)$ to $\mathcal{T}$\;
\uIf{$\overline{K_{0}}(i, j)\neq \bot$}{$(i, j)\spadesuit 0 := (i, j+\overline{K_{0}}(i, j))$ and
$(i, j)\clubsuit 0 := q_y\clubsuit_y 0^{j+\overline{K_{0}}(i, j)}$\;
}
\Else{$(i, j)\spadesuit 0 := q_x\spadesuit_x 1^{i}0$ and
$(i, j)\clubsuit 0 := q_x\clubsuit_x 1^{i}0$\;
}
\uIf{$\overline{K_{1}}(i, j)\neq \bot$}{$(i, j)\spadesuit 1 := (i+\overline{K_{1}}(i, j), j)$ and
$(i, j)\clubsuit 1 := q_x\clubsuit_x 1^{i+\overline{K_{1}}(i, j)}$\;
}
\Else{$(i, j)\spadesuit 1 := q_y\spadesuit_y 0^{j}1$ and
$(i, j)\clubsuit 1 := q_y\clubsuit_y 0^{j}1$\;
}
}
}
$q_0\gets (0, 0)$, $T\gets T_x\cup T_y$\;
\Comment{Perform an identification of the relevant states.}
Remove the states $(i, |\cont(y)|)$ for all $i\in \{|\cont(x)|, \ldots, 0\}$\;
Remove the states $(|\cont(x)|, j)$ for all $j\in \{|\cont(y)|, \ldots, 0\}$\;
\For{$q\in Q$}{\For{$\alpha\in \bool$}{    \uIf{$q\spadesuit \alpha = (i, |\cont(y)|)$ for some $i$}{        $q\spadesuit \alpha = q_x\spadesuit_x 1^i$\;
    }
    \ElseIf{$q\spadesuit \alpha = (|\cont(x)|, j)$ for some $j$}{        $q\spadesuit \alpha = q_y\spadesuit_y 0^j$\;
    }
}
}
\KwRet{$\mathcal{T}$}\;
\caption{$\Multiply$}
\label{algorithm:multiply}
\end{algorithm}

\begin{theorem}\label{thm:mult_is_correct}
    Let $x, y\in \FB(A)$, $\alpha\in \bool$, and let $\mathcal{T}_x$ and $\mathcal{T}_y$ be transducers representing $x$ and $y$ respectively. Then
  $\Multiply$ has time complexity $\mathcal{O}(|\mathcal{T}_x| + |\mathcal{T}_y| + |A|^2)$ where $|\mathcal{T}_x|$ and $|\mathcal{T}_y|$ are the numbers of states in
  $\mathcal{T}_x$ and $\mathcal{T}_y$, respectively.
\end{theorem}
\begin{proof}

    Note that line 1 is constant time.
    Since there are at most twice as many transitions as states in any transducer, line 2 requires $\mathcal{O}(|\mathcal{T}_x| + |\mathcal{T}_y|)$ time.
    Line 3 takes $\mathcal{O}(|A|^2)$ time by \cref{lem:kxy_correct}.

    As we already noted in the proof of \cref{lem:kxy_correct} we can precompute the values $q_x\spadesuit_x 1^i0$, $q_y\spadesuit_y 0^j1$, $q_x\clubsuit_x 1^i$, and $q_y\clubsuit_y 0^j$ for $1\leq i\leq |\cont(x)|$ and $1\leq j\leq |\cont(y)|$ in $\mathcal{O}(|A|)$ time.
    Since $\overline{K_{\alpha}}$ is also precomputed, each of the checks and assignments in lines 6-16 can be done in constant time. This means that the for loop in lines 4-18 takes a total of $\mathcal{O}(|A|^2)$ time.

    Line 19 takes at most $\mathcal{O}(|\mathcal{T}_x|+|\mathcal{T}_y|)$ time since each transducer can have no more terminal states than states in total.
    Lines 21 and 22 are at most $\mathcal{O}(A)$ time, since we can remove a state in constant time.

    Finally, since we can check if a state is equal to $(|\cont(x)|, j)$ or $(i, |\cont(y)|)$ in constant time, lines 25-29 are constant time. So lines 23-31 take time proportional to the total number of states in $\mathcal{T}$ which by construction is $\mathcal{O}( |\mathcal{T}_x|+|\mathcal{T}_y|+|A|^2)$.
    Therefore the overall time complexity is $\mathcal{O}(|\mathcal{T}_x|+|\mathcal{T}_y|+|A|^2)$, as required.
\end{proof}

\begin{example}\label{ex:mult}
    Let $x = eaec$ and $y = bcacbcd \in \FB(a, b, c, d, e)$. We illustrate the operation of $\Multiply$ in four steps. First, we construct $Q^\prime$, initially with no transitions. In the second step, the values in the tables for $\overline{K}_0$ and $\overline{K}_1$  in \cref{table-k-0-1} are used to find the transitions between the states in $Q^\prime$. For example, since $\overline{K}_0(0, 0)=1$, we can add the $0$-transition from $(0,0)$ to $(0, 1)$ labelled by $q_y\clubsuit_y 0^{1} = d$. Similarly, $\overline{K}_1(0, 0) = 2$, so there is a $1$-transition from $(0, 0)$ to $(0, 2)$ labelled by $q_y\clubsuit 0^2 = e$. In the same way, the  remaining transitions between the states in $Q'$ can be defined whenever $\overline{K}_0(i,j)\neq \bot$ or $\overline{K}_1(i,j)\not=\bot$; see \cref{subfig-1}.
    In the third step, we consider those states $(i, j) \in Q'$ such that $\overline{K}_0(i,j)= \bot$ and $\overline{K}_1(i,j) = \bot$. For example,
     $\overline{K}_1(2, 0) = \bot$ and so $(2, 0)\spadesuit 1 = y_0\spadesuit_y 0^0 1 = y_2$ and $(i, j)\clubsuit 1 = y_0\clubsuit_y 0^0 1= a$. Similarly,
    $\overline{K}_1(2, 1) = \bot$, so $(2, 1)\spadesuit 1 = y_0\spadesuit_y 0 1 = y_4$ and $(2, 1)\clubsuit 1 = y_0\clubsuit_y 0 1= a$; see \cref{subfig-2}.
    Finally, the remaining transitions are essentially those from $\mathcal{T}_x$ and $\mathcal{T}_y$ but with some states identified; see
    \cref{fig:xy-product-transducer}.
\end{example}
    \begin{figure}
    \begin{subfigure}{0.5\textwidth}
    \begin{tikzpicture}[scale=0.9]
        \node (Q00) at (0,0) {$(0, 0)$};

        \node (Q01) at (-1,-1) {$(0, 1)$};
        \node (Q10) at (1,-1) {$(1, 0)$};

        \node (Q02) at (-2,-2) {$(0, 2)$};
        \node (Q11) at (0,-2) {$(1, 1)$};
        \node (Q20) at (2,-2) {$(2, 0)$};

        \node (Q03) at (-3,-3) {$(0, 3)$};
        \node (Q12) at (-1,-3) {$(1, 2)$};
        \node (Q21) at (1,-3) {$(2, 1)$};
        \node (Q30) at (3,-3) {$(3, 0)$};

        \node (Q04) at (-4,-4) {$(0, 4)$};
        \node (Q13) at (-2,-4) {$(1, 3)$};
        \node (Q22) at (0,-4) {$(2, 2)$};
        \node (Q31) at (2,-4) {$(3, 1)$};

        \node (Q14) at (-3,-5) {$(1, 4)$};
        \node (Q23) at (-1,-5) {$(2, 3)$};
        \node (Q32) at (1,-5) {$(3, 2)$};

        \node (Q24) at (-2,-6) {$(2, 4)$};
        \node (Q33) at (0,-6) {$(3, 3)$};

        \node (Q34) at (-1,-7) {$(3, 4)$};

 \draw [-latex] (Q00) ++(0,0.75) -- (Q00);
 \draw [-latex] (Q34) -- ++(0,-0.75);

        \draw [-latex, draw=color3] (Q00) -- (Q01);
        \draw [-latex reversed, dashed, draw=color3] (Q00) to [out=0,in=90] (Q20);

        \draw [-latex, draw=color1] (Q01) to [out=180, in=90] (Q04);
        \draw [-latex reversed, dashed, draw=color4] (Q01) to [out=0,in=90] (Q21);
        \draw [-latex, draw=color3] (Q10) -- (Q11);
        \draw [-latex reversed, dashed, draw=color4] (Q10) -- (Q20);

        \draw [-latex, draw=color1] (Q02) to [out=180, in=90] (Q04);
        \draw [-latex reversed, dashed, draw=color0] (Q02) --(Q12);
        \draw [-latex, draw=color0] (Q11) -- (Q12);
        \draw [-latex reversed, dashed, draw=color4] (Q11) -- (Q21);
        \draw [-latex, draw=color3] (Q20) -- (Q21);

        \draw [-latex, draw=color1] (Q03) -- (Q04);
        \draw [-latex reversed, dashed, draw=color0] (Q03) --(Q13);
        \draw [-latex, draw=color1] (Q12) to [out=180, in=90] (Q14);
        \draw [-latex reversed, dashed, draw=color4] (Q12) --(Q22);
        \draw [-latex, draw=color0] (Q21) -- (Q22);
                \draw [-latex, draw=color3] (Q30) -- (Q31);

                \draw [-latex reversed, dashed, draw=color0] (Q04) --(Q14);
        \draw [-latex, draw=color1] (Q13) -- (Q14);
        \draw [-latex reversed, dashed, draw=color4] (Q13) --(Q23);
        \draw [-latex, draw=color1] (Q22) to [out=180, in=90] (Q24);
                \draw [-latex, draw=color0] (Q31) -- (Q32);

                \draw [-latex reversed, dashed, draw=color4] (Q14) -- (Q24);
        \draw [-latex, draw=color1] (Q23) -- (Q24);
        \draw [-latex reversed, dashed, draw=color2] (Q23) --(Q33);
        \draw [-latex, draw=color2] (Q32) -- (Q33);

                \draw [-latex reversed, dashed, draw=color2] (Q24) -- (Q34);
        \draw [-latex, draw=color1] (Q33) -- (Q34);
    \end{tikzpicture}
    \caption{Transitions for $(i, j) \in Q'$, $\overline{K}_0(i,j)\neq\bot$ or $\overline{K}_1(i,j) \neq \bot$.}
    \label{subfig-1}
    \end{subfigure}
    \begin{subfigure}{0.5\textwidth}
    \begin{tikzpicture}[scale=0.9]

        \node (Q00) at (0,0) {$(0, 0)$};

        \node (Q01) at (-1,-1) {$(0, 1)$};
        \node (Q10) at (1,-1) {$(1, 0)$};

        \node (Q02) at (-2,-2) {$(0, 2)$};
        \node (Q11) at (0,-2) {$(1, 1)$};
        \node (Q20) at (2,-2) {$(2, 0)$};

        \node (Q03) at (-3,-3) {$(0, 3)$};
        \node (Q12) at (-1,-3) {$(1, 2)$};
        \node (Q21) at (1,-3) {$(2, 1)$};
        \node (Q30) at (3,-3) {$(3, 0)$};

        \node (Q04) at (-4,-4) {$(0, 4)$};
        \node (Q13) at (-2,-4) {$(1, 3)$};
        \node (Q22) at (0,-4) {$(2, 2)$};
        \node (Q31) at (2,-4) {$(3, 1)$};

        \node (Q14) at (-3,-5) {$(1, 4)$};
        \node (Q23) at (-1,-5) {$(2, 3)$};
        \node (Q32) at (1,-5) {$(3, 2)$};

        \node (Q24) at (-2,-6) {$(2, 4)$};
        \node (Q33) at (0,-6) {$(3, 3)$};

        \node (Q34) at (-1,-7) {$(3, 4)$};

        \node (x_1) at (-5, -5) {$x_1$};
        \node (x_3) at (-4, -6) {$x_3$};
        \node (x_5) at (-3, -7) {$x_5$};

        \node (y_2) at (4, -4) {$y_2$};
        \node (y_4) at (3, -5) {$y_4$};
        \node (y_7) at (2, -6) {$y_7$};
        \node (y_9) at (1, -7) {$y_9$};

 \draw [-latex] (Q00) ++(0,0.75) -- (Q00);
 \draw [-latex] (Q34) -- ++(0,-0.75);

        \draw [-latex, draw=color3] (Q00) -- (Q01);
        \draw [-latex reversed, dashed, draw=color3] (Q00) to [out=0,in=90] (Q20);

        \draw [-latex, draw=color1] (Q01) to [out=180, in=90] (Q04);
        \draw [-latex reversed, dashed, draw=color4] (Q01) to [out=0,in=90] (Q21);
        \draw [-latex, draw=color3] (Q10) -- (Q11);
        \draw [-latex reversed, dashed, draw=color4] (Q10) -- (Q20);

        \draw [-latex, draw=color1] (Q02) to [out=180, in=90] (Q04);
        \draw [-latex reversed, dashed, draw=color0] (Q02) --(Q12);
        \draw [-latex, draw=color0] (Q11) -- (Q12);
        \draw [-latex reversed, dashed, draw=color4] (Q11) -- (Q21);
        \draw [-latex, draw=color3] (Q20) -- (Q21);

        \draw [-latex, draw=color1] (Q03) -- (Q04);
        \draw [-latex reversed, dashed, draw=color0] (Q03) --(Q13);
        \draw [-latex, draw=color1] (Q12) to [out=180, in=90] (Q14);
        \draw [-latex reversed, dashed, draw=color4] (Q12) --(Q22);
        \draw [-latex, draw=color0] (Q21) -- (Q22);
                \draw [-latex, draw=color3] (Q30) -- (Q31);

                \draw [-latex reversed, dashed, draw=color0] (Q04) --(Q14);
        \draw [-latex, draw=color1] (Q13) -- (Q14);
        \draw [-latex reversed, dashed, draw=color4] (Q13) --(Q23);
        \draw [-latex, draw=color1] (Q22) to [out=180, in=90] (Q24);
                \draw [-latex, draw=color0] (Q31) -- (Q32);

                \draw [-latex reversed, dashed, draw=color4] (Q14) -- (Q24);
        \draw [-latex, draw=color1] (Q23) -- (Q24);
        \draw [-latex reversed, dashed, draw=color2] (Q23) --(Q33);
        \draw [-latex, draw=color2] (Q32) -- (Q33);

                \draw [-latex reversed, dashed, draw=color2] (Q24) -- (Q34);
        \draw [-latex, draw=color1] (Q33) -- (Q34);

        \draw [-latex reversed, dashed, draw=color0] (Q20) to [out=0, in=90] (y_2);
        \draw [-latex reversed, dashed, draw=color0] (Q30) -- (y_2);
        \draw [-latex reversed, dashed, draw=color0] (Q21) to [out=0, in=90] (y_4);
        \draw [-latex reversed, dashed, draw=color0] (Q31) -- (y_4);
        \draw [-latex reversed, dashed, draw=color1] (Q22) to [out=0, in=90] (y_7);
        \draw [-latex reversed, dashed, draw=color1] (Q32) -- (y_7);
        \draw [-latex reversed, dashed, draw=color1] (Q33) -- (y_9);

        \draw [-latex reversed, dashed, draw=color2] (Q04) -- (x_1);
        \draw [-latex reversed, dashed, draw=color2] (Q14) -- (x_3);
        \draw [-latex reversed, dashed, draw=color2] (Q24) -- (x_5);
    \end{tikzpicture}
    \caption{Transitions for $(i, j) \in Q'$, $\overline{K}_0(i,j)= \bot =\overline{K}_1(i,j)$.}
    \label{subfig-2}
    \end{subfigure}
    \caption{The second and third steps in the construction of the product transducer $\mathcal{T}_{xy}$ for $x=eaec$ and $y = bcacbcd$ constructed from the minimal transducers $\mathcal{T}_x$ and $\mathcal{T}_y$ in \cref{fig:xy-minimal-transducers}. The pinkish-red transitions are labelled by $a$, the light-green ones by $b$, the blue ones by $c$, the yellow ones by $d$ and purple ones by $e$.}
    \label{figure-product-2}
    \end{figure}
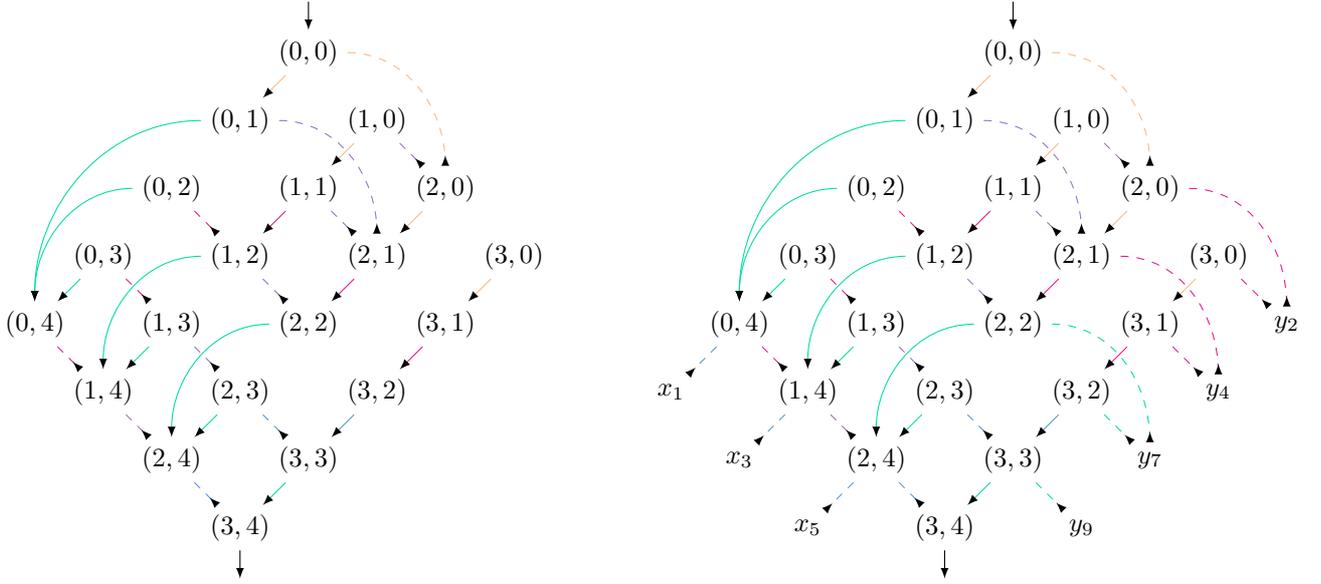

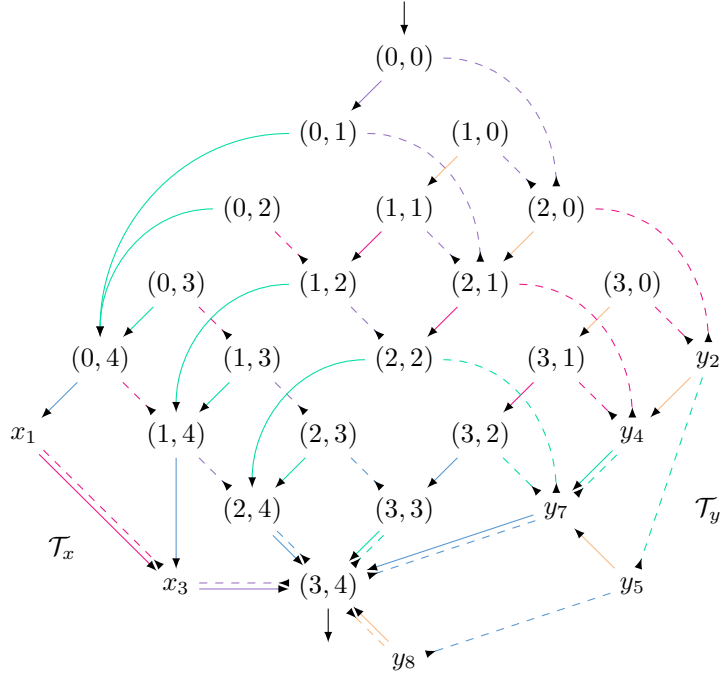
\begin{figure}
 \centering
  \pgfdeclaredecoration{sl}{initial}{
 \state{initial}[width=\pgfdecoratedpathlength-1sp]{
  \pgfmoveto{\pgfpointorigin}
 }
 \state{final}{
  \pgflineto{\pgfpointorigin}
 }
}

\tikzset{thin_parallel_arrow_1/.style={-latex, decoration={sl,raise=-0.4mm},decorate}}
\tikzset{thin_parallel_arrow_2/.style={-latex reversed, dashed, decoration={sl,raise=0.4mm},decorate}}

\usetikzlibrary{calc}

\tikzstyle{vertex}=[circle,
text opacity=1,
inner sep=2pt,
align=center,
minimum size=0.2em]

\begin{tikzpicture}
 \node (x_0) at (-4,-4) {$(0, 4)$};  \node (x_1) at (-5, -5) {$x_1$};
 \node (x_2) at (-3,-5) {$(1, 4)$};  \node (x_3) at (-3, -7) {$x_3$};
 \node (x_4) at (-2,-6) {$(2, 4)$};  \node (x_5) at (-1, -7) {$(3, 4)$};  \node at (-4.5, -6.5) [draw=none]{$\mathcal{T}_x$};

 \node (y_0) at (3,-3) {$(3, 0)$};
 \node (y_1) at (2,-4) {$(3, 1)$};
 \node (y_2) at (4,-4) {$y_2$};
 \node (y_3) at (1,-5) {$(3, 2)$};
 \node (y_4) at (3,-5) {$y_4$};
 \node (y_5) at (3,-7) {$y_5$};
 \node (y_6) at (0,-6) {$(3, 3)$};
 \node (y_7) at (2,-6) {$y_7$};
 \node (y_8) at (0,-8) {$y_8$};
 \node (y_9) at (-1,-7) {};
 \node at (4, -6) [draw=none]{$\mathcal{T}_y$};

 \node (Q00) at (0,0) {$(0, 0)$};

 \node (Q01) at (-1,-1) {$(0, 1)$};
 \node (Q10) at (1,-1) {$(1, 0)$};

 \node (Q02) at (-2,-2) {$(0, 2)$};
 \node (Q11) at (0,-2) {$(1, 1)$};
 \node (Q20) at (2,-2) {$(2, 0)$};

 \node (Q03) at (-3,-3) {$(0, 3)$};
 \node (Q12) at (-1,-3) {$(1, 2)$};
 \node (Q21) at (1,-3) {$(2, 1)$};

 \node (Q13) at (-2,-4) {$(1, 3)$};
 \node (Q22) at (0,-4) {$(2, 2)$};

 \node (Q23) at (-1,-5) {$(2, 3)$};

 \draw [-latex] (x_5) -- ++(0,-0.75);
 \draw [-latex] (Q00) ++(0,0.75) -- (Q00);

 \draw [-latex, draw=color2] (x_0) -- (x_1);
 \draw [-latex reversed, dashed, draw=color0] (x_0) -- (x_2);

 \draw [thin_parallel_arrow_1, draw=color0] (x_1) --  (x_3);
 \draw [thin_parallel_arrow_2, draw=color0] (x_1) --  (x_3);

 \draw [-latex, draw=color2] (x_2) -- (x_3);
 \draw [-latex reversed, dashed, draw=color4] (x_2) -- (x_4);

 \draw [thin_parallel_arrow_1, draw=color4] (x_3) -- (x_5);
 \draw [thin_parallel_arrow_2, draw=color4] (x_3) -- (x_5);

 \draw [thin_parallel_arrow_1, draw=color2] (x_4) -- (x_5);
 \draw [thin_parallel_arrow_2, draw=color2] (x_4) -- (x_5);

 \draw [-latex, draw=color3] (y_0) -- (y_1);
 \draw [-latex reversed, dashed, draw=color0] (y_0) -- (y_2);

 \draw [-latex, draw=color0] (y_1) -- (y_3);
 \draw [-latex reversed, dashed, draw=color0] (y_1) -- (y_4);

 \draw [-latex, draw=color3] (y_2) -- (y_4);
 \draw [-latex reversed, dashed, draw=color1] (y_2) -- (y_5);

 \draw [-latex, draw=color2] (y_3) -- (y_6);
 \draw [-latex reversed, dashed, draw=color1] (y_3) -- (y_7);

 \draw [thin_parallel_arrow_2, draw=color1] (y_4) -- (y_7);
 \draw [thin_parallel_arrow_1, draw=color1] (y_4) -- (y_7);

 \draw [-latex, draw=color3] (y_5) -- (y_7);
 \draw [-latex reversed, dashed, draw=color2] (y_5) -- (y_8);

 \draw [thin_parallel_arrow_2, draw=color1] (y_6) -- (x_5);
 \draw [thin_parallel_arrow_1, draw=color1] (y_6) -- (x_5);

 \draw [thin_parallel_arrow_2, draw=color2] (y_7) -- (x_5);
 \draw [thin_parallel_arrow_1, draw=color2] (y_7) -- (x_5);

 \draw [thin_parallel_arrow_2, draw=color3] (y_8) -- (x_5);
 \draw [thin_parallel_arrow_1, draw=color3] (y_8) -- (x_5);

 \draw [-latex, draw=color4] (Q00) -- (Q01);
        \draw [-latex reversed, dashed, draw=color4] (Q00) to [out=0,in=90] (Q20);

        \draw [-latex, draw=color1] (Q01) to [out=180, in=90] (x_0);
        \draw [-latex reversed, dashed, draw=color4] (Q01) to [out=0,in=90] (Q21);
        \draw [-latex, draw=color3] (Q10) -- (Q11);
        \draw [-latex reversed, dashed, draw=color4] (Q10) -- (Q20);

        \draw [-latex, draw=color1] (Q02) to [out=180, in=90] (x_0);
        \draw [-latex reversed, dashed, draw=color0] (Q02) --(Q12);
        \draw [-latex, draw=color0] (Q11) -- (Q12);
        \draw [-latex reversed, dashed, draw=color4] (Q11) -- (Q21);
        \draw [-latex, draw=color3] (Q20) -- (Q21);

        \draw [-latex, draw=color1] (Q03) -- (x_0);
        \draw [-latex reversed, dashed, draw=color0] (Q03) --(Q13);
        \draw [-latex, draw=color1] (Q12) to [out=180, in=90] (x_2);
        \draw [-latex reversed, dashed, draw=color4] (Q12) --(Q22);
        \draw [-latex, draw=color0] (Q21) -- (Q22);

                        \draw [-latex, draw=color1] (Q13) -- (x_2);
        \draw [-latex reversed, dashed, draw=color4] (Q13) --(Q23);
        \draw [-latex, draw=color1] (Q22) to [out=180, in=90] (x_4);

                        \draw [-latex, draw=color1] (Q23) -- (x_4);
        \draw [-latex reversed, dashed, draw=color2] (Q23) --(y_6);

        \draw [-latex reversed, dashed, draw=color0] (Q20) to [out=0, in=90] (y_2);
        \draw [-latex reversed, dashed, draw=color0] (Q21) to [out=0, in=90] (y_4);
        \draw [-latex reversed, dashed, draw=color1] (Q22) to [out=0, in=90] (y_7);

\end{tikzpicture}

 \caption{The complete product transducer $\mathcal{T}_{xy}$ for $x=eaec$ and $y = bcacbcd$ constructed from the minimal transducers $\mathcal{T}_x$ and $\mathcal{T}_y$ in \cref{fig:xy-minimal-transducers}. The colors of the edges correspond to the generators labeling them in the same way as in \cref{figure-product-2}.}
 \label{fig:xy-product-transducer}
\end{figure}

\section{Minimal Word Representative}
\label{sec:minimal}

In this section, we present an algorithm for determining the short-lex least word $\min(w)$ belonging to the $\sim$-class of $w\in A ^ *$.  We will abuse our notation, by writing $\min(x)$ to mean the short-lex least word such that $\min(x)/{\sim} = x$, where $x\in \FB(A)$. Clearly, $\min(w/{\sim}) = \min(w)$.
 Given a transducer $\mathcal{T}_x$ representing $x\in \FB(A)$, the algorithm for computing $\min(x)$, that we introduce in this section,  has $\mathcal{O}(|\mathcal{T}_x|\cdot|A|)$ time complexity. This leads to a $\mathcal{O}(|w|\cdot|A| + |\min(w)|\cdot |A|^2)$ time algorithm for computing $\min(w)$ given $w\in A^*$.

To establish our algorithm we first show some technical results about the structure of $\min(x)$. The following lemma was first suggested to the authors by T. Conti-Leslie in the summer of 2020; it is possible that this result is known, but that authors could not locate it in the literature.
\begin{lemma}[Conti-Leslie Lemma]\label{lem:contileslie}
    Let $w\in A^*$ and let $s = \min(w\circ 0) (w\ast 0)$ and $t = (w\ast 1) \min(w\circ 1)$.
    If $i\in \{1, \ldots, |s|\}$ is such that $s_{(i, |s|)}$ is the longest suffix of $s$ that is also a prefix of $t$, then
    $\min(w) = s_{(1, i-1)}t$.
\end{lemma}
\cref{lem:contileslie} says, in other words, that to find a minimal word representative, it suffices to minimize the prefix and suffix, and then overlap them as much as possible. To prove \cref{lem:contileslie} we require the following.

\begin{lemma}\label{lem:contilesliehelper}
For all $w\in A^\ast$ and for all $\alpha\in \bool^\ast$, $\min(w\circ\alpha) = \min(w)\circ\alpha$.
\end{lemma}
\begin{proof}
Note that for all $w\in A^*$ and $\alpha \in\bool^*$, there exist $u,v\in A^*$ such that $w = u(w\circ \alpha)v$. Since $w\circ\alpha \sim \min(w\circ\alpha)$ by definition, $w=u(w\circ\alpha)v\sim u\min(w\circ \alpha)v$ and $u\min(w\circ\alpha)v\leq u(w\circ \alpha)v = w$.
 Repeating this process for $\min(w)$ yields $u\min(\min(w)\circ \alpha)v\leq u(\min(w)\circ \alpha)v = \min(w)$. By the minimality of $\min(w)$, $u\min(\min(w)\circ \alpha)v = u(\min(w)\circ \alpha) v$ and so
$\min(\min(w)\circ\alpha) = \min(w)\circ \alpha$.
 By the Green-Rees Theorem (\cref{thm:equal_free_band}),
$\min(w)\circ \alpha \sim w\circ\alpha$ and so $\min(\min(w)\circ \alpha) = \min(w\circ \alpha)$. It follows that $\min(w)\circ\alpha = \min(w\circ \alpha)$ for all $w\in A^*$ and $\alpha\in \bool^*$.
\end{proof}

\begin{proof}[Proof of \cref{lem:contileslie}]
By \cref{lem:contilesliehelper}, it follows that $\min(w)\circ 0 = \min(w\circ0)$ and $\min(w)\circ 1 = \min(w\circ 1)$. By the Green-Rees Theorem, $\min(w)\ast 0 = w\ast 0$ and $\min(w)\ast 1 = w\ast 1$. Hence from the definitions of $\circ$ and $\ast$, $\min(w)$ has a prefix $s = (\min(w)\circ 0)\ (\min(w)\ast0) = (\min(w\circ 0))\ (w\ast 0)$  and a suffix $t = (w\ast 1)\ (\min(w\circ 1))$.

There are now two possible cases. Either $s$ and $t$ overlap within $\min(w)$ or they do not. If there is no overlap, then $\min(w) = su^\prime t$ for some $u^\prime \in A^*$. Note, however, that $st = \min(w\circ0)(w\ast0)(w\ast1)\min(w\circ1)\sim w$ by the Green-Rees Theorem. So if $s,t$ do not overlap within $\min(w)$, then $u^\prime=\varepsilon$ and $\min(w) = st$.

Otherwise, if there is some overlap of $s$ and $t$ within $\min(w)$, then $\min(w) = s^\prime v^\prime t^\prime$ for some $s^\prime, t^\prime\in A^*$ and $v^\prime\in A^+$ such that $s^\prime v^\prime = s$ and $v^\prime t^\prime = t$. If $s^{\prime\prime}, t^{\prime\prime}\in A^\ast$ and $v^{\prime\prime}\in A^+$ with $|v^\prime| < |v^{\prime\prime}|$ are such that $s^{\prime\prime}v^{\prime\prime} = s$ and $v^{\prime\prime}t^{\prime\prime} = t$, then $\min(w)\sim s^{\prime\prime}v^{\prime\prime}t^{\prime\prime}$ and $|s^{\prime\prime} v^{\prime\prime}t^{\prime\prime}| < |s^\prime v^\prime t^\prime|$, contradicting the minimality of $\min(w)$. Therefore $v^\prime$ is the largest suffix of $s$ that is also a prefix of $t$, as required.
\end{proof}

In order to later use \cref{lem:contileslie} in a computational manner, we will fully classify the possible types of overlap between $s$ and $t$ in \cref{lem:one_overlap}. In general, there are  many different suffixes of a given word that are also prefixes of other words, e.g. for $bbababa$, each of the proper suffixes $a, aba, ababa$ is a prefix of $ababaa$. Perhaps surprisingly, as it will turn out in \cref{lem:one_overlap}, at most one suffix $s = \min(w\circ 0)(w\ast 0)$ (from \cref{lem:contilesliehelper}) is a prefix of $t = (w\ast 1)\min(w\circ 1)$. This is key to the linear complexity of the algorithm presented in this section.

We require the following observation in the proof of \cref{lem:one_overlap}.
 If $w\in A ^ \ast$, then
 $w\ast 0\not \in \cont(w\circ 0)$, and so $\cont(w) = \cont(w\circ 0)\cup \{w\ast 0\}$. By repeatedly applying this, we get $\cont(w) = \bigcup_{k=1}^{|\cont(w)|} \{w\ast 0^k\}$.
Hence if $a\in \cont(w)$, then $a = w\ast 0 ^ k$ for some $k$. An analogous statement holds when $0$ is replaced by $1$.

\begin{lemma}\label{lem:one_overlap}
Let $w\in A^\ast$, let $s = \min(w\circ 0)(w\ast 0)$, and let $t=(w\ast 1)\min(w\circ 1)$. Then there is at most one non-empty suffix $v^\prime$ of $s$ that is also a prefix of $t$.
\begin{enumerate}[label=\normalfont{(\roman*)}]
    \item If $w\ast 0 = w\ast 1$, then $v' = w\ast0$ and so  $\min(w) = \min(w\circ 0)(w\ast0)\min(w\circ 1)$.
\end{enumerate}
If $w\ast 0\neq w\ast 1$, then there exist $k,l>0$ such that $w\ast 01^k = w\ast 1$ and $w\ast 10^l = w\ast 0$ and either:
\begin{enumerate}[label=\normalfont(\roman*), start=2]
    \item $k=l$, $\min(w\circ 01^k) = \min(w\circ 10^k)$, and $v^\prime = (w\ast 1)\min(w\circ 01^k)(w\ast 0)$ and so $\min(w)$ is the unique word with prefix $\min(w\circ 0)(w\ast 0)$ and suffix $(w\ast 1)\min(w\circ 1)$ whose overlap is $(w\ast 1)\min(w\circ 01^k)(w\ast 0)$; or
    \item no non-empty suffix of $s$ is a prefix of $t$, and so $\min(w) = \min(w\circ 0)(w\ast0)(w\ast 1)\min(w\circ 1)$.
\end{enumerate}
\end{lemma}
\begin{proof}
Assume that $s = s'v'$ and $t = v't'$ for some $s', v', t' \in A ^ *$.

\begin{enumerate}[label=(\roman*)]
\item If $w\ast 0 = w\ast 1$, then we will show that $v^\prime = w\ast 0$. Clearly, $v^\prime=w\ast 0$ is a suffix of $s = (w\circ 0)(w\ast 0)$ and a prefix of $t = (w\ast 1)(w\circ 1)$.

Assume that there is another non-empty suffix $v^{\prime\prime}$ of $s$ that is a prefix of $t$ not equal to $v'$. Since $v''$ is distinct from $w\ast 0$, $|v^{\prime\prime}|> 1$. Then $v^{\prime\prime} t'' = t= (w\ast 1)\min(w\circ 1)$ for some $t''\in A ^ *$, and so the first letter of $v^{\prime\prime}$ is $w\ast1 = w\ast 0$. On the other hand, $s'' v^{\prime\prime} = s = \min(w\circ 0)(w\ast0)$ for some $s'' \in A  ^ *$, $|v^{\prime\prime}|>1$, and so the first letter of $v^{\prime\prime}$, $w\ast 0$, occurs somewhere in $\min(w\circ 0)$. But by definition, $\cont(\min(w\circ0)) = \cont(w\circ0)=\cont(w)\setminus\{w\ast 0\}$, a contradiction.

So if $w\ast 0 = w\ast 1$, then $v^\prime = w\ast 0$ is the only suffix of $s$ that is also a prefix of $t$, which implies that $\min(w) = \min(w\circ 0)(w\ast 0)\min(w\circ 1)$.
\end{enumerate}

For the remainder of the proof we suppose that $w\ast 0\neq w\ast 1$. We establish that there is at most one non-empty suffix $v^\prime$ of $s$ that is also a prefix of $t$.

Assume that such a suffix $v^\prime$ exists. Clearly $|v^\prime| > 1$ since $v^\prime$ is non-empty and if $|v^\prime| = 1$, then we would have $w\ast 0=w\ast 1$, a contradiction.

Since $s^\prime v^\prime = s = \min(w\circ 0)(w\ast0)$ and $v^\prime t^\prime = t = (w\ast1)\min(w\circ1)$, it follows that $v^\prime = (w\ast 1)v^{\prime\prime}(w\ast 0)$ for some $v^{\prime\prime}\in A^\ast$. By considering the content of $\min(w\circ 0)$ and $\min(w\circ 1)$, we deduce that $\cont(v^{\prime\prime}) = \cont(w)\setminus\{w\ast 0, w\ast 1\}$. Since $w\ast 1$ is the first letter of $v'$ and $w\ast 1$ does not occur anywhere else in $v'$, $v^\prime$ starts at the rightmost occurrence of $w\ast 1$ in $s$. Similarly, $v^\prime$ must end at the leftmost occurrence of $w\ast 0$ in $t$.
In particular, $v'$ is the unique suffix of $s$ that is also a prefix of $t$.

Next, we establish the existence of the values $k$ and $l$ in the statement of the lemma.
By the comment immediately before the lemma, and because
 $w\ast 1\in \cont(w)\setminus\{w\ast 0\} = \cont(\min(w)\circ 0)$, there exists $k>0$ such that $(w\circ 0)\ast 1^k = w\ast 1$. But $w\ast 01^k = (w\circ 0)\ast 1^k$ by definition, and so $w\ast 01^k = w\ast 1$.
In a similar manner, there exists $l>0$ such that $w\ast 10^l  = w\ast 0$.

Before proving cases (ii) and (iii) of the lemma,
we determine the rightmost occurrence of $w\ast 1$ in $s$ and the leftmost occurrence of $w\ast 0$ in $t$. Since $w\ast 01^k = w\ast 1$, the definition of $\circ$ implies that
there is a prefix $s''$ of $\min(w)\circ 0$ with
\[
\min(w)\circ 0 = s^{\prime\prime}(\min(w)\ast 01^k)(\min(w)\circ 01^k) = s^{\prime\prime}(w\ast 01^k)\min(w\circ 01^k) = s^{\prime\prime}(w\ast 1)\min(w\circ 01^k).\]
Since $(\min(w)\circ 0)(w\ast 0) = \min(w\circ 0)(w\ast 0) = s$, the word $s^{\prime\prime}$ is also a prefix of $s$, and so
\[s = s^{\prime\prime}(w\ast 1)\min(w\circ 01^k)(w\ast 0).\]
Since $w\ast 1  = w\ast 01^k = \min(w)\ast 01^k \not \in \cont(\min(w\circ 01^k))$ and $s = s^{\prime\prime}(w\ast 1)\min(w\circ 01^k)(w\ast0)$, the rightmost occurrence of $w\ast 1$ in $s$ is just after $s^{\prime\prime}$.
In a similar manner, $t = (w\ast 1)\min(w\circ10^l)(w\ast0)t^{\prime\prime}$ for some suffix $t^{\prime\prime}$, and the leftmost occurrence of $w\ast 0$ in $t$ is just before $t''$.

\begin{enumerate}[label=(\roman*), start=2]
    \item
     If $k=l$ and $\min(w\circ 01^k) = \min(w\circ10^k)$, then from the above $v^\prime = (w\ast 1)\min(w\circ01^k)(w\ast0)$ is a suffix of $s$ and a prefix of $t$, and hence the unique word with this property.

    \item Assume that $k\neq l$ or $\min(w\circ 01^k) \neq \min(w\circ10^k)$.

     If there exists a non-empty suffix $v^\prime$ of $s$ that is also a prefix of $t$, then since $v^\prime$ must begin at the rightmost occurrence of $w\ast 1$ in $s$, it follows that $v^\prime = (w\ast 1)\min(w\circ 01^k)(w\ast0)$. Similarly, by considering the content of $\min(w\circ 10^l)$, we can establish that $v^\prime = (w\ast1)\min(w\circ 10^l)(w\ast 0)$.
    Hence $(w\ast 1)\min(w\circ 01^k)(w\ast 0) = (w\ast 1) \min(w\circ 10^l)(w\ast0)$  and so $\min(w\circ 01^k) = \min(w\circ 10^l)$.
    Since $|\cont(w)|-(l+1) = |\cont(\min(w\circ10^l))| = |\cont(\min(w\circ01^k))| = |\cont(w)|-(k+1)$, this implies that $k=l$, a contradiction.
    Therefore $v'$ does not exist, and there is no non-empty suffix of $s$ that is a prefix of $t$, as required.\qedhere
\end{enumerate}
\end{proof}

Let $x\in\FB(A)$ and let $\mathcal{T}_x= (Q, \bool, A, q_0, T, \spadesuit, \clubsuit)$ be the minimal transducer representing $x$. Note that by the Green-Rees Theorem (\cref{thm:equal_free_band}), it follows that any two $v, w\in A^\ast$ with $v/{\sim} = w/{\sim} = x$ will be in the same case of \cref{lem:one_overlap}, and if they are in case (ii), then they will have the same value for $k$. So we may  refer to these cases applying to $x\in \FB(A)$.
If $q\in Q$ is any state and $y\in \FB(A)$ is the element represented by $q$ (i.e. $y$ is not necessarily equal to $x$), then \cref{lem:one_overlap} can be applied to $y$.
We say that \cref{lem:one_overlap} applies to a state $q$ of $\mathcal{T}_x$ to mean that it applies to the element of the free band represented by $q$.

In \cref{algorithm:minimalwordhelper}, we present a $\mathcal{O}(|\cont(x)|)$ time algorithm named $\ClassifyCase$ for determining which of the cases (i), (ii), or (iii) of \cref{lem:one_overlap} applies to a given $q\in Q$. If case (ii) applies, then the algorithm also finds $1\leq k < |\cont(x)|$ such that $\min(y\circ 01^k) = \min(y\circ 10^k)$ where $y\in \FB(A)$ is represented by $q$. For convenience when presenting \cref{algorithm:minimalword} later, $\ClassifyCase$ returns both the case of \cref{lem:one_overlap} that applies to $q$ and an integer. In case (ii) this integer is $k$, and in cases (i) and (iii) it is $|\cont(y)|$ where $y\in \FB(A)$ is represented by the argument $q\in Q$.

\begin{algorithm}[H]
\SetAlgoLined\DontPrintSemicolon\KwArguments{a minimal transducer $\mathcal{T}_x = (Q, \bool, A, q_0, T, \spadesuit, \clubsuit)$ representing some $x\in\FB(A)$, and a state $q\in Q$,}
\KwReturns{the case $c$ of \cref{lem:one_overlap} that applies to the element of $\FB(A)$ represented by $q$, and the corresponding integer $k$ if $c$ is (ii).}
\If{$q\clubsuit 0 = q\clubsuit 1$}{    \KwRet{case (i) and $|\cont(\Psi(q))|$}\;
}
$u\gets q\spadesuit 0$,  $v\gets q\spadesuit 1$\;
\For{$k\in \{1, \ldots, |\cont(\Psi(q))|\}$}{    \If{$u \clubsuit 1 = q\clubsuit 1$, $v \clubsuit 0 = q\clubsuit 0$, and $u\spadesuit 1 = v\spadesuit 0$}{        \KwRet{case (ii) and $k$}
    }
    $u\gets u\spadesuit 1$,
    $v\gets v\spadesuit 0$\;
    \If{$u$ or $v$ is terminal}{        \KwRet{case (iii) and $|\cont(\Psi(q))|$}
    }
}
\caption{$\ClassifyCase$}
\label{algorithm:minimalwordhelper}
\end{algorithm}

\begin{theorem}\label{thm:classify_case_is_correct}
    \cref{algorithm:minimalwordhelper} is correct and has time complexity $\mathcal{O}(|A|)$.
\end{theorem}
\begin{proof}
Let $y = \Psi(q)$ be the element of $\FB(A)$ represented by $q$.
Since $q\clubsuit \alpha = y\ast\alpha$ for all $\alpha\in\bool^*$, $q\clubsuit 0 = q\clubsuit 1$ in line 1 if and only if $y\ast 0 = y\ast 1$. If this holds, then case (i) of \cref{lem:one_overlap} applies to $y$ and the algorithm returns in line 2.

Supposing that $q\clubsuit 0 \neq q\clubsuit 1$, we will show that $u = q\spadesuit 01^{k-1}$ and $v = q\spadesuit 10^{k-1}$ at the start of the loop in line 5 for every value of $k$.
In line 4, the states $u$ and $v$ represent $y\circ 0$ and $y\circ 1$, respectively, so the statement is true when $k=1$.
In line 9 we set
the values of $u$ and $v$ to be $u\spadesuit 1$ and $v\spadesuit 0$, respectively. So if $u=q \spadesuit 01^{k-1}$ and $v=q\spadesuit 10^{k-1}$ prior to line 9, then $u = q\spadesuit 01^k, v = q\spadesuit 10^k$ after line 9, as required.

 Since $u$ and $v$ represent $y\circ 01^{k-1}$ and $y\circ {10^{k-1}}$, respectively, by \cref{thm:subtransducers},  $u\spadesuit 1 = v\spadesuit 0$ holds if and only if $y\circ 01^k = y\circ 10^k$. Additionally, $u\clubsuit 1 = (y\circ 01^{k-1})\ast 1 = y\ast 01^{k}$ and $q\clubsuit 1 = y\ast 1$, and so $u\clubsuit1 = q\clubsuit 1$ if and only if $y\ast 01^k = y\ast 1$. Similarly, $v\clubsuit 0 = q\clubsuit 0$ if and only if $y\ast 10^k = y\ast 0$.
 So the condition in line 6 holds if and only if we are in case (ii) of \cref{lem:one_overlap} and $y\circ 01^k = y\circ 10^k$.

Finally for the correctness of the algorithm, if we terminate at line 11, it follows that the condition of case (ii) does not hold for any $k$, and so we must be in case (iii).

For the complexity, $|\cont(y)|$ can be computed once at the start of the algorithm in time $\mathcal{O}(|\cont(y)|)$ by simply traversing $\mathcal{T}_x$ starting from the state $q$ until we reach a terminal state.
The checks and operations done at each step of \cref{algorithm:minimalwordhelper} are constant time. The loop in line 5 will execute at most $|\cont(y)|\leq |A|$ times. So the time complexity of the algorithm is $\mathcal{O}(|A|)$ as required.
\end{proof}

We now present an algorithm $\MinWord$ for computing $\min(x)$ given a minimal transducer representing $x\in \FB(A)$. The input to $\MinWord$ consists of:
\begin{itemize}
\item a minimal transducer $\mathcal{T}_x = \{Q, \bool, A, q_0, T, \spadesuit, \clubsuit\}$ representing $x$;
\item a state $q\in Q$;
\item a word $w\in A^*$ that is a prefix of $\min(x)$;
\item a non-negative integer $l$;
\item a partial function $B: Q\rightarrow \mathbb{N}_0 \times \mathbb{N}_0$.
\end{itemize}
In \cref{thm:minimal_word_helper}, we will show that $\MinWord(\mathcal{T}_x, q, w, l, B)$ produces the minimum word corresponding to the element represented by the state $q$, assuming that the input satisfies certain constraints. It will follow that $\min(x) = \MinWord(\mathcal{T}_x, q, \varepsilon, 0, B)$, where $B$ is defined $B(q) = (0, 0)$ for all terminal states $q\in T$ and $B(q) = \bot$ otherwise.

The algorithm $\MinWord$ is stated formally in \cref{algorithm:minimalword}.
\begin{algorithm}[]
\SetAlgoLined\DontPrintSemicolon\KwArguments{$\mathcal{T}_x, q, w, l, B$ as described above.}
\KwReturns{a word $w^\prime$ and a partial function $B^\prime$.}
$s\gets |w|-l+1$\;
\If{$q\in \dom(B)$}{    $(i, j)\gets B(q)$\;
    \KwRet{$ww_{(i + l, j)}, B$}\;
}
$w, B\gets \MinWord(\mathcal{T}_x, q\spadesuit 0, w, l, B)$\;

$c, k := \ClassifyCase(\mathcal{T}_x, q)$\;

\uIf{$c$ indicates that \cref{lem:one_overlap}(i) applies}{    $w\gets w(q\clubsuit 0)$\;
    $l^\prime \gets 0$\;
}
\uElseIf{$c$ indicates that \cref{lem:one_overlap}(ii) applies}{    $l^\prime\gets |w_{B(q\spadesuit01^k)}|$
}
\Else{    $w\gets w(q\clubsuit 0)(q\clubsuit 1)$\;
    $l^\prime \gets 0$\;
}
$w, B \gets \MinWord(\mathcal{T}_x, q\spadesuit 1, w, l^\prime, B)$\;
$B(q) := (s, |w|)$\;

\KwRet{$w$, $B$}
\caption{$\MinWord$}
\label{algorithm:minimalword}
\end{algorithm}

Informally, the idea of $\MinWord$ is to build up $\min(x)$  from its constituent parts $\min(x)\ast0 = x\ast0$, $\min(x)\ast 1 = x\ast 1$, $\min(x)\circ 0$, and $\min(x)\circ 1$.
Since $\min(x)\circ 0 = \min(x\circ0)$ and similarly for $1$, $\min(x)\circ0$ and $\min(x)\circ 1$ can be computed recursively by considering the transducer states representing $x\circ0$ and $x\circ 1$, respectively. In this way, the recursion proceeds by traversing the states of $\mathcal{T}_x$. The parameter $q$ represents the current state in the recursion.

The parameter $w$ represents the prefix of $\min(x)$ found so far. In $\MinWord$, $w$ is extended to be a longer and longer prefix of $\min(x)$ until the two words coincide.

The first recursive call in $\MinWord$ will compute $\min(x\circ 0)$. How $\min(x\circ 0)$ is combined  with $x\ast 0$, $x\ast 1$ and $\min(x\circ 1)$ depends on the case of \cref{lem:one_overlap} that applies. If case (i) or (iii) applies, then there is no overlap in $\min(x\circ 0)$ and $\min(x\circ 1)$, and so having recursively computed them, $\min(x)$ is obtained by taking the appropriate concatenation of $\min(x\circ 0)$, $x\ast 0$, $x\ast 1$, and $\min(x\circ 1)$.
On the other hand, if case (ii) of \cref{lem:one_overlap} occurs, then there is non-trivial overlap between $\min(x\circ 0)$ and $\min(x\circ 1)$.
In this case, we find $\min(x\circ 0)$ initially and then compute $\min(x\circ 1)$ in a way which avoids recomputing the overlap.

This is where the fourth parameter $l$ comes into play. We set $l$ to be equal to the size of the overlap of \cref{lem:one_overlap} when computing $\MinWord$ of $x\circ 1$. The value of $l$ indicates the size of the prefix $\min(x\circ 1^k)$ of $\min(x\circ 1)$ that has already been calculated, and allows us to pick up the calculation from there.

 The function $B$ is used to avoid repeatedly recomputing recursive calls; in other words, $B$ implements the well-known memoization paradigm. In particular, the minimum word $\min(y)$ corresponding to the element $y$ of the free band represented by each state $q\in Q$ is only computed once, and, in some sense, the value is stored in $B$. If the state $q$ corresponding to $y$ has been visited by $\MinWord$ already, then  the word $w$ contains $\min(y)$. Hence rather than storing $\min(y)$ itself, it is only necessary to store the start and end index of $\min(y)$ in $w$. More precisely,
$B(q) = (i, j)$ if and only if $w_{(i, j)}= \min(y)$ where $y$ is the element of the free band corresponding to the state $q$.

In order to prove that $\MinWord$ is correct, we require the following notation.
A partial function $B:Q\rightarrow \mathbb{N}_0\times \mathbb{N}_0$  is \defn{closed under reachability} if $\dom(B)$ contains all terminal states in $Q$, and $\{r\spadesuit\alpha : \alpha \in \bool^\ast\} \subseteq \dom(B)$ for all $r \in \dom(B)$.
A pair $(w, B)$ where $w\in A ^ *$ and $B:Q\rightarrow \mathbb{N}_0\times \mathbb{N}_0$ is closed under reachability is  \defn{min-compatible}
$w_{B(r)} = \min(\Psi(r))$ for every $r\in \dom(B)$.

Suppose that $(w, B)$ is min-compatible and $w^\prime\in A^\ast$ has $w$ as a prefix. If $B(r) = (i, j)$ and $i\leq j$, since $w_{(i, j)} = \min(\Psi(r))$, then $j\leq |w|$ and $w_{(i, j)}\neq \varepsilon$. It follows that
$w^{\prime}_{(i, j)} = w_{(i, j)}$ since $w^\prime_{(1, |w|)} = w$. If $B(r) = (i, j)$ with $i>j$, then $w_{(i,j)}=\varepsilon = \min(\Psi(r))$ and so $w^\prime_{(i,j)}=\varepsilon = \min(\Psi(r))$ also. We have shown that $(w^\prime, B)$ is min-compatible.

The next lemma provides a crucial component in the proof of correctness of $\MinWord$.

\begin{lemma}\label{thm:minimal_word_helper}
    Suppose that $x\in \FB(A)$ and that $\mathcal{T}_x = (Q, \bool, A, q_0, T, \spadesuit, \clubsuit)$ is the minimal transducer representing $x$.
   If $w\in A^\ast$, $B:Q\rightarrow \mathbb{N}_0\times\mathbb{N}_0$, $q\in Q$, and $l \in \mathbb{N}_0$ are such that there exists a $K\in\mathbb{N}_0$ where:
    \begin{enumerate}[label=\normalfont{(\roman*)}]
        \item $(w, B)$ is min-compatible;
        \item $q\spadesuit 0^K\in\dom(B)$;
        \item $w_{(s, |w|)} = \min(\Psi(q\spadesuit 0^K))$ where $s = |w|-l+1$;
    \end{enumerate}
    then the pair $(w^\prime, B^\prime) = \MinWord(\mathcal{T}_x, q, w, l, B)$ is $\min$-compatible, $q\in \dom(B^\prime)$, $w$ is a prefix of $w^\prime$,
    and $w^{\prime}_{(s, |w^\prime|)} = \min(\Psi(q))$.
\end{lemma}
\begin{proof}
    We will proceed recursively on the states reachable from $q\in Q$.
    It suffices to show that the lemma holds for those states $q\in Q$ that have no child states, and that if the statement holds for $q\spadesuit 0$ and $q\spadesuit 1$, then the statement holds for $q$ also.

    For the assumptions of the lemma to hold, $q \spadesuit 0 ^ K$ has to be defined. Therefore, if $q$ has no child states, then $q\spadesuit 0^K$ is defined if and only if $K = 0$. But in this case, $q\spadesuit 0^0 =q \spadesuit \varepsilon = q\in \dom(B)$. Thus to establish this case it suffices to prove the lemma for all states $q\in\dom(B)$.

    Note that the $s$ calculated in line $1$ is exactly the $s$ in (iii) of the lemma.

    Suppose that $q\in \dom(B)$.
    Then the check in line 2 will pass. In line 3 we set $(i,j) = B(q)$ so that $w_{(i,j)} = \min(\Psi(q))$ since $(w, B)$ is min-compatible by assumption.
    In line 4 we return $w^\prime = ww_{(i+l, j)}$ and $B^\prime=B$.
    Since $q\in \dom(B) = \dom(B^\prime)$ and $w$ is trivially a prefix of $w^\prime$, in order to show the lemma holds, we need to ensure that $(w^\prime, B^\prime)$ is min-compatible and that $w^\prime_{(s, |w^\prime|)} = \min(\Psi(q))$.
    For the former, the input pair $(w, B)$ is min-compatible, $w$ is a prefix of $w^\prime$, and $B^\prime = B$, and so, by the comments before the lemma, $(w^\prime, B^\prime)$ is min-compatible also.
    For the latter, recall that $w_{(i, j)} = \min(\Psi(q))$ and $w_{(s, |w|)} = \min(\Psi(q\spadesuit 0^K)) = \min(\Psi(q))\circ 0^K$. Since $w_{(i,j)}\circ 0^K = \min(\Psi(q))\circ 0^K = w_{(s, |w|)}$ is a prefix of $w_{(i,j)}$ with length $|w| - s +1 = l$, it follows that $w_{(s, |w|)} = w_{(i, i+l-1)}$. This implies that $w_{(s, |w|)}w_{(i+l, j)} = w_{(i, i+l-1)}w_{(i+l, j)} = w_{(i,j)}= \min(\Psi(q))$. So the returned $w^\prime = ww_{(i+l, j)}$ has the suffix $w^\prime_{(s, |w^\prime|)} = \min(\Psi(q))$ as required.

    Next, we assume that the lemma holds for $q\spadesuit 0$ and $q\spadesuit 1$ and any $w\in A^\ast$, $l\in\mathbb{N}_0$, and $B:Q\rightarrow \mathbb{N}_0\times\mathbb{N}_0$ that satisfy the hypotheses of the lemma.

    We already covered the case when $q\in\dom(B)$, so we may assume that $q\not\in \dom(B)$. The check in line 2 will now fail and in line 6 we calculate
    $\MinWord(\mathcal{T}_x, q\spadesuit 0, w, l, B)$.
    We will show that the assumptions of the lemma hold for this invocation of $\MinWord$. The pair $(w, B)$ is min-compatible, by assumption, and so condition (i) holds. Since $q\not\in \dom(B)$, $K\geq 1$ and so $(q\spadesuit 0)\spadesuit 0^{K-1} =q\spadesuit 0^K\in\dom(B)$, and  condition (ii) holds. In the same vein, $w_{(s, |w|)} = \min(\Psi((q\spadesuit 0)\spadesuit 0^{K-1}))$ and condition (iii) holds. Therefore  if $(w^\prime, B^\prime) = \MinWord(\mathcal{T}_x, q\spadesuit 0, w, l, B)$, then $(w^ \prime, B^ \prime)$ satisfy the conclusion of the lemma, i.e.
    $(w^\prime, B^\prime)$ is min-compatible, $q\spadesuit 0\in \dom(B^\prime)$, the input $w$ is a prefix of $w^\prime$ and $w^\prime_{(s, |w^\prime|)} = \min(\Psi(q\spadesuit 0))$.

    In line 7 we use $\ClassifyCase$ to calculate $c$ and $k$.
    By \cref{thm:classify_case_is_correct} , $c$ represents the case of \cref{lem:one_overlap} that applies to $\Psi(q)$. Depending on the case, we will append some letters to $w$ and set the value of the parameter $l^\prime$ in  lines 8-16 and then calculate $\MinWord(\mathcal{T}_x, q\spadesuit 1, w, l^\prime, B)$ on line 17.

    We will show that after calculating $\MinWord(\mathcal{T}_x, q\spadesuit 1, w, l^\prime, B)$ in line 17, that the following hold: $(w, B)$ is min-compatible; $q\spadesuit 1\in \dom(B)$; the input $w$ is a prefix of the current $w$; and $w_{(s, |w|)} = \min(\Psi(q))$.
     Having established that this holds in line 17, after executing line 18, $q\in \dom(B)$ and $B$ is still reachability closed, since $q\spadesuit 0, q\spadesuit 1\in \dom(B)$. Also $(w, B)$ will be min-compatible since $B(q\spadesuit 0^k) = (s, |w|)$ and $w_{(s, |w|)} = \min(\Psi(q))$. Hence $w$ and $B$ returned in line 19 will satisfy the conclusions of the lemma, and the proof will be complete.

     Therefore it suffices to show that  after computing
$\MinWord(\mathcal{T}_x, q\spadesuit 1, w, l^\prime, B)$ in line 17, the following hold: $(w, B)$ is min-compatible; $q\spadesuit 1\in \dom(B)$; the input $w$ is a prefix of the current $w$; and $w_{(s, |w|)} = \min(\Psi(q))$.
    We will consider each of the cases as they apply to $\Psi(q)$ in \cref{lem:one_overlap} separately.

 \textbf{(i)}
 In this case,
        \[\min(\Psi(q)) = \min(\Psi(q)\circ 0)(\Psi(q)\ast 0)\min(\Psi(q)\circ1) =
        \min(\Psi(q\spadesuit0))(q\clubsuit0)\min(\Psi(q\spadesuit1))\]
        and lines 9 and 10 are executed. After line 9, $w_{(s, |w|)} = \min(\Psi(q\spadesuit0))(q\clubsuit0)$.

        Since $l^\prime = 0$, in line 17 we invoke $\MinWord(\mathcal{T}_x, q\spadesuit 1, w, 0, B)$. We now check that
        the assumptions of the lemma hold for this invocation.
         Since $(w, B)$ is min-compatible after line $6$ and we have only appended letters to $w$, $(w, B)$ is min-compatible before line 17, and so condition (i) holds.
         Since $\mathcal{T}_x$ is minimal, there exists $K^\prime$ such that  $(q\spadesuit 1)\spadesuit 0^{K^\prime}$ is terminal.
         Hence $\min(\Psi((q\spadesuit 1)\spadesuit 0^{K^\prime})) = \varepsilon$ and $(q\spadesuit 1)\spadesuit 0^{K^\prime}\in\dom(B)$ since $B$ is closed under reachability. This shows that condition (ii) holds. To see that condition (iii) holds, note that $l^\prime = 0$ (line 10) and so $w_{(|w|-l^\prime+1, |w|)} = w_{(|w|+1, |w|)} = \varepsilon = \min(\Psi((q\spadesuit 1)\spadesuit 0^{K^\prime}))$.

        We can recursively apply the lemma to conclude that if $(w^\prime, B^\prime) = \MinWord(\mathcal{T}_x, q\spadesuit 1, w, 0, B)$, then $(w^\prime, B^\prime)$ is min-compatible, $q\spadesuit 1\in \dom(B^\prime)$, $w$ is a prefix of $w^\prime$ and $w^\prime_{(|w|+1, |w^\prime|)} = \min(\Psi(q\spadesuit 1))$.

        It remains to show that $w^\prime_{(s, |w^\prime|)} =\min(\Psi(q))$.
        Since $w$ is a prefix of $w^\prime$, it follows that $w^{\prime}_{(1, |w|)} = w$ and so $w^\prime_{(s, |w|)} = w_{(s, |w|)} = \min(\Psi(q\spadesuit0))(q\clubsuit0)$. Thus
        \[w^\prime_{(s, |w^\prime|)} =w^\prime_{(s, |w|)}w^\prime_{(|w|+1, |w^\prime|)}= \min(\Psi(q\spadesuit0))(q\clubsuit0)\min(\Psi(q\spadesuit1)) = \min(\Psi(q)).\]

         \noindent \textbf{(ii)}
        Since \cref{lem:one_overlap}(ii) holds for $\Psi(q)$, $\min(\Psi(q))$ is the unique word that has prefix
        $\min(\Psi(q)\circ 0))(\Psi(q)\ast 0) = \min(\Psi(q\spadesuit 0))(q\clubsuit 0)$ and suffix $(\Psi(q)\ast 1)\min(\Psi(q)\circ 1))= (q\clubsuit 1)\min(\Psi(q\spadesuit 1))$ such that this prefix and suffix overlap exactly on the word $(q\clubsuit 1)\min(\Psi(q)\circ 01^k)(q\clubsuit 0)$, where $k$ is the same as the value we computed in line 7 (by the correctness of $\ClassifyCase$).

        Since we are in case (ii), in line 12 we set $l^\prime = |w_{B(q\spadesuit 01^k)}|$. Note that $q\spadesuit 01^k\in \dom(B)$ since $q\spadesuit 0\in\dom(B)$ after line 6, and $B$ is reachability-closed.

        In line 17 we invoke $\MinWord(\mathcal{T}_x, q\spadesuit 1, w, l^\prime, B)$. We now check that
        the assumptions of the lemma hold for this invocation.
        As in the proof when \cref{lem:one_overlap}(i) held, $(w, B)$ is min-compatible since it was after line $6$ and we have only appended letters to $w$ and left $B$ unchanged. In other words, condition (i) of the current lemma holds.

        It is part of the assumption of \cref{lem:one_overlap}(ii) that $\min(\Psi(q)\circ 1 0^{k}) = \min(\Psi(q)\circ 01^{k})$. Hence since elements of the free band, such as $\Psi(q)\spadesuit 1 0^{k}$ and $\Psi(q)\circ 01^{k}$,  are equal if and only if their minimum word representatives are equal, it follows that $\Psi(q)\circ 1 0^{k}= \Psi(q)\circ 01^{k}$. Since $\mathcal{T}_x$ is minimal,  $q\spadesuit 1 0^{k} = q\spadesuit 01^{k}$  and so $(q\spadesuit 1)\spadesuit 0^{k} = q\spadesuit 1 0^{k} = q\spadesuit 01^{k}\in \dom(B)$. This shows that condition (ii) of the hypothesis of the current lemma holds.

        Finally, we will show that condition (iii) of the current lemma holds for $\MinWord(\mathcal{T}_x, q\spadesuit 1, w, l^\prime, B)$.
        Since $\min(\Psi(q)\circ 01^k)$ is a suffix of $\min(\Psi(q)\circ 0)$ and, as we showed above, $w_{(s, |w|)} = \min(\Psi(q\spadesuit 0)) = \min(\Psi(q)\circ 0)$ after line 6, it follows that $w_{(|w|-|\min(\Psi(q)\circ 01^k)|+1, |w|)} = \min(\Psi(q)\circ 01^k)$ as the only length $|\min(\Psi(q)\circ 01^k)|$ suffix of $w$. Note that $l^\prime = |w_{B(q\spadesuit 01^k)}| = |\min(\Psi(q)\circ 01^k)|$ by min-compatibility. But then $w_{(|w|-l^\prime+1, |w|)} = \min(\Psi(q)\circ 01^k) = \min(\Psi(q)\circ 10^k) = \min(\Psi((q\spadesuit 1)\spadesuit 0^k))$ and condition (iii) holds.

        So all the assumptions of the lemma hold and so we can recursively apply it to conclude that if $(w^\prime, B^\prime) = \MinWord(\mathcal{T}_x, q\spadesuit 1, w, l^\prime, B)$ in line 17, then $(w^\prime, B^\prime)$ is min-compatible, $q\spadesuit 1\in \dom(B^\prime)$, $w$ is a prefix of $w^\prime$ and $w^\prime_{(s^\prime, |w^\prime|)} = \min(\Psi(q\spadesuit 1))$ where $s^\prime = |w|-l^\prime +1$.

        It remains to show that $w^\prime_{(s, |w^\prime|)} = \min(\Psi(q))$ where $s = |w| - l + 1$.
        Since $w$ is a prefix of $w^\prime$, it follows that $w^\prime_{(s, |w|)} = w_{(s, |w|)} = \min(\Psi(q\spadesuit 0))$. This implies that
        \[
        w^\prime_{(s^\prime,|w|)}
        = w_{(s^\prime,|w|)}
        = w_{(|w|-l^\prime +1, |w|)}
        = \min(\Psi(q)\circ 01^k) = \min(\Psi(q\spadesuit01^{k-1}))\circ 1
        \]
        is a suffix of $w^\prime_{(s, |w|)}$.
        So, by \cref{lem:one_overlap}(ii),
        \[
        w^\prime_{(s^\prime-1,s^\prime-1)} = \min(\Psi(q\spadesuit 01^{k-1}))\ast 1 =
        \Psi(q\spadesuit 01^{k-1})\ast 1=\Psi(q)\ast 01^k = \Psi(q)\ast 1 = q\clubsuit 1.
        \]
        Similarly, since $w^\prime_{(s^\prime, |w|)} = \min(\Psi(q\spadesuit 10^{k-1}))\circ 0$ is a prefix of $w^\prime_{(s^\prime, |w^\prime|)}=\min(\Psi(q\spadesuit 1))$, it follows that
        $w^\prime_{(|w|+1,|w|+1)} = \Psi(q)\ast 10^k= \Psi(q)\ast 0 = q\clubsuit 0$.

        Therefore $w^\prime_{(s, |w^\prime|)}$ is exactly the word with prefix $w^\prime_{(s, |w|+1)} = \min(\Psi(q\spadesuit 0))(q\clubsuit 0)$ and suffix $w^\prime_{(s^\prime -1, |w^\prime|)} = (q\clubsuit 1)\min(\Psi(q\spadesuit 1))$ whose common overlap is $w^\prime_{(s^\prime-1, |w|+1)} =(q\clubsuit 1)\min(\Psi(q)\circ 01^k))(q\clubsuit 0)$.
        Hence \cref{lem:one_overlap} implies that $w^\prime_{(s, |w^\prime|)} = \min(\Psi(q))$.

         \noindent \textbf{(iii)}
         The proof in the case that \cref{lem:one_overlap}(iii) holds is similar to the proof given above when  \cref{lem:one_overlap}(i) holds.
        In this case,
        \[\min(\Psi(q)) = \min(\Psi(q)\circ 0)(\Psi(q)\ast 0)(\Psi(q)\ast 1)\min(\Psi(q)\circ1 ) =
        \min(\Psi(q\spadesuit0))(q\clubsuit0)(q\clubsuit 1)\min(\Psi(q\spadesuit1)).\]
        After line 6, $w_{(s, |w|)} = \min(\Psi(q\spadesuit0))$ and because of the assumption of this case lines 14 and 15 are applied. After line 15, $w_{(s, |w|)} = \min(\Psi(q\spadesuit0))(q\clubsuit0)(q\clubsuit1)$. So as in case (i) it remains to show that after line 17 we have appended $\min(\Psi(q\spadesuit1))$ to $w$.

        It is possible to verify, using the same argument as in the proof of case (i), that $(w, B)$ is min-compatible and if $K^\prime$ is such that $(q\spadesuit 1)\spadesuit 0^{K^\prime}$ is terminal, then $(q\spadesuit 1)\spadesuit 0^{K^\prime}\in\dom(B)$ and $w_{(|w|-l^\prime+1,|w|)} = w_{(|w|+1, |w|)} = \varepsilon = \min(\Psi(q\spadesuit 0^{K^\prime}))$ (the latter equality holds since $l^\prime$ is defined to be $0$ in line 15). In other words, the conditions of the lemma apply in line 17.
        Therefore we can recursively apply the lemma to conclude that if $(w^\prime, B^\prime) = \MinWord(\mathcal{T}_x, q\spadesuit 1, w, 0, B)$, then $(w^\prime, B^\prime)$ is min-compatible, $q\spadesuit 1\in \dom(B^\prime)$, $w$ is a prefix of $w^\prime$ and $w^\prime_{(|w|+1, |w^\prime|)} = \min(\Psi(q\spadesuit 1))$.

        It remains to show that $w^\prime_{(s, |w^\prime|)} = \min(\Psi(q))$ where $s = |w| - l + 1$.
        Since $w$ is a prefix of $w^\prime$, it follows that
        $w ^ \prime _{(s, |w|)} = w_{(s, |w|)}  = \min(\Psi(q\spadesuit0))(q\clubsuit0)(q\clubsuit1)$. Therefore
        \[w^\prime_{(s, |w^\prime|)} =w^\prime_{(s, |w|)}w^\prime_{(|w|+1, |w^\prime|)}= \min(\Psi(q\spadesuit0))(q\clubsuit0)(q\clubsuit1)\min(\Psi(q\spadesuit1)) = \min(\Psi(q)).\qedhere\]
\end{proof}

For the purposes of the following theorem, we assume the RAM model of computation. This is quite a modest assumption in the sense that the time complexity is realized in the implementation of $\MinWord$ in \cite{github_repo}; see \cref{section-benchmarks} for more details.

\begin{corollary}
\label{cor:minword_is_correct}
    Let $x\in \FB(A)$, let $\mathcal{T}_x = (Q, \bool, A, q, T, \spadesuit, \clubsuit)$ be the minimal transducer representing $x$, and let $B: Q\rightarrow \mathbb{N}_0\times\mathbb{N}_0$ be such that $B(q) = (0, 0)$
    for all $q\in T$ and $B(q) = \bot$ on all other inputs.
    Then
    the first component of $\MinWord(\mathcal{T}_x, q_0, \varepsilon, 0, B)$ is equal to $\min(x)$ and
    $\MinWord$ has time complexity $\mathcal{O}(|A|\cdot|Q| +|\min(x)|) = \mathcal{O}(|A|^2\cdot |\min(x)|)$ and space complexity $\mathcal{O}(|Q|+|\min(x)|) = \mathcal{O}(|A|\cdot|\min(x)|)$.
\end{corollary}
\begin{proof}
    We show that the conditions of \cref{thm:minimal_word_helper} applies to $\MinWord(\mathcal{T}_x, q_0, \varepsilon, 0, B)$ with the parameters as stated.
    The initial partial function $B$ is defined only on the terminal states $T$ of the minimal transducer $\mathcal{T}_x$. Hence $B$ is reachability closed since terminal states have no child states.
    By assumption,  $B(q) = (0,0)$ for all $q\in T$ and $w = \varepsilon$,
    and so $w_{B(q)} = \varepsilon$. Since $q$ is terminal, $q \spadesuit \alpha$ and $q\clubsuit \alpha$ are undefined for all $\alpha\in\bool ^ *$. In particular, $\min(\Psi(q)) = \varepsilon = w_{B(q)}$ for all $q\in T$ and so $(w, B)$ is min-compatible.
    This establishes that \cref{thm:minimal_word_helper}(i) holds.
    Since $\Psi(q_0) = x$ and $|\cont(x\circ 0^{|\cont(x)|})| = |\cont(x)| - |\cont(x)| = 0$, $q_0\spadesuit 0^{|\cont(x)|}$ is terminal and so $q_0\spadesuit 0^{|\cont(x)|}\in \dom(B)$. Therefore \cref{thm:minimal_word_helper}(ii) holds with $K = |\cont(x)|$.
    Since $l = 0$ and $|w|=0$,  $s = |w|-l+1 = 1$  and so $w_{(s, |w|)} = w_{(1, 0)} = \varepsilon = \min(\Psi(q\spadesuit 0^{|\cont(x)|}))$. This shows that
     \cref{thm:minimal_word_helper}(iii) holds.
    Hence the hypothesis of \cref{thm:minimal_word_helper} is satisfied. Therefore if $(w^\prime, B^\prime) = \MinWord(\mathcal{T}_x, q_0, \varepsilon, 0, B)$, then $w^\prime_{(s, |w^\prime|)} = w^\prime_{(1, |w^\prime|)} = w^\prime = \min(x)$, as required.

    For the time and space complexities, we start by elaborating some of the assumptions in the model of computation.  In particular, in the RAM model, we may assume that the following operations are constant time for all $q\in Q$ and $\alpha\in\bool$: access a letter in $w$, retrieve a value of $B(q)$, $q\spadesuit \alpha$, or $q\clubsuit \alpha$, append a letter to $w$, and define a value $B(q)$. We may also
    assign $w$ and/or $B$ in constant time in lines 6, 9, 14, and 17 by modifying $w$ and $B$ in-place. Changing $w$ or $B$ in-place in a recursive call, will modify $w$ and $B$ in the original call, but this does not cause any issues since upon reentering the original function, in line 6 or 17, we immediately assign the returned $w$ and $B$ to the current $w$ and $B$.

    We denote by $t_{\text{total}}$ the total number of steps taken by $\MinWord(\mathcal{T}_x, q_0, \varepsilon, 0, B)$
   with the inputs given in the statement of the corollary. We split $t_{\text{total}} = t_{\text{word}} + t_{\text{main}}$, where $t_{\text{word}}$ is the total number of steps taken to write letters of $w$ and $t_{\text{main}}$ accounts for all of the other steps.
    Since $w$ is only modified by appending letters, and appending a letter is constant time, it follows that $t_{\text{word}} \in \mathcal{O}(|\min(x)|)$.

    In the remainder of the proof, we will exclude the time taken to append letters to $w$, since it is already accounted for in $t_{\text{word}}$.
    For $t_{\text{main}}$, note the role played by $B$. If a particular state $q$ belongs to $\dom(B)$, then the check in line 2 passes, and line 3 is $\mathcal{O}(1)$ and we return in line 4.
    If $q\not\in \dom(B)$, then, irrespective of which case of \cref{lem:one_overlap} applies, $B(q)$ is defined in line 18 and so $q\in \dom(B)$ upon completion of the call to  $\MinWord$ when $q$ is first considered.

    We denote by $t_{\text{body}}$ the number of steps taken in $\MinWord$ excluding steps taken by the recursive calls in lines 6 and 17 (and excluding the time taken to append to $w$).
    For every state $q\in Q$ and every invocation of $\MinWord$ with input $q$, either $q\in \dom(B)$ at a cost of $\mathcal{O}(1)$ or $q\not\in \dom(B)$ incurring $t_{\text{body}}$. The latter occurs precisely once as described in the previous paragraph. Hence $t_{\text{main}} \in\mathcal{O}(|Q|\cdot t_{\text{body}})$.

    To determine $t_{\text{body}}$, note that apart from lines 7 and 12, the remaining steps are constant time. Line 7 takes $\mathcal{O}(|A|)$ steps as per
    \cref{thm:classify_case_is_correct}. In line 12, we calculate $|w_{B(q\spadesuit 01^k)}|$. By assumption  $q\spadesuit 01 ^ k$ requires $\mathcal{O}(k + 1) = \mathcal{O}(|A|)$ steps since $k\leq |\cont(x)| \leq |A|$. The subsequent lookup of $B(q\spadesuit 01^k)$ is constant time.
    The calculation of $|w_{(i, j)}|$ for any $i, j$ is constant time, since if $0<i\leq j$, then $|w_{(i, j)}| = j-i+1$ and otherwise $|w_{(i,j)}| = 0$.
     So the total for the calculation in line 12 is $\mathcal{O}(|A|)$ steps.

     Therefore $t_{\text{main}} \in \mathcal{O}(|Q|\cdot t_{\text{body}}) = \mathcal{O}(|Q|\cdot |A|)$, and so
     the total time complexity is $\mathcal{O}(|\min(x)|+|Q|\cdot |A|)$.
     Since $\mathcal{T}_x$ is
     minimal, by \cref{thm:size_min_transducer}, $\mathcal{O}(|Q|) = \mathcal{O}(|\mathcal{T}_x|) = \mathcal{O}(|A|\cdot |\min(x)|)$. It follows that $\mathcal{O}(|\min(x)|+|Q|\cdot |A|) = \mathcal{O}(|\min(x)|+|\min(x)|\cdot |A| ^ 2) = \mathcal{O}(|A|^2\cdot |\min(x)|)$.

     For the space complexity, storing $\mathcal{T}_x$ already requires $\mathcal{O}(|\min(x)|\cdot |A|)$ space. Since $w$ is a prefix of $\min(x)$, the space taken by $w$ is $\mathcal{O}(|\min(x)|)$. The partial function $B$ has at most $|Q|$ values where it is defined, and each $B(q)$ is a pair of integers. Hence $B$ requires  $\mathcal{O}(|Q|) = \mathcal{O}(|\min(x)|\cdot |A|)$ space. Thus in total  $\mathcal{O}(|\min(x)|\cdot |A|)$ space is required.
     \end{proof}

    \section*{Acknowledgements}
    The authors would like to thank the anonymous referee for their helpful comments, and for highlighting the several important articles in the literature related to the content of this paper. The authors would also like to thank Carl-Fredrik Nyberg-Brodda for pointing out some inaccuracies in the introduction, and for his helpful comments on the historical aspects of the Burnside problems.
    This work was supported by the UK Engineering and Physical Sciences Research Council (EPSRC) doctoral training grant EP/V520123/1 for the University of St Andrews, which funded the first author's research on this paper.

\printbibliography[heading=bibintoc]

\appendix
\section{Benchmarks}\label{section-benchmarks}

The Python package \cite{github_repo} implements the algorithms $\EqualInFreeBand$, $\Multiply$, and $\MinWord$.
In this section we describe some benchmarks for these algorithms.
The purpose of \cite{github_repo} is to provide a reference implementation, and
the benchmarks are included to demonstrate that even a relatively na\"ive implementation in a high-level language, such as Python, exhibits the asymptotic time complexity described in this paper. The implementation in \cite{github_repo} is not optimised, and the absolute times provided in this section can almost certainly be improved significantly.

All the benchmarks in this section were run on a 2021 MacBook Pro with an Apple M1 processor and 16GB of RAM running Python 3.9.12. Instructions for how to reproduce the benchmarks in this section can be found in~\cite{github_repo}.

In order to benchmark our code, we generated the following set of sample data: for each alphabet of size $m\in \{2, 7, 12, \ldots, 47\}$ and each length $l\in\{20, 520, 1020, \ldots, 4520\}$ we generated a sample of 100 words of length $l$ (selected uniformly at random) over an alphabet of size $m$ using the Python \texttt{random} library for randomness. We will refer to these as the \defn{word samples}.
The interval transducer for each word in each sample was computed using  $\IntervalTransducer$; we refer to this as the \defn{interval transducer sample}. Finally, each of the interval transducer samples was run through the $\Minimize$ algorithm, forming the \defn{minimal transducer sample}.
The benchmarks were produced using the \texttt{pytest-benchmark} Python library. Any quoted time is the mean of a number of runs of the benchmark using the same input data. This makes the results more reproducible and less susceptible to random variation in the device used to run them.

The algorithm $\EqualInFreeBand$ is the composition of the 3 algorithms defined in \cref{sec:equality}: $\IntervalTransducer$, $\Minimize$, and $\TrimTransducerIsomorphism$.

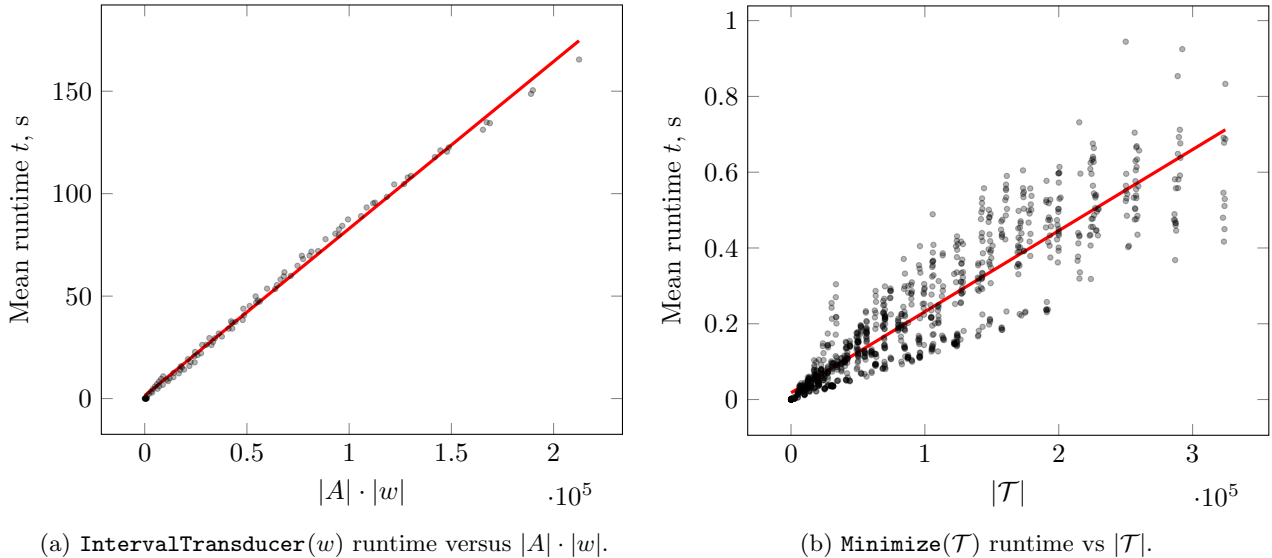
\begin{figure}[h]
    \begin{subfigure}{0.5\textwidth}
\begin{tikzpicture}
\begin{axis}[xlabel={$|A|\cdot |w|$}, ylabel={Mean runtime $t$, s}]
\addplot[only marks, mark=*,mark size=1pt, black, opacity = 0.3] file {./anc/interval_benchmarks.dat};
\addplot[domain=40:212440, samples=100, smooth, red, very thick] {0.0008153*x+1.419};
\end{axis}
\end{tikzpicture}
        \caption{$\IntervalTransducer(w)$ runtime versus $|A|\cdot |w|$.}
        \label{fig:benchmarks-interval}
    \end{subfigure}
    \begin{subfigure}{0.5\textwidth}
\begin{tikzpicture}
\begin{axis}[xlabel={$|\mathcal{T}|$}, ylabel={Mean runtime $t$, s}]
\addplot[only marks, mark = *, mark size=1pt, black, opacity=0.3] file {./anc/minimize_benchmarks.dat};
\addplot[domain=82:324465, samples=100, smooth, red, very thick] {0.00000213980667*x+0.0177936788};
\end{axis}
\end{tikzpicture}
        \caption{$\Minimize(\mathcal{T})$ runtime vs $|\mathcal{T}|$.}
        \label{fig:benchmarks-minimize}
    \end{subfigure}
    \caption{Benchmarks for the $\IntervalTransducer$ and $\Minimize$ algorithms.}
    \label{fig:benchmarks-iterval-and-minimize}
\end{figure}

\cref{fig:benchmarks-interval} shows run times of $\IntervalTransducer(w)$ for the words $w$ in the word samples described above. Each data point in \cref{fig:benchmarks-interval} shows the cumulative run time for a sample consisting of 100 words. This is plotted against $|A|\cdot |w|$ which is the theoretical asymptotic complexity of the algorithm. The red line is the line of best fit obtained by running a linear regression in Python.

\cref{fig:benchmarks-minimize} shows the run times of $\Minimize(\mathcal{T})$ for transducers $\mathcal{T}$ in the interval transducer samples. This is plotted against the number of states of $|\mathcal{T}|$ which is the theoretical asymptotic complexity of the $\Minimize$ algorithm. Each point represents a single transducer $\mathcal{T}$. The line of best fit is shown in red. The fit of the line in \cref{fig:benchmarks-minimize} might not seem very good, with a fairly wide range of times
for the given transducers of  each number of states. However the linear regression has an $R^2$ value of $0.8368$, which is perhaps not as bad as it first appeared.

\begin{figure}[h]
    \begin{subfigure}{0.5\textwidth}
\begin{tikzpicture}
\begin{axis}[xlabel={$|\mathcal{T}|$}, ylabel={Mean runtime $t$, s}]
\addplot[only marks, mark = *, mark size=1pt, black, opacity=0.3] file {./anc/isomorphism_benchmarks.dat};
\addplot[domain=5:12591, samples=100, smooth, red, very thick] {0.00000360912595*x-0.000576196934};
\end{axis}
\end{tikzpicture}
        \caption{$\TrimTransducerIsomorphism(\mathcal{T}, \mathcal{T})$ runtime vs $|\mathcal{T}|$.}
        \label{fig:benchmarks-isomorphism}
    \end{subfigure}
    \begin{subfigure}{0.5\textwidth}
\begin{tikzpicture}
\begin{axis}[xlabel={$(|u|+|v|)\cdot |A|$}, ylabel={Mean runtime $t$, s}]
\addplot[only marks, mark = *, mark size=1pt, black, opacity=0.3] file {./anc/equal_benchmarks.dat};
\addplot[domain=80:424880, samples=100, smooth, red, very thick] {0.0000106244955*x +0.0297170074};
\end{axis}
\end{tikzpicture}
        \caption{$\EqualInFreeBand(u, v)$ runtime vs $(|u|+|v|)\cdot |A|$.}
        \label{fig:benchmarks-equal}
    \end{subfigure}
    \caption{Benchmarks for the $\TrimTransducerIsomorphism$ and $\EqualInFreeBand$ algorithms.}
    \label{fig:benchmarks-isomorphism-and-equal}
\end{figure}
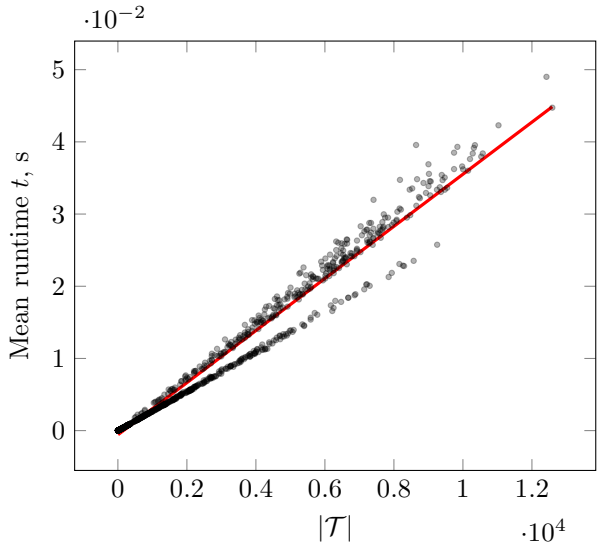
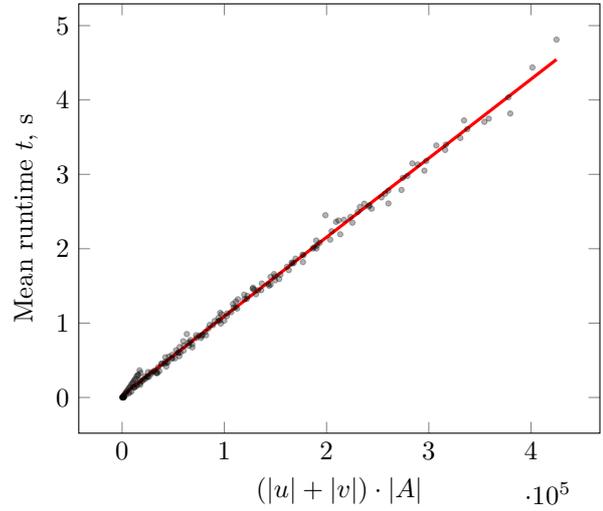

\cref{fig:benchmarks-isomorphism} shows run times of
$\TrimTransducerIsomorphism(\mathcal{T}, \mathcal{T})$ for a transducer
$\mathcal{T}$ in the interval transducer samples. This is plotted against $|\mathcal{T}|$ which is the theoretical asymptotic complexity of the algorithm. The line of best fit is shown in red. We chose to run the isomorphism of a transducer against itself since this represents the longest execution path in the code.

\cref{fig:benchmarks-equal} shows the run times of $\EqualInFreeBand(u, v)$ for
words $u$ and $v$ obtained by randomly sampling $10000$ pairs from the word
samples. This is plotted against $(|u|+|v|)\cdot |A|$ which is the theoretical
asymptotic complexity of the algorithm. The line of best fit is shown in red.
We chose to only sample 10000 pairs to reduce the total run time of the
benchmark.

\begin{figure}[H]
    \begin{subfigure}{0.5\textwidth}
\begin{tikzpicture}
\begin{axis}[xlabel={$|\mathcal{T}_1|+|\mathcal{T}_2| + |A|^2$}, ylabel={Mean runtime $t$, s}]
\addplot[domain=16:20315, samples=100, smooth, red, very thick] {0.00000179895551*x -0.00172193307};
\addplot[only marks, mark = *, mark size=1pt, black, opacity=0.05] file {./anc/minimal_multiply_benchmarks.dat};
\end{axis}
\end{tikzpicture}
        \caption{$\Multiply(\mathcal{T}_1, \mathcal{T}_2)$ runtime vs $|\mathcal{T}_1| + |\mathcal{T}_2| + |A|^2$. The transducers $\mathcal{T}_1$ and $\mathcal{T}_2$ sampled from the minimal transducer sample.}
        \label{fig:benchmarks-multiply-minimal}
    \end{subfigure}
    \begin{subfigure}{0.5\textwidth}
\begin{tikzpicture}
\begin{axis}[xlabel={$|\mathcal{T}_1|+|\mathcal{T}_2| + |A|^2$}, ylabel={Mean runtime $t$, s}]
\addplot[only marks, mark = *, mark size=1pt, black, opacity=0.3] file {./anc/interval_multiply_benchmarks.dat};
\addplot[domain=25021:404143, samples=100, smooth, red, very thick] {0.00000262090654*x+0.0263825366};
\end{axis}
\end{tikzpicture}
        \caption{$\Multiply(\mathcal{T}_1,\mathcal{T}_2)$ runtime vs $|\mathcal{T}_1| + |\mathcal{T}_2| + |A|^2$.
        The transducers $\mathcal{T}_1,\mathcal{T}_2$ sampled from the interval transducer sample.}
        \label{fig:benchmarks-multiply-interval}
    \end{subfigure}
    \caption{Benchmarks for the $\Multiply$ algorithm when applied to minimal and interval transducers respectively.}
    \label{fig:benchmarks-multiply-minimal-and-multiply-interval}
\end{figure}
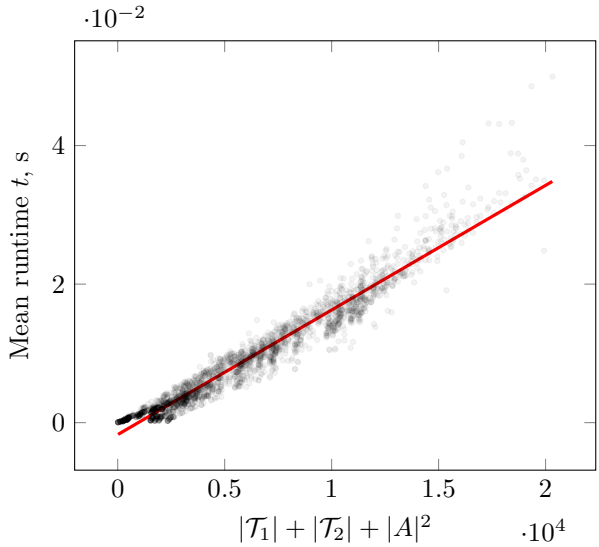
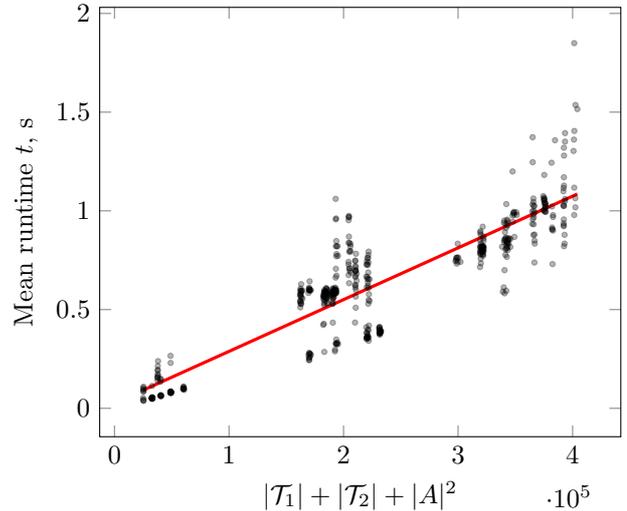

\cref{fig:benchmarks-multiply-minimal} shows the runtime of $\Multiply(\mathcal{T}_1, \mathcal{T}_2)$ for transducers $\mathcal{T}_1$ and $\mathcal{T}_2$ obtained by randomly sampling 50 transducers from the minimal transducer samples and considering all possible pairs of these transducers (2500 pairs total). This is plotted against $|\mathcal{T}_1|+|\mathcal{T}_2| + |A|^2$ which is the theoretical asymptotic complexity of the algorithm. The line of best fit is shown in red. We chose to only sample 2500 pairs to reduce the total runtime of the benchmark.

\cref{fig:benchmarks-multiply-interval} shows the runtime of $\Multiply(\mathcal{T}_1, \mathcal{T}_2)$ for transducers $\mathcal{T}_1$ and $\mathcal{T}_2$ obtained by batching the interval transducer samples into 20 batches, randomly sampling 5 transducers from each batch and considering all possible pairs of these transducers in each batch (500 pairs total). This is plotted against $|\mathcal{T}_1|+|\mathcal{T}_2| + |A|^2$ which is the theoretical asymptotic complexity of the algorithm. The line of best fit is shown in red. We chose to only sample 500 pairs to reduce the total runtime of the benchmark.

\begin{figure}[h]
\centering
\begin{tikzpicture}
\begin{axis}[xlabel={$|\mathcal{T}|\cdot |A|$}, ylabel={Mean runtime $t$, s}, title={}]
\addplot[only marks, mark = *, mark size = 1pt, black, opacity=0.1] file {./anc/minword_benchmarks.dat};
\addplot[domain=10:591777, samples=100, smooth, red, very thick] {0.000000118569040*x+0.000799812408};
\end{axis}
\end{tikzpicture}
    \caption{$\MinWord(\mathcal{T})$ runtime versus $|\mathcal{T}|\cdot |A|$.}
    \label{fig:benchmarks-minword}
\end{figure}
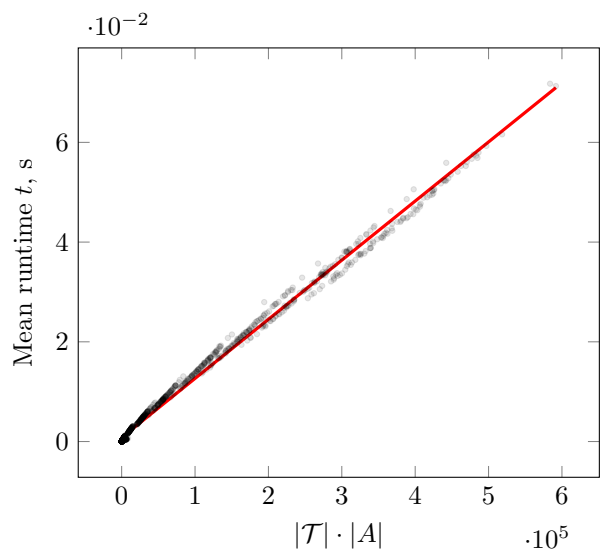

\cref{fig:benchmarks-minword} shows the runtime of $\MinWord(\mathcal{T})$ for transducers $\mathcal{T}$ in the minimal word samples. This is plotted against $|\mathcal{T}|\cdot |A|$ which is the theoretical asymptotic complexity of the algorithm. The line of best fit is shown in red.
\end{document}